\documentclass[a4paper,12pt]{article}

\usepackage{amsmath}
\usepackage{amsfonts}

\usepackage[margin=1.2in]{geometry} 

\usepackage{rotating}
\usepackage{natbib}
\usepackage{bm}
\usepackage{bbm}
\usepackage{amsthm}
\usepackage{adjustbox}
\usepackage[notablist,nofiglist,]{endfloat}
\usepackage{longtable}
\usepackage{endfloat}
\usepackage{rotating}
\usepackage{pdflscape}
\usepackage{bm}

\DeclareMathAlphabet{\pazocal}{OMS}{zplm}{m}{n}
\SetMathAlphabet\pazocal{bold}{OMS}{zplm}{bx}{n}
\DeclareDelayedFloatFlavor{sidewaystable}{table}
\DeclareDelayedFloatFlavour*{longtable}{table}
\DeclareDelayedFloatFlavour*{longsidewaystable}{table}
\usepackage{comment} 
\DeclareMathAlphabet{\mathdutchcal}{U}{dutchcal}{m}{n}
\usepackage{enumitem} 

\usepackage{multirow}
\usepackage{graphicx}
\usepackage[dvipsnames]{xcolor}
\usepackage{nameref,zref-xr}
\usepackage{authblk}

\usepackage{color}
\usepackage{booktabs}
\usepackage{pdflscape}
\usepackage{afterpage}

\DeclareMathOperator{\Var}{Var}

\newcommand{\mb}{\mathbf}
\newcommand{\mc}{\mathcal}
\newcommand{\bsb}{\boldsymbol}

\newcommand{\op}[1]{o_{P}({#1})}

\newcommand{\bbR}{\mathbbm{R}}

\newcommand{\Exp}{\mathbb{E}}
\newcommand{\Prob}{\mathbb{P}}

\newcommand{\argmax}{\operatornamewithlimits{\arg\max}}

\newcommand{\mathph}[1]{\mathrel{\phantom{#1}}}

\newcommand{\pto}{\xrightarrow{P}}

\newcounter{counter_m}

\newcommand{\balpha}{\bm{\alpha}}

\newcommand{\bb}{\bm{b}}

\newcommand{\bc}{\bm{c}}

\newcommand{\bbbeta}{\bm{\beta}}

\newcommand{\bphi}{\bm{\phi}}

\newcommand{\bgamma}{\bm{\gamma}}

\newcommand{\bcK}{\bm{\mathcal{K}}}

\newcommand{\bSigma}{\bm{\Sigma}}

\newcommand{\btheta}{\bm{\theta}}

\newcommand{\bV}{\bm{V}}

\newcommand{\bx}{\mathbf{x}}

\newcommand{\bZ}{\bm{Z}}
\newcommand{\bz}{\mathbf{z}}
\newcommand{\tbz}{\tilde{\mathbf{z}}}

\newcommand{\cD}{\mathcal{D}}
\newcommand{\cd}{\mathcal{d}}

\newcommand{\cK}{\mathcal{K}}

\newcommand{\cN}{\mathcal{N}}
\newcommand{\cn}{\mathcal{n}}

\newcommand{\cS}{\mathcal{S}}

\newcommand{\bkappa}{\bm{\kappa}}
\newcommand{\hbkappa}{\hat{\bm{\kappa}}}

\newcommand{\bzeta}{\bm{\zeta}}
\newcommand{\hbzeta}{\hat{\bm{\zeta}}}

\newcommand{\hM}{\hat{M}}

\numberwithin{equation}{section}
\usepackage{cleveref}
\newtheorem{Assu}{Assumption}[section]
\Crefname{Assu}{Assumption}{Assumptions}

\Crefname{Cond}{Condition}{Conditions}

\Crefname{Exam}{Example}{Examples}
\newtheorem{Theo}{Theorem}[section]
\Crefname{Theo}{Theorem}{Theorems}
\newtheorem{Prop}{Proposition}[section]
\Crefname{Prop}{Proposition}{Propositions}
\newtheorem{Lemm}{Lemma}[section]
\Crefname{Lemm}{Lemma}{Lemmas}

\Crefname{Coro}{Corollary}{Corollaries}

\Crefname{Conj}{Conjecture}{Conjectures}

\Crefname{Algo}{Algorithm}{Algorithms}
\newtheorem{Defi}{Definition}[section]
\Crefname{Defi}{Definition}{Definitions}

\Crefname{Hypo}{Hypothesis}{Hypotheses}
\theoremstyle{remark}

\Crefname{Ques}{Question}{Questions}
\newtheorem{Rema}{Remark}[section]
\Crefname{Rema}{Remark}{Remarks}

\newcommand{\bbb}{\bm{b}}

\newcommand{\ep}{\varepsilon}

\begin{document}

\title{On a penalised likelihood approach for joint modelling of longitudinal covariates and partly interval-censored 
data -- 
an application to the Anti-PD1 brain collaboration trial }
\author[1,2]{Annabel Webb \footnote{To whom correspondence should be addressed: annabel.webb@mq.edu.au}}
\author[1]{Nan Zou}
\author[3,4]{Serigne N Lo}
\author[1]{Jun Ma}
\affil[1]{School of Mathematical and Physical Sciences, Macquarie University, Australia}
\affil[2]{Cerebral Palsy Alliance Research Institute, Faculty of Medicine and Health, The University of Sydney, Australia}
\affil[3]{Melanoma Institute Australia, Sydney, Australia}
\affil[4]{Faculty of Medicine and Health, The University of Sydney, Australia}
\date{}

\maketitle
\thispagestyle{empty}

\begin{abstract}
This article considers the joint modeling of longitudinal covariates and partly-interval censored time-to-event data. Longitudinal time-varying covariates play a crucial role in obtaining accurate clinically relevant predictions using a survival regression model. However, these covariates are often measured at limited time points and may be subject to measurement error. Further methodological challenges arise from the fact that, in many clinical studies, the event times of interest are interval-censored. A model that simultaneously  accounts for all these factors is expected to improve the accuracy of survival model estimations and predictions. In this article, we consider joint models that combine longitudinal time-varying covariates with the Cox model for time-to-event data which is subject to interval censoring. The proposed model employs a novel penalised likelihood approach for estimating all parameters, including the random effects. The covariance matrix of the estimated parameters can be obtained from the penalised log-likelihood. The performance of the model is compared to an existing method under various scenarios. The simulation results demonstrated that our new method can provide reliable inferences when dealing with interval-censored data. Data from the Anti-PD1 brain collaboration clinical trial in advanced melanoma is used to illustrate the application of the new method.
\end{abstract}
\thispagestyle{empty}

\newpage
\setcounter{page}{1}
\section{Introduction}

The simultaneous observation of longitudinal covariates and time-to-event data has been increasingly common in medical research. For example, new immunotherapies have revolutionized cancer treatment in breast and melanoma cancers. However, not every patient benefits from these new drugs. Additionally, they are expensive and can cause significant toxicity. Therefore, it is crucial for advanced cancer patients receiving immunotherapy to be closely monitored to assess their response or resistance to treatment. During these follow-up visits, radiologic imaging and blood samples are collected and considered crucial for patient management. For instance, in stage IV melanoma patients, certain baseline tumor markers and blood serum levels have been identified as independent factors associated with treatment outcomes. Throughout treatment, physicians intuitively adjust their prognosis predictions at each clinical visit based on changes in clinical status, tumor markers, and blood serum levels \citep{Gide_2019_abbrv, Owen_2020_abbrv}. Fortunately, novel statistical methods exist that combine longitudinal biomarker data and time-to-event outcomes to dynamically estimate individual patients' risk during treatment. These models, known as joint models, have received significant attention in the statistical literature (e.g., \cite{tsiatis_2004_JM_Overview,Rizo_2012}). When implemented in clinical settings, such models can guide treatment decisions and facilitate the early identification of patients who are resistant to existing treatments, thereby determining their eligibility for evaluating new agents.

A typical approach to joint modelling involves defining a joint likelihood for the simultaneous estimation of parameters from a time-to-event model and a model for the longitudinal observations, conditional on a random effects distribution. The random effects are used in the longitudinal model to capture individual variations in the time-varying covariate. To make statistical inferences on this type of model, a common approach is maximum likelihood estimation (MLE), usually via an expectation-maximisation (EM) algorithm e.g., \cite{Rizo_2012, song2002semiparametric, zeng2005asymptotic}. Bayesian approaches have also been considered (\cite{brown2003bayesian, rizo_2011_BayesianSP_JM} among others), as have other approaches such as h-likelihood \citep{ha_2017_hlikJM} and exact likelihood \citep{barrett_2015_JM_exactLikelihood}.

In the context of monitoring long-term diseases such as advanced melanoma, where patients are followed-up at intermittent time points, it is natural that the event of interest may be subject to partly-interval censoring. 
 Partly interval-censored datasets \citep{Kim03} 
refer to datasets where some of the event times may be known exactly, while other event times are known only to lie within a certain time interval. Where the event is known to fall between two observation points for a subject, this subject is considered interval censored \citep{sun2006statistical}. Partly-interval censored datasets may also include right-censored subjects, where the event has not occurred up to the last observation time, or left-censored subjects, where the event occurred before the start of the study. To our knowledge, in the joint modelling literature only a small number of studies have considered interval-censored event times, including \cite{chen_2018_TJM_IC, yi_2020_IC_TVC, wu_2020_multivariate_IC_TVC}.

Previously, penalised likelihood methods have been proposed for fitting a Cox model to partly-interval censored data \citep{MaCoHeIa21}. This penalised likelihood approach incorporates spline approximation of the baseline hazard function, with a roughness penalty in the likelihood function to regularise this approximation, and the non-negativity of the baseline hazard guaranteed by using a Newton multiplicative-iterative (MI) algorithm \citep{ChMa12}. This approach has been applied in the context of time-varying covariates \citep{WeMa23}, but not in joint modelling. In this article, we develop a maximum penalised likelihood estimation procedure for fitting joint models to partly-interval censored survival data, extending the approach in \cite{MaCoHeIa21}. Model parameters are estimated in two stages. Firstly, non-variance component parameters, including regression coefficients and the random effects, are estimated using a full joint likelihood. Then, 
secondly, all the variance components, including the smoothing parameter for the penalty term, are estimated using an approximated marginal likelihood \citep{Wood_REML_2011}.

This paper is organised as follows. Section \ref{sec:model} explains the joint model under partly-interval censoring and introduces the notation used throughout this paper. Section \ref{sec:MPLesti} develops the maximum penalised likelihood method for estimating the model parameters, including the estimation of the baseline hazard, which requires a non-negativity constraint. Section \ref{sec:asymp} provides a large sample distribution results and gives details of the variance estimation.
Section \ref{sec:simu} describes a simulation study for evaluation of our proposed method and reports 
the simulation results. As an illustration, the proposed method is applied to real data from a clinical trial investigating the effectiveness of immunotherapy in advanced melanoma patients, as detailed in Section \ref{sec:app}, and Section \ref{sec:concl} presents some concluding remarks and future directions for this research.

\section{Model setup  
} \label{sec:model}
 
We consider a sample of $n$ individuals, who are followed intermittently until the event-of-interest is observed. Let $Y_i$ denote the event time of interest for individual $i$, $i=1,\ldots, n$. Due to censoring, $Y_i$ may not be observed exactly. In this article, we assume $Y_i$ is partly interval-censored, meaning it can be exactly observed or left-, right- or interval-censored \citep{Kim03}. 
To accommodate interval censoring, we denote, for individual $i$, a pair of random time points $T_i^L$ and $T_i^R$ such that $Y_i \in [T_i^L, T_i^R]$. We denote a pair of observed $T_i^L, T_i^R$ by $t_i^L, t_i^R$. When $t_i^L=t_i^R=t_i$ then $Y_i=t_i$ so that $t_i$ is the observed event time. If $t_i^L=0$ and $t_i^R$ is finite then $Y_i$ is left-censored at $t_i^R$; if $t_i^L\neq 0$ and $t_i^R=\infty$ then $Y_i$ is right-censored at $t_i^L$; otherwise, $Y_i$ is interval-censored between $t_i^L$ and $t_i^R$.
We assume independent interval censoring given the covariates. A definition of independent interval censoring can be found in, for example, Chapter 1 of \cite{sun2006statistical}. We assume every $i$ possesses both time-fixed and time-varying covariates, and their values are denoted by, respectively, $\mb{x}_i=(x_{i1}, \ldots, x_{ip})^\top$ and $\mb{z}_i(t)=(z_{i1}(t), \ldots, z_{iq}(t))^\top$. The time-fixed covariate values are usually collected at the baseline time. 

First, we introduce a 
Cox model containing time-fixed and time-varying covariates. Based on the time-varying covariates vector $\mb{z}_i(t)$, the corresponding history $\mc{Z}_i(t)$ is given by 
$
\mc{Z}_i(t) = \{\mb{z}_i(s): 0\leq s \leq t \}
$
and the Cox model is given by
\begin{equation}\label{eq1}
h_i(t | \mb{x}_i, \mc{Z}_i(t)) = h_0(t) e^{\mathbf{x}_i^\top \boldsymbol{\beta} + \mathbf{z}_i(t)^\top \boldsymbol{\gamma}},
\end{equation}
where $h_0(t)~ (\geq 0)$ is the non-parametric baseline hazard function. 
Therefore, this is a semi-parametric model. We 
consider $h_0(t)$ as a smooth function throughout this paper.

Instead of observing the entire trajectory $\mb{z}_i(t)$ for the time-varying covariates for individual $i$, in practice we may only observe its measurement error contaminated values, denoted by $\tilde{\mb{z}}_{i}(t_{i1}),\dots, \tilde{\mb{z}}_{i}(t_{in_i})$, at the intermittent time points $t_{i1}, \dots, t_{in_i}$. We assume the following model for these error contaminated  measurements:
\begin{equation}\label{Equa:dataLongitude}
    \tilde{\mb{z}}_{i}(t_{ia}) = \mb{z}_{i}(t_{ia}) + \bsb{\varepsilon}_i(t_{ia}),
\end{equation}
for $a = 1,\dots,n_i$, where $n_{i}$ can differ for different $i$ and $\bsb{\varepsilon}_i(t_{ia})$ denotes the $q\times 1$ vector for measurement errors. At a given $t_{ia}$, we assume for each $i$ and $a$ that $\bsb{\varepsilon}_i(t_{ia})$ is multivariate normal 
$\mathcal{N}(0, \sigma_{\varepsilon}^2\mb{I}_q)$, where $\mb{I}_q$ denotes an identity matrix of the dimension $q\times q$. For different $i$ and different $a$, these $\bsb{\varepsilon}_i(t_{ia})$'s are assumed independent.

The unobserved time-varying covariate values $z_{ir}(t)$ need to be estimated as they are demanded in the Cox model (\ref{eq1}). We consider each $z_{ir}(t)$ to be a smooth function which represents an individual trajectory perturbed from an overall mean trajectory.
Based on these assumptions, we can model  $z_{ir}(t)$ as:
\begin{equation}\label{eq:longit}
    z_{ir}(t) = \mu_r(t) + v_{ir}(t).
\end{equation}   
This represents a combination of a mean trajectory $\mu_r(t)$, common to all individuals, and an individualized random function $v_{ir}(t)$ which allows $z_{ir}(t)$ to deviate from the mean trajectory. 
Here, both $\mu_r(t)$ and $v_{ir}(t)$ are approximated using flexible basis functions, such as polynomials or splines, so that: 

$$ \mu_r(t) = \sum_{l' = 1}^{b} \alpha_{l'r}\phi_{l'r}(t);~ 
    v_{ir}(t) = \sum_{l = 1}^{c}  \kappa_{ilr}\xi_{lr}(t),
$$
where we assume $\alpha_{l'r}$ are fixed coefficients and $\kappa_{ilr}$ are random coefficients. Specifically, we assume, for different $i$, $\kappa_{ilr}$ 
are independently distributed with $\kappa_{ilr} \sim N(0, \sigma_{lr}^2)$. Clearly, the random coefficients $\kappa_{ilr}$ introduce within cluster $r$ dependence for different $i$.
We can express (\ref{eq:longit}) using matrices by 
\begin{equation} \label{Equa:tvcmodel}
 z_{ir}(t) = \bsb\phi_r(t)^\top\bsb\alpha_r+\bsb\xi_r(t)^\top\bsb\kappa_{ir},
\end{equation} 
where $\bsb\phi_r(t)=(\phi_{1r}(t), \ldots, \phi_{br}(t))^\top$, $\bsb\xi_r(t)=(\xi_{1r}(t), \ldots, \xi_{cr}(t))^\top$, $\bsb\alpha_r=(\alpha_{1r}, \ldots, \alpha_{br})^\top$ and $\bsb\kappa_{ir}=(\kappa_{i1r}, \ldots, \kappa_{icr})^\top$. 
The distributional assumption made about $\kappa_{ilr}$ can be written collectively as 
$\bsb\kappa_{ir} \sim \mathcal{N}(\mb{0}_{c\times 1}, \bsb\Sigma_{r})$, where $\bsb\Sigma_r = \text{diag}(\sigma_{1r}^2, \ldots, \sigma_{cr}^2)$.
Let $\bsb\kappa_i = (\bsb\kappa_{i1}^\top, \ldots, \bsb\kappa_{iq}^\top)^\top$ and $\bsb\alpha=(\bsb\alpha_1^\top, \ldots, \bsb\alpha_q^\top)^\top$.  
We comment that it is possible to also include baseline covariates (and interactions between baseline covariates and functions of time) in the model (\ref{Equa:dataLongitude}) (e.g., \cite{tsiatis_2004_JM_Overview}). However, as inclusion of baseline covariates will not complicate computations, we do not explicitly consider this option here 
to simplify discussions. 

Since a non-parametric baseline hazard $h_0(t)$ is an infinite dimensional parameter, direct estimation of this function from a finite number of observations is an ill-posed problem. We avoid this issue by restricting $h_0(t)$ to a finite dimensional functional space spanned by non-negative basis functions, where the dimension of this space is allowed to grow slowly with the sample size $n$. Thus, an approximated $h_0(t)$ is:
\begin{equation}\label{eq:approxBase}
    h_0(t) = \sum_{u=1}^m \theta_u \psi_u(t),
\end{equation}
where $\psi_u(t) \geq 0$ ($u=1, \ldots, m$) are typically determined using knots. Common choice of $m$ includes $m \approx n_0^{1/3}$, where $n_0$ is the non-right censored sample size. Common basis functions include M-spline, B-spline, Gaussian density function or even indicator functions; see \cite{MaCoHeIa21}. 
Let $\bsb{\theta} = (\theta_1, ..., \theta_m)^{\top}$, represent a vector of basis coefficients.  

Let $\tilde{\mb{z}}_i = (\tilde{\mb{z}}_i(t_{i1}), \ldots, \tilde{\mb{z}}_i(t_{in_i}))^\top$. This  represents the noise contaminated longitudinal covariates values for $i$, and it has the dimension of $n_i \times q$.
Then, from partly interval-censored event time observations $[t_i^L, t_i^R]$, and the associated covariates vectors $\mb{x}_i$ and $\tilde{\mb{z}}_i$, the data for individual $i$ can be expressed as $(t_i^L, t_i^R, \delta_i, \delta_i^L, \delta_i^R, \delta_i^I, \mb{x}_i, \tilde{\mb{z}}_i)$, where $\delta_i$ is the indicator for event times and $\delta_i^L, \delta_i^R, \delta_i^I$ are respectively indicators for a left, right or interval censoring times. 

In this article, we will develop maximum penalised likelihood estimates for the regression coefficient vectors $\bsb\beta$ and $\bsb\gamma$, as well as the baseline hazard coefficient vector $\bsb\theta$, the longitudinal fixed effect vector $\bsb\alpha$ and the longitudinal random effects vector $\bsb\kappa=(\bsb\kappa_{1}^\top, \ldots, \bsb\kappa_{n}^\top)^\top$, using data from all the individuals. Our approach differs from the conventional joint modelling methods, such as \cite{tsiatis_2004_JM_Overview} and \cite{Rizo_2012}, where an integration-out of random effects or an expectation-maximization (EM) algorithm are adopted. For the EM algorithm, the random effects are typically treated as missing values and therefore are not directly estimated.

\section{Penalised likelihood estimation 
} \label{sec:MPLesti}

\subsection{Penalised likelihood}

Following \cite{WeMa23}, we first re-write the Cox model in (\ref{eq1}) as 
\begin{equation}
 h_i(t | \mb{x}_i, \mc{Z}_i(t)) = h_0^*(t | \mc{Z}_i(t)) e^{\mb{x}_i^\top \bsb\beta},
\end{equation}
where $h_0^*(t | \mc{Z}_i(t)) = h_0(t) e^{\mb{z}_i(t)\bsb\gamma}$. The corresponding cumulative hazard is: 
\begin{equation}
 H_i(t | \mb{x}_i, \mc{Z}_i(t)) = H_0^*(t | \mc{Z}_i(t)) e^{\mb{x}_i^\top \bsb\beta} 
\end{equation}
with $H_0^*(t | \mc{Z}_i(t)) = \int_0^t h_0(s) e^{\mb{z}_i(s)^\top \bsb\gamma} ds$, which depends on the history of the time-varying covariates vector $\mb{z}_i(t)$ up to time $t$. In the following discussions, we will simply denote $H_0^*(t | \mc{Z}_i(t))$, $H_i(t | \mb{x}_i, \mc{Z}_i(t))$ and $S_i(t | \mb{x}_i, \mc{Z}_i(t))$ by $H_0^*(t)$, $H_i(t)$ and $S_i(t)$ respectively. Since $H_0^*(t)$ depends on $\bsb\alpha$ and $\bsb\kappa$ through $\mb{z}_i(t)$, so will $H_i(t)$ and $S_i(t) = \exp(-H_i(t))$.

We wish to estimate simultaneously the parameters of the Cox regression model \eqref{eq1} 
and the longitudinal model \eqref{Equa:dataLongitude}, including the fixed and random effects parameters. Assuming that the random interval $[T_i^L, T_i^R]$ and $\tilde{\mb{z}}_i(t)$ are conditionally independent given the random effects vector $\bsb\kappa$, an assumption usually required for joint modelling \citep{Rizo_2012}, we can write the joint log-likelihood from the observed partly interval-censored survival data and the longitudinal data as: 
\begin{align}
     l (\bsb\beta, \bsb\gamma, \bsb\theta, \bsb{\alpha}, \bsb\kappa      ) = & \sum_{i=1}^{n} l_i(\bsb\beta, \bsb\gamma, \bsb\theta, \bsb\alpha \,|\, \bsb\kappa)  
     -\frac{1}{2\sigma_{\varepsilon}^2} \sum_{i=1}^n\sum_{a=1}^{n_i}\|\tilde{\mb{z}}_{i}(t_{ia}) - \mb{z}_i(t_{ia})\|^2- \frac{N}{2} \ln \sigma_{\varepsilon}^2   \nonumber \\
     &  - \frac{n}{2}\sum_{r=1}^q \sum_{l=1}^c \ln \sigma_{lr}^2 - \frac{1}{2}\sum_{i=1}^n\sum_{r=1}^q \bsb\kappa_{ir}^\top\bsb\Sigma_r^{-1}\bsb\kappa_{ir}  , \label{Equa:full_loglik}
\end{align}
where $\mb{z}_i(t_{ia})$ 
is a time-varying covariates vector for $i$ at time $t_{ia}$ and 
$N=\sum_i n_i$. In (\ref{Equa:full_loglik}), $l_i(\boldsymbol{\beta}, \boldsymbol{\gamma}, \boldsymbol{\alpha}, \boldsymbol{\theta} ,|, \boldsymbol{\kappa})$ represents the conditional log-likelihood of the survival data given $\boldsymbol{\kappa}$ for individual $i$. This log-likelihood is given by
\begin{align*}
 l_i(\bsb\beta, \bsb\gamma,  \bsb\theta, \bsb\alpha \,|\, \bsb\kappa)=  &\delta_i \big(\ln h_0(t_i) + \mb{x}_i^{\top} \bsb{\beta} + \mb{z}_i(t_i)^{\top} \bsb{\gamma} - H(t_i)\big) - \delta_i^R H(t_i^{L}) \nonumber \\
 &+\delta_i^L \ln(1 - S(t_i^{R}))
     + \delta_i^I \ln (S(t_i^L) - S(t_i^R)). \label{iloglik}
\end{align*}

Since $h_0(t)$ is assumed to be a smooth function, we can enforce this smoothness by either selecting a small number of bases or choosing a moderate number of bases accompanied by a roughness penalty function. The benefit of a roughness penalty is that it reduces the impact of the number and location of knots. 
In this paper we use the penalty
given by $J_1(\bsb\theta)=\int h_0''(t)^2 dt = \bsb\theta^\top \mb{R}_{\bsb\theta} \bsb\theta$, where the $(u, v)$-th element of matrix $\mb{R}_{\bsb\theta}$ is $\int \psi_u''(t)\psi_v''(t)dt$. Similarly, penalty functions can also be used for estimating $\boldsymbol{\alpha}$ when, for example, the smoothness assumption is imposed on
the $\mu_r(t)$'s, where roughness penalties are 
given by $\sum_{r=1}^q\bsb\alpha_r^\top \mb{R}_{\bsb\alpha, r} \bsb\alpha_r$, and matrices 
$\mb{R}_{\bsb\alpha, r}$ have similar structure as $\mb{R}_{\bsb\theta}$.

Using these penalty functions, we define 
the following penalised log-likelihood function:
$$
[ \Phi(\bsb\beta, \bsb\gamma, \bsb\theta, \bsb{\alpha}, \bsb\kappa) = l(\bsb\beta, \bsb\gamma, \bsb\theta, \bsb{\alpha}, \bsb\kappa) - \lambda_{\bsb\theta} \bsb\theta^\top \mb{R}_{\bsb\theta} \bsb\theta - \sum_{r=1}^q\lambda_{\bsb\alpha, r} \bsb\alpha_r^\top \mb{R}_{\bsb\alpha, r} \bsb\alpha_r ],
$$
with $\lambda_{\bsb\theta}$ and $\lambda_{\bsb\alpha, 1}, \ldots, \lambda_{\bsb\alpha, q}$ being smoothing parameters which are all non-negative.
We will estimate the model parameters by solving the following constrained optimisation problem: 
\begin{equation} \label{Equa:mplcrit}
 (\widehat{\bsb\beta}, \widehat{\bsb\gamma}, \widehat{\bsb\theta}, \widehat{\bsb\alpha}, \widehat{\bsb\kappa}) = \argmax_{\bsb\beta, \bsb\gamma, \bsb\theta, \bsb{\alpha}, \bsb\kappa}\Phi(\bsb\beta, \bsb\gamma, \bsb\theta, \bsb{\alpha}, \bsb\kappa),
\end{equation}
subject to $\bsb\theta\geq 0$, where the inequalities are interpreted element-wise. The extent of penalisation will depend on the smoothing parameter values. This means that in addition to the estimation of the model parameters $\bsb\beta, \bsb\gamma, \bsb\theta, \bsb{\alpha}, \text{ and }\bsb\kappa$ and the estimation of the variance components $\sigma_{\varepsilon}^2$ and $\sigma_{lr}^2$ for $l = 1, \ldots, c$ and $r=1, \ldots, q$, we also need to estimate the optimal smoothing values for $\lambda_{\bsb\theta}$ and $\lambda_{\bsb\alpha, 1}, \ldots, \lambda_{\bsb\alpha, q}$.
In \ref{sec:smth} we will develop a marginal likelihood-based method to estimate $\lambda_{\bsb\theta}$ and $\lambda_{\bsb\alpha, 1}, \ldots, \lambda_{\bsb\alpha, q}$, which requires associating $\bsb\theta$ and
$\bsb\alpha_r$ quadratic penalties with normal prior distributions.

\subsection{Constrained optimisation }\label{constr_opt}

For given values of $\sigma_\varepsilon^2$, $\sigma_{lr}^2$, $\lambda_{\bsb\theta}$, and $\lambda_{\bsb\alpha, 1}, \ldots, \lambda_{\bsb\alpha, q}$, we propose to estimate the parameters $\bsb\beta, \bsb\gamma, \bsb\theta, \bsb\alpha, \bsb\kappa$ 
by maximising the penalised log-likelihood $\Phi(\bsb\beta, \bsb\gamma, \bsb{\alpha}, \bsb\theta, \bsb\kappa)$. The algorithm we adopt is an alternating iterative algorithm, similar to the Newton-MI algorithm described in \cite{WeMa23}. This algorithm ensures that the non-negativity constraint on $\bsb\theta$ is enforced. 

For a parameter vector $\mb{a}$, we let $\mb{a}^{(k)}$ denote the estimate of the parameters at iteration $k$. The values at iteration $k+1$ are then obtained through the following steps. 
Iteration $k+1$ of our algorithm updates all the parameter vectors in an alternating process. It computes $\bsb\beta^{(k+1)}$, $\boldsymbol{\gamma}^{(k+1)}$, $\boldsymbol{\gamma}^{(k+1)}$, $\boldsymbol{\theta}^{(k+1)}$, $\bsb{\alpha}^{(k+1)}$ and $\bsb{\kappa}^{(k+1)}$ in turn, and when updating each parameter vector, the previously updated parameters are used immediately, and also the penalised log-likelihood $\Phi$ is increased after updating a parameter vector.
These incremental penalised log-likelihood requirements are achieved through adopting a line search step size when updating each parameter vector. These conditions ensure that the penalised likelihood $\Phi(\bsb\beta, \bsb\gamma, \bsb\theta, \bsb{\alpha}, \bsb\kappa)$ is maximised at the end of the iterative process \citep{ChMa12}.

Specifically, similar to \cite{WeMa23}, we obtain $\bsb\beta^{(k+1)}$ according to
\begin{align} 
    \bsb\beta^{(k+1)} &= \bsb\beta^{(k)} + \omega_1^{(k)}(\mb{X}^\top \mb{A}^{(k)} \mb{X})^{-1}\mb{X}^\top(\bsb\delta - \mathbf{C}^{(k)}\mathbf{1}_n), \label{betaupdat}
\end{align}
where, $\mb{X}$ denotes the model matrix from the time-fixed covariates of the Cox model, $\bsb\delta$ is a vector of length $n$ with its elements equal to 1 for event times and 0 for others, and $\mb{1}_n$ is a vector of $1$'s with length $n$ and both $\mb{C}$ and $\mb{A}$ are $n \times n$ diagonal matrices with diagonal vectors $\mb{c}$ and $\mb{a}$ respectively, and details of $\mb{c}$ and $\mb{a}$ are available in the Supplementary Material of this paper. 

The update for $\bsb{\gamma}$ is given by
\begin{align}
     \bsb{\gamma}^{(k+1)} &= \bsb{\gamma}^{(k)} - \omega_2^{(k)} \left[ \mb{H}_{\bsb\gamma}^{(k)} \right]^{-1} \mb{g}_{\bsb\gamma}^{(k)}, \label{gam_update}
\end{align}
where $\mb{H}_{\bsb\gamma}^{(k)}$ and $\mb{g}_{\bsb\gamma}^{(k)}$ are respectively the modified Hessian matrix and gradient vector for $\bsb\gamma$, both evaluated at $(\bsb\beta^{(k+1)}, \bsb\gamma^{(k)}, \bsb\theta^{(k)}, \bsb\alpha^{(k)}, \bsb\kappa^{(k)})$. 

We compute the update $\bsb\theta^{(k+1)}$ using an MI algorithm so that the estimate of $\bsb\theta$ is guaranteed to be non-negative. This algorithm can also be expressed as a gradient algorithm \citep{ChMa12}. From a current estimate $\bsb\theta^{(k)}$, we update $\boldsymbol{\theta}$ using
\begin{align} 
    \bsb\theta^{(k + 1)} = \bsb\theta^{(k)} + \omega_3^{(k)}[\mb{S}_{\bsb\theta}^{(k)}]^{-1} \mb{g}_{\bsb\theta}^{(k)},\label{theta_update} 
\end{align}
where $\mb{g}_{\bsb\theta}$ represents the gradient of $\Phi$ for $\bsb\theta$ and $\mb{S}_{\bsb\theta}$ is a diagonal matrix involving only the first derivative of $\Phi$, and these are evaluated at $(\bsb\beta^{(k+1)}, \bsb\gamma^{(k+1)}, \bsb\theta^{(k)}, \bsb\alpha^{(k)}, \bsb\kappa^{(k)})$, with details given in the Supplementary Material. 

Finally, the updates for $\bsb\alpha$ and $\bsb\kappa$ are given by modified Newton algorithms:
\begin{align}
    \bsb{\alpha}^{(k+1)} &= \bsb{\alpha}^{(k)} - \omega_4^{(k)} \left[ \mb{H}_{\bsb\alpha}^{(k)} \right]^{-1} \mb{g}_{\bsb\alpha}^{(k)}, \label{alp_update}\\
    \bsb\kappa^{(k+1)} &= \bsb\kappa^{(k)} - \omega_5^{(k)} \left[\mb{H}_{\kappa}^{(k)}\right]^{-1} \mb{g}_{\kappa}^{(k)}, \label{kap_update}
\end{align}
where, in (\ref{alp_update}), $\mb{H}_{\bsb\alpha}^{(k)}$ and $\mb{g}_{\bsb\alpha}^{(k)}$ are respectively the modified Hessian matrix and gradient vector for $\bsb\alpha$, both evaluated at $(\bsb\beta^{(k+1)}, \bsb\gamma^{(k+1)}, \bsb\theta^{(k+1)}, \bsb\alpha^{(k)}, \bsb\kappa^{(k)})$, and in (\ref{kap_update}), $\mb{H}_{\bsb\kappa}^{(k)}$ and $\mb{g}_{\bsb\kappa}^{(k)}$ are respectively the modified Hessian and gradient for $\bsb\kappa$,  evaluated at $(\bsb\beta^{(k+1)}, \bsb\gamma^{(k+1)},  \bsb\theta^{(k+1)}, \bsb\alpha^{(k+1)}, \bsb\kappa^{(k)})$. These steps are repeated until the differences of all the parameter estimates at two consecutive iterations are smaller than some tolerance level, such as $10^{-5}$. 

In the above, $\omega_s$ for $s = 1, \dots, 5$ are line search step sizes used to ensure that the likelihood does not decrease at any iteration. There exist different ways to do the line search, and we recommend the Armijo backtracking line search method (e.g. \cite{LuenbergerYe}). 

\subsection{Estimation of smoothing parameters and variance components} \label{sec:smth}

The estimation of the model parameters vector $\bsb\eta = (\bsb\beta^\top, \bsb\gamma^\top, \bsb\theta^\top, \bsb\alpha^\top, \bsb\kappa^\top)^\top$, discussed in the last section, is made with given values for the variance components $\sigma_\varepsilon^2$ and $\sigma_{lr}^2$ (for $l = 1, \ldots, c;~ r = 1, \ldots, q$) and the smoothing parameters $\lambda_{\bsb\theta}$ and $\lambda_{\bsb\alpha, 1}, \ldots, \lambda_{\bsb\alpha, q}$. Following convergence of the Newton-MI algorithm described in Section \ref{constr_opt}, another set of iterations will be conducted to update the estimates of $\sigma_{\varepsilon}^2$, $\sigma_{lr}^2$ (for all $l$ and $r$), $\lambda_{\bsb\theta}$, and $\lambda_{\bsb\alpha, 1}, \ldots, \lambda_{\bsb\alpha, q}$. This is achieved by maximising an approximate marginal log-likelihood, adopting the Laplace approximation. 

For the penalised log-likelihood $\Phi(\bsb\eta)$ given in (\ref{Equa:mplcrit}), due to the fact that penalties on $\bsb\theta$ and all $\bsb\alpha_r$ are quadratic, $\Phi(\bsb\eta)$ can be viewed as a log-posterior with normal priors: $\bsb\theta \sim N(\mb{0}, \sigma_{\bsb\theta}^2\mb{R}_{\bsb\theta}^{-1})$ and $\bsb\alpha_r \sim N(\mb{0}, \sigma_{\bsb\alpha, r}^2\mb{R}_{\bsb\alpha, r}^{-1})$ for $r=1, \ldots, q$, where $\lambda_{\bsb\theta} = 1/2 \sigma_{\bsb\theta}^2$ and $\lambda_{\bsb\alpha, r} = 1/2 \sigma_{\bsb\alpha, r}^2$. 
Let 
$\bsb\sigma_{\bsb\alpha}^2=(\sigma_{\bsb\alpha, 1}^2, \ldots, \sigma_{\bsb\alpha, q}^2)^\top$ be a vector for all the variance components of
$\bsb\alpha$. 
Using the prior distributions for $\bsb\theta$ and $\bsb\alpha_r$'s, function $\Phi(\bsb\eta)$, which will be denoted as $\Phi(\bsb\eta; \sigma_{\varepsilon}^2, \sigma_{\bsb\theta}^2, \bsb\sigma_{\bsb\alpha}^2, \bsb\sigma_{\bsb\kappa}^2)$ in the following discussions to reflect its dependence on all the variance components, can be re-expressed as: 
\begin{align}
 \Phi(\bsb\eta; \sigma_{\varepsilon}^2, \sigma_{\bsb\theta}^2, \bsb\sigma_{\bsb\alpha}^2, \bsb\sigma_{\bsb\kappa}^2) = & l(\bsb\eta; \sigma_{\varepsilon}^2, \sigma_{\bsb\theta}^2, \bsb\sigma_{\bsb\alpha}^2, \bsb\sigma_{\bsb\kappa}^2) \nonumber - \frac{m}{2}\ln \sigma_{\bsb\theta}^2 - \frac{1}{2\sigma_{\bsb\theta}^2}\bsb\theta^\top \mb{R}_{\bsb\theta} \bsb\theta \\
 & + \sum_{r=1}^q\Big(-\frac{b}{2}\ln \sigma_{\bsb\alpha, r}^2 - \frac{1}{2\sigma_{\bsb\alpha, r}^2}\bsb\alpha_r^\top \mb{R}_{\bsb\alpha, r} \bsb\alpha_r\Big),
\end{align}
where $l(\bsb\eta; \sigma_{\varepsilon}^2, \sigma_{\bsb\theta}^2, \bsb\sigma_{\bsb\alpha}^2, \bsb\sigma_{\bsb\kappa}^2)$ is the log-likelihood given by (\ref{Equa:full_loglik}). We estimate the variance components $(\sigma_{\varepsilon}^2, \sigma_{\bsb\theta}^2, \bsb\sigma_{\bsb\alpha}^2, \bsb\sigma_{\bsb\kappa}^2)$ by maximising their marginal log-likelihood, obtained via Laplace approximation. The Laplace approximation gives the marginal log-likelihood
\begin{align}
    l_m(\sigma_{\varepsilon}^2, \sigma_{\bsb\theta}^2, \bsb\sigma_{\bsb\alpha}^2, \bsb\sigma_{\bsb\kappa}^2) \approx  
    \Phi(\widehat{\bsb\eta}; \sigma_{\varepsilon}^2, \sigma_{\bsb\theta}^2, \bsb\sigma_{\bsb\alpha}^2, \bsb\sigma_{\bsb\kappa}^2) -\frac{1}{2} \ln \big|\widehat{\mb{F}}_{\bsb\eta}\big| \label{marlike}
\end{align}
where 
$\widehat{\mb{F}}_{\bsb\eta} = -\partial^2 \Phi(\widehat{\bsb\eta}; \sigma_{\varepsilon}^2, \sigma_{\bsb\theta}^2, \bsb\sigma_{\bsb\alpha}^2, \bsb\sigma_{\bsb\kappa}^2)/\partial \bsb\eta\bsb\eta^\top$ is the negative Hessian of $\Phi$ with respect to $\bsb\eta$ and evaluated at the MPL estimate $\widehat{\bsb\eta}$ (where variance components are fixed at their current values). Details of all components of the matrix ${\mb{F}}_{\bsb\eta}$ can be found in the Supplementary Material.

From this marginal likelihood we can obtain approximate solutions for each of the variance components. 
These approximate solutions are:
 \begin{align}
    \widehat{\sigma}_{\varepsilon}^2 = \sum_{i = 1}^n \sum_{a = 1}^{n_i} \|\tilde{\mb{z}}_i(t_{ia}) - \widehat{\mb{z}}_i(t_{ia}) \|^2/(N - \nu_{\varepsilon}), \label{eq11} \\
    \widehat{\sigma}_{\bsb\theta}^2 = \widehat{\bsb{\theta}}^{\top} \mb{R}_{\bsb\theta} \widehat{\bsb{\theta}}/(m - \nu_{\bsb\theta}), \label{eq12} \\
    \widehat{\sigma}_{\bsb\alpha, r}^2 = \widehat{\bsb{\alpha}}_r^{\top} \mathbf{R}_{\bsb\alpha, r} \widehat{\bsb\alpha}_r/(b - \nu_{\bsb\alpha,r}), \label{eq13} \\
    \widehat{\sigma}_{lr}^2 = \sum_{i=1}^n \kappa_{ilr}^2/(n - \nu_{\bsb\kappa, lr}), \label{eq14} 
\end{align}
where 
$\nu_{\varepsilon} = \text{tr}(\widehat{\mb{F}}_{\bsb\eta}^{-1} \widehat{\mb{Q}}_{\varepsilon})$, $\nu_{\bsb\theta} = \text{tr}(\widehat{\mb{F}}_{\bsb\eta}^{-1} \widehat{\mb{Q}}_{\bsb\theta})$, $\nu_{\bsb\alpha, r} = (\widehat{\mb{F}}_{\bsb\eta}^{-1} \widehat{\mb{Q}}_{\bsb\alpha, r})$ and $\nu_{\bsb\kappa, lr} = \text{tr}(\widehat{\mb{F}}_{\bsb\eta}^{-1} \widehat{\mb{Q}}_{\bsb\kappa, lr})$; the matrices $\widehat{\mb{Q}}_{\varepsilon}$, $\widehat{\mb{Q}}_{\bsb\theta}$, $\widehat{\mb{Q}}_{\bsb\alpha, r}$, and $\widehat{\mb{Q}}_{\bsb\kappa}$ are components of $\widehat{\mb{F}}_{\bsb\eta}$, detailed in the Supplementary material.

The structure of these variance component updating formulae
suggest an iterative structure with the penalised likelihood estimation of $\boldsymbol{\eta}$ and marginal likelihood estimation of variance components forming inner and outer iterations respectively. In practice, after the Newton-MI algorithm has converged, we update each variance component in the model using the formulae above. We then alternate between the Newton-MI algorithm and the variance component updating steps until both have been stabilised.

\section{Large sample variance estimation} \label{sec:asymp}

In this section, we focus on some details of estimating the large sample covariance matrix for the fixed parameters. In the Supplementary Material, we provide additional detail on some asymptotic results for the maximum penalised likelihood estimates of the parameters of the joint model, including a large sample distribution result for the estimated parameters.

Let 
$\bsb\zeta = (\bsb\beta^\top, \bsb\gamma^\top, \bsb\theta^\top, \bsb\alpha^\top)^\top$ be 
the collection of all the 
parameters except $\bsb\kappa$. In practice, estimation of the large sample covariance matrix for $\widehat{\bsb\zeta}$ 
should allow for scenarios where 
the smoothing parameters $\lambda_{\theta}$ and $\lambda_{\alpha, r}$ (for all $r$)
are not zero, and also for scenarios where there are active constraints in the estimation of $\boldsymbol{\theta}$, i.e. where some $\widehat{\theta}_u = 0$ at convergence. 

Here, we describe a similar sandwich formula for a large sample 
covariance matrix as \cite{MaCoHeIa21}. Recall $\boldsymbol{\eta} = (\boldsymbol{\zeta}^\top, \boldsymbol{\kappa}^\top)^\top$ and
the negative Hessian of $\Phi$ with respect to $\bsb\eta$ is denoted by $\mb{F}_{\bsb\eta}$. $\widehat{\mb{F}}_{\bsb\eta}$ represents matrix $\mb{F}_{\bsb\eta}$ with $\bsb\eta=\widehat{\bsb\eta}$. Assuming at the end of the iterations we have identified the active constraints from $\theta_u\geq 0$, and we denote their index set by $\mathcal{B}$. 
Since the asymptotic distribution result described above is for $\widehat{\bsb\zeta}$ with a fixed $\bsb\kappa$, the covariance matrix of $\widehat{\bsb\zeta}$ can be obtained by the inverse of the negative Hessian of the profile penalized log-likelihood where $\bsb\kappa=\widehat{\bsb\kappa}$. This can be computed by taking the portion corresponding to $\bsb\zeta$ from the inverse of $\widehat{\mb{F}}_{\bsb\eta}$. 
Active constraints will make this process more complicated as we need to first remove rows and columns, corresponding to ${\cal B}$, from $\widehat{\mb{F}}_{\bsb\eta}$. Denote the resulted $\widehat{\mb{F}}_{\bsb\eta}$ by $\widehat{\mb{F}}^*_{\bsb\eta}$. Then we invert $\widehat{\mb{F}}^*_{\bsb\eta}$ and add zeros back to this inverse matrix; see \cite{MaCoHeIa21}. Mathematically, this process can be expressed using a matrix $\mb{U}$ (see \cite{MaCoHeIa21} for how to define $\mb{U}$), which is basically a combination of a zero (corresponding to $\mathcal{B}$) and an identity matrix. Let $c$ be the number of active constraints and $d$ the size of $\bsb\eta$. Then, this matrix has the dimension of $(d-c)\times d$. Use these notations, the pseudo inverse, denoted by $\widehat{\mathbf{A}}^{-1}$, of the negative Hessian matrix becomes:
$
\widehat{\mathbf{A}}^{-1} = \mathbf{U} \left( \mathbf{U}^\top \widehat{\mb{F}}_{\bsb\eta}\mathbf{U} \right)^{-1} \mathbf{U}^\top.
$
The sandwich formula for an approximte covariance of $\widehat{\bsb\eta}$ is then  
\begin{align} \label{full_sandwich}
    \widehat{\mb{H}}_{\bsb\eta}^{-1} = \widehat{\mathbf{A}}^{-1}\frac{\partial^2 l(\widehat{\bsb\eta})}{\partial \bsb\eta \partial \bsb\eta^\top} \widehat{\mathbf{A}}^{-1}, 
\end{align}
where $l(\bsb\eta)$ is the log-likelihood given in (\ref{Equa:full_loglik}). Let $\widehat{\mathcal{I}}_{\bsb\zeta\bsb\zeta}$ be the portion of $\widehat{\mb{H}}_{\bsb\eta}^{-1}$ corresponding to $\bsb\zeta$, then 
an estimate the covariance matrix of $\bsb\zeta$ is:  
$ \widehat{\text{Cov}}(\widehat{\boldsymbol{\eta}}) = \widehat{\mathcal{I}}_{\bsb\zeta\bsb\zeta}$.

To use the above formulae, we require a method for identifying active constraints in practice. The method proposed here closely follows that of \cite{WeMa23}. Active constraints can be identified by inspecting both the value of $\hat{\theta}_u$ and the corresponding gradient for each $u$. After the Newton-MI algorithm has reached convergence, some $\hat{\theta}_u$ may be exactly zero with negative gradients, and thus  clearly they are active constraints. Furthermore, there may be some $\hat{\theta}_u$ that are very close to, but not exactly, zero. For these $\hat{\theta}_u$, a corresponding negative gradient value is indicative that they are also subject to an active constraint. In practice, active constraints are defined where, for a given $u$, $\hat{\theta}_u < 10^{-2}$ and the corresponding gradient is less than $-\varepsilon$ where $\varepsilon$ is a positive threshold value such as $10^{-2}$.

The simulation results reported in Section \ref{sec:simu} reveal that the standard error estimates obtained 
using this approach are very close to those obtained from the Monte Carlo method.

\section{Simulation studies}
\label{sec:simu}

In this section we evaluate the performance of our proposed MPL method via two simulation studies. We investigate the bias, standard error estimation (by comparing estimates from the large sample formula with Monte Carlo estimates) and coverage probabilities for the fixed parameters, $\bsb{\beta}$, $\bsb{\gamma}$, and $\bsb{\alpha}$, and evaluate the performance of the baseline hazard function spline approximation visually and by computing the mean integrated square error (MISE). We also investigate the bias in the estimates of the variance components of the measurement error and random effects distributions.

In Study 1, we compare the performance of our MPL method to the method implemented in the \texttt{JM R} package when applied to data subject to right censoring only (no left or interval censoring). In Study 2, we consider (partly-)interval censored data. In Study 2a, we demonstrate the performance of our MPL method (which specifically accommodates partly-interval censoring) under partly-interval censoring, compared to fitting the model using midpoint imputation in the \texttt{JM R} package (which otherwise does not accommodate partly-interval censored data). In Study 2b, we compare the performance of our proposed MPL approach with the method proposed by \cite{chen_2018_TJM_IC} under interval censoring. In Study 2c, detailed in the Supplementary Material, we further investigate the performance of our proposed MPL approach when a spline model is used to approximate the longitudinal trajectory.

\subsection{Data generation strategy}

For Study 1, event times $y_i$ for $i = 1, ..., n$ were generated from \eqref{eq1} using a Weibull baseline hazard $h_0(t) = 3t^2$, with true coefficient values $\beta = -0.5$ and $\gamma = 0.5$, and  time-fixed covariate values for the Cox model from $x_i \sim \text{Bern}(0.5)$. Values for the time-varying covariate were generated from the model
\begin{equation}\label{long_generate}
    \tilde{z}_i(t) = \alpha_0 + \alpha_1 t +\alpha_2t^2 +\alpha_3t^3+\kappa_{i1} + \kappa_{i2} t 
 + \varepsilon_i(t)
\end{equation}
where the true coefficient values were $\alpha_0 = 0.5$, $\alpha_1 = -0.5$, $\alpha_2 = 1$ and $\alpha_3 = -0.5$ and $\kappa_{i1} \sim \mathcal{N}(0, 0.5^2)$, and $\kappa_{i2} \sim \mathcal{N}(0, 0.8^2)$ were the random effects. We considered two scenarios, one with the random measurement error $\varepsilon_i(t) \sim \mathcal{N}(0, 0.05^2)$ and another with $\varepsilon_i(t) \sim \mathcal{N}(0, 0.2^2)$. The censoring times $c_i$, for $i = 1, ..., n$, were drawn from a uniform distribution $\text{Unif}[\tau_1, \tau_2]$ where the values of $\tau_1$ and $\tau_2$ could be adjusted to control the proportion of observations that were right censored. The observed times $t_i$ were obtained by taking $t_i = \min(y_i, c_i)$.

Observed values of the time-varying covariate $\tilde{z}_i(t)$ were obtained at random observation times $t_{ia}$, where $a = 1, ..., n_i$. The number $n_i$ and timing of $t_{ia}$ were allowed to vary across individuals. For each $i$, the corresponding $t_{ia}$'s 
were generated firstly by drawing a random number $n_i \sim \text{Poiss}(\nu_1)$.
Then the gap times between neighbouring $t_{ia}$ were drawn independently from $\text{Unif}[0, \nu_2]$. The value of each $t_{ia}$ was then found by taking the cumulative sum of these gap times up to the value of $t_i$ (i.e. any generated $t_{ia}$ greater than the event or censoring time $t_i$ was discarded). We considered two scenarios for $n_i$: in one the mean $n_i = 5$, and in the other the mean $n_i = 20$. 

In Study 2a, we now generated interval-censored data, and compared the performance of our proposed MPL approach to fitting a joint model using midpoint imputation. True event times $y_i$ were generated with a log-logistic baseline hazard function $h_0(t) = 4t^3/ (1 + t^4)$, with Cox model covariates $\beta_1 = -0.5$, $\beta_2 = 0.5$ and $\gamma = 0.25$, and time-fixed covariates in the Cox model drawn from $x_{i1} \sim \text{Unif}[-1, 1]$ and $x_{i2} \sim \text{Bern}(0.5)$. We now generated longitudinal trajectories from a model including an additional time-fixed covariate and a covariate $\times$ time interaction, such that
\begin{align*}
    \tilde{z}_i(t) = \alpha_0 + \alpha_1 x_{i3} + \alpha_2 t + \alpha_3 x_{i3} t + \alpha_4 t^2 + \alpha_5 t^3 + \kappa_{0i} + \kappa_{1i}t + \varepsilon_i(t)
\end{align*}
where the true coefficient values were $\alpha_0 = \alpha_1 = \alpha_2 = \alpha_3 = 0.5$, $\alpha_4 = -0.8$ and $\alpha_5 = 0.2$, with time-fixed covariate $x_{i3} \sim \mathcal{N}(0, 1)$; the random intercept was $\kappa_{0i} \sim \mathcal{N}(0, 0.2^2)$, the random slope for time was $\kappa_{1i} \sim \mathcal{N}(0, 0.3^2)$, and the measurement error was $\varepsilon_i(t) \sim \mathcal{N}(0, 0.1^2)$.

The observed times $t_i^L$ and $t_i^R$ were found by taking 
\begin{align*}
    t_i^L = y_i^{I(U_i^E < \pi^E)} (\tau_L U_i^L)^{I(\pi^E < U_i^E, \tau_L U_i^L \leq y_i \leq \tau_R U_i^R)}(\tau_R U_i^R)^{I (\pi^E \leq U_i^E, \tau_R U_i^R < y_i)} 0^{I(\pi^E \leq U_i^E, y_i < \tau_L U_i^L)} \\
    t_i^R = y_i^{I(U_i^E < \pi^E)}(\tau_L U_i^L)^{I(\pi^E \leq U_i^E, y_i < \tau_L U_i^L)} (\tau_R U_i^R)^{I(\pi^E \leq U_i^E, \tau_L U_i^L \leq y_i \leq \tau_R U_i^R)} \infty^{I(\pi^E \leq U_i^E, \tau_R U_i^R < y_i)}
\end{align*}
where $\pi^E$ denotes the desired event proportion, $U_i^E$ and $U_i^L$ denote independent standard uniform variables, $U_i^R \sim \text{Unif}[U_i^L, 1]$, and $\tau_L$ and $\tau_R$ are scalars that define the width of the censoring intervals and therefore they control the proportions of left, right and interval censoring. Function $I(\cdot)$ in the above expressions represents an indicator function. The proportion of right, left and interval censoring resulting from this generation process are given in Table \ref{int_cens} below.

In Study 2b, we compare our proposed penalised likelihood approach with the method for fitting a joint model with interval censored event times proposed in \cite{chen_2018_TJM_IC}. The \texttt{R} code available in the Supplementary Material of \cite{chen_2018_TJM_IC} allows for a simple linear mixed model with one time-fixed covariate in the longitudinal model, with one longitudinal and one time-fixed covariate in the Cox model, and hence for ease of comparison in this study we generate time-varying covariate values from the model
\begin{align}
    \tilde{z}_i(t) = \alpha_0 + \alpha_1 t + \alpha_2 x_{i1} + \kappa_{i0} + \kappa_{i1}t + \varepsilon_i(t)
\end{align}
where $\alpha_0 = \alpha_1 = -0.1$ and $\alpha_2 = -0.3$, with the time-fixed covariate $x_{i1} \sim \text{Bern}(0.5)$, the random effects $\kappa_{i0} \sim \mathcal{N}(0, 0.2^2)$ and $\kappa_{i1} \sim \mathcal{N}(0, 0.4^2)$, and the measurement error $\varepsilon_i(t) \sim \mathcal{N}(0, 0.1^2)$. We use a Gompertz baseline hazard function $h_0(t) = 0.5 e^{2t}$, with the time-fixed covariate in the Cox model being $x_{i2} \sim \mathcal{N}(0,1)$ and true Cox regression coefficient values of $\beta = -1$ and $\gamma = -0.3$.

The method proposed by \cite{chen_2018_TJM_IC} estimates the cumulative baseline hazard function as a series of non-negative probability masses at the (finite) ends of each individual censoring interval, meaning that the number of quantities estimated increases with each unique censoring interval. To reduce computational burden, we therefore assume an observation scheme where the longitudinal covariate and event status of each individual is observed at a series of common pre-specified time points; we considered a series of observation times $\boldsymbol{\tau} = \{ 0.25, 0.5, \dots \}$. This observation sequence was truncated by an independently generated ``drop-out" time $c_i \sim \text{Unif}[0.5, 3]$ to obtain $\boldsymbol{\tau}_i$, and hence an individual was considered right censored if their true event time $t_i$ was greater than $\max(\boldsymbol{\tau}_i)$, left censored if $t_i < \min(\boldsymbol{\tau}_i)$, and interval censored otherwise.

Details of the data generation strategy for Study 2c are available in the Supplementary Material in Section G.

\subsection{Study 1 (right censored) simulation results}

The simulation 
results from Study 1, including the bias, asymptotic and Monte Carlo standard error estimates and the asymptotic and Monte Carlo coverage probabilities for the regression coefficients in the survival and longitudinal models, are summarised in Tables \ref{right_cens1} and \ref{right_cens2}. For the regression coefficients in the survival model, $\beta$ and $\gamma$, the MPL and JM methods have generally comparable performance in terms of bias and standard error estimation. The proposed estimate of the standard errors of $\beta$ and $\gamma$ performs well, with coverage probabilities close to the nominal value of 0.95 regardless of sample size, censoring proportion or the magnitude of the measurement error. 

The bias in the estimates of $\boldsymbol{\alpha}$ are again comparable between both the MPL and JM methods, with both methods having small biases in these estimates. However, our results indicate that in these simulation scenarios the MPL estimates of the standard errors for $\boldsymbol{\alpha}$ are superior to those from the \texttt{JM} package. For the MPL method, the estimated asymptotic standard errors for $\boldsymbol{\alpha}$ are close to the Monte Carlo values across all scenarios and all MPL coverage probabilities for $\boldsymbol{\alpha}$ are close to the nominal value. In contrast, the asymptotic standard errors for $\boldsymbol{\alpha}$ from the JM method are consistently much smaller than the Monte Carlo estimates, and coverage probabilities consequently fall below the nominal value. 
We hypothesise that this difference in the standard error estimation between the MPL and \text{JM} methods arises from the different likelihoods used; the \text{JM} standard error estimates are found from the second derivative of a likelihood that marginalises out the random effects in the longitudinal model, while the random effects are not marginalised in the likelihood used in our proposed MPL method.

Figure \ref{fig:right_cens_h0t} shows the performance of the baseline hazard function estimates from the MPL method. In all scenarios in Study 1, the bias in the estimate of the baseline hazard function was small and the true function was contained in the asymptotic $95\%$ point-wise confidence interval 
area. Table \ref{tab:var_comp_stdy1} also compares the mean integrated square error (MISE) of the baseline hazard function estimates from the MPL and JM methods across all scenarios in Study 1. In the scenarios with $70\%$ event proportions, the MPL and JM estimates have comparable MISEs. When censoring is increased, the MISEs of the MPL method are substantially lower than those from the JM method, suggesting that the MPL baseline hazard approximation is adequate.

\begin{table}[h!]
    \centering
    \tiny
    \begin{tabular}{ll | cccccc | cccccc}
         \hline
         && \multicolumn{6}{c}{$n = 200$, $\Bar{n}_i = 5$, $\pi^E = 0.7$} & \multicolumn{6}{c}{$n = 200$, $\Bar{n}_i = 20$, $\pi^E = 0.7$}\\
         \hline
         && $\beta$ & $\gamma$ & $\alpha_0$ & $\alpha_1$ & $\alpha_2$ & $\alpha_3$ & $\beta$ & $\gamma$ & $\alpha_0$ & $\alpha_1$ & $\alpha_2$ & $\alpha_3$ \\
         \hline
         Bias & MPL & 0.012 & 0.002 & 0.003 & 0.002 & -0.002 & 0.001 & -0.006 & 0.016 & 0.002 & 0.003 & -0.001 & 0.001 \\
         & JM & 0.025 & -0.015 & 0.003 & -0.006 & -0.001 & -0.001 & 0.011 & 0.012 & -0.001 & 0.002 & 0.017 & -0.008 \\
         SE & MPL & 0.145 & 0.086 & 0.035 & 0.063 & 0.064 & 0.032 & 0.143 & 0.086 & 0.035 & 0.058 & 0.028 & 0.013 \\
         & & (0.146) & (0.098) & (0.037) & (0.063) & (0.065) & (0.033)& (0.154) & (0.099) & (0.035) & (0.064) & (0.027) & (0.012) \\
         & JM & 0.148 & 0.102 & 0.014 & 0.072 & 0.124 & 0.056 & 0.147 & 0.098 & 0.009 & 0.046 & 0.072 & 0.031 \\
         & & (0.150) & (0.108) & (0.091) & (0.138) & (0.232) & (0.130) & (0.162) & (0.105) & (0.107) & (0.182) & (0.202) & (0.109) \\
         CP & MPL & 0.96 & 0.94 & 0.93 & 0.95 & 0.95 & 0.95 & 0.94 & 0.92 & 0.97 & 0.94 & 0.95 & 0.98 \\
         & & (0.96) & (0.97) & (0.96) & (0.95) & (0.93) & (0.94)& (0.96) & (0.96) & (0.96) & (0.95) & (0.95) & (0.96) \\
         & JM & 0.95 & 0.93 & 0.26 & 0.66 & 0.73 & 0.63 & 0.94 & 0.93 & 0.13 & 0.35 & 0.51 & 0.40 \\
         & & (0.96) & (0.97) & (0.96) & (0.95) & (0.93) & (0.94) & (0.95) & (0.95) & (0.97) & (0.95) & (0.95) & (0.95) \\
         \hline
         && \multicolumn{6}{c}{$n = 1000$, $\Bar{n}_i = 5$, $\pi^E = 0.7$} & \multicolumn{6}{c}{$n = 1000$, $\Bar{n}_i = 20$, $\pi^E = 0.7$}\\
         \hline
         && $\beta$ & $\gamma$ & $\alpha_0$ & $\alpha_1$ & $\alpha_2$ & $\alpha_3$ & $\beta$ & $\gamma$ & $\alpha_0$ & $\alpha_1$ & $\alpha_2$ & $\alpha_3$ \\
         \hline
         Bias & MPL & 0.001 & -0.005 & 0.001 & 0.002 & 0.003 & 0.001 & 0.012 & 0.001 & -0.002 & 0.003 & -0.001 & 0.001 \\
         & JM & -0.004 & 0.004 & -0.001 & -0.011 & -0.018 & -0.014 & 0.006 & 0.017 & -0.001 & 0.009 & 0.009 & -0.004 \\
         SE & MPL & 0.064 & 0.037 & 0.016 & 0.028 & 0.027 & 0.013 & 0.064 & 0.037 & 0.016 & 0.026 & 0.012 & 0.005 \\
         & & (0.063) & (0.041) & (0.016) & (0.029) & (0.026) & (0.013) & (0.062) & (0.037) & (0.015) & (0.030) & (0.013) & (0.005) \\
         & JM & 0.065 & 0.044 & 0.010 & 0.041 & 0.069 & 0.031 & 0.064 & 0.043 & 0.005 & 0.020 & 0.030 & 0.013 \\
         & & (0.062) & (0.043) & (0.076) & (0.092) & (0.158) & (0.084) & (0.062) & (0.035) & (0.084) & (0.129) & (0.102) & (0.053) \\
         CP & MPL & 0.95 & 0.95 & 0.93 & 0.94 & 0.97 & 0.95 & 0.97 & 0.96 & 0.94 & 0.96 & 0.94 & 0.94 \\
         & & (0.95) & (0.95) & (0.95) & (0.95) & (0.95) & (0.96)& (0.97) & (0.96) & (0.94) & (0.96) & (0.96) & (0.94) \\
         & JM & 0.95 & 0.95 & 0.23 & 0.61 & 0.55 & 0.50 & 0.96 & 0.99 & 0.11 & 0.24 & 0.43 & 0.35 \\
         & & (0.94) & (0.95) & (0.95) & (0.95) & (0.95) & (0.96) & (0.96) & (0.97) & (0.94) & (0.94) & (0.97) & (0.94) \\
         \hline
         && \multicolumn{6}{c}{$n = 200$, $\Bar{n}_i = 5$, $\pi^E = 0.3$} & \multicolumn{6}{c}{$n = 200$, $\Bar{n}_i = 20$, $\pi^E = 0.3$}\\
         \hline
         && $\beta$ & $\gamma$ & $\alpha_0$ & $\alpha_1$ & $\alpha_2$ & $\alpha_3$ & $\beta$ & $\gamma$ & $\alpha_0$ & $\alpha_1$ & $\alpha_2$ & $\alpha_3$ \\
         \hline
         Bias & MPL & 0.035 & -0.015 & -0.001 & -0.002 & -0.023 & 0.017 & 0.009 & -0.005 & -0.001 & 0.005 & -0.001 & 0.001 \\
         & JM & 0.034 & -0.028 & -0.008 & 0.008 & 0.021 & -0.018 & 0.009 & -0.001 & 0.008 & 0.006 & 0.021 & 0.011 \\
         SE & MPL & 0.239 & 0.184 & 0.035 & 0.071 & 0.148 & 0.119 & 0.236 & 0.178 & 0.035 & 0.061 & 0.065 & 0.048 \\
         & & (0.234) & (0.189) & (0.037) & (0.074) & (0.151) & (0.120) & (0.235) & (0.199) & (0.034) & (0.063) & (0.068) & (0.051) \\
         & JM & 0.238 & 0.195 & 0.013 & 0.097 & 0.266 & 0.199 & 0.235 & 0.187 & 0.008 & 0.060 & 0.147 & 0.102 \\
         & & (0.245) & (0.191) & (0.088) & (0.156) & (0.434) & (0.389) & (0.238) & (0.197) & (0.097) & (0.211) & (0.361) & (0.322) \\
         CP & MPL & 0.95 & 0.94 & 0.94 & 0.95 & 0.95 & 0.96 & 0.94 & 0.94 & 0.97 & 0.93 & 0.96 & 0.96 \\
         & & (0.95) & (0.94) & (0.95) & (0.96) & (0.95) & (0.94) & (0.93) & (0.95) & (0.96) & (0.95) & (0.97) & (0.97) \\
         & JM & 0.96 & 0.96 & 0.27 & 0.75 & 0.76 & 0.67 & 0.93 & 0.94 & 0.12 & 0.47 & 0.57 & 0.45 \\
         & & (0.96) & (0.95) & (0.94) & (0.95) & (0.94) & (0.94) & (0.94) & (0.94) & (0.95) & (0.95) & (0.96) & (0.95) \\
         \hline
         && \multicolumn{6}{c}{$n = 1000$, $\Bar{n}_i = 5$, $\pi^E = 0.3$} & \multicolumn{6}{c}{$n = 1000$, $\Bar{n}_i = 20$, $\pi^E = 0.3$}\\
         \hline
         && $\beta$ & $\gamma$ & $\alpha_0$ & $\alpha_1$ & $\alpha_2$ & $\alpha_3$ & $\beta$ & $\gamma$ & $\alpha_0$ & $\alpha_1$ & $\alpha_2$ & $\alpha_3$ \\
         \hline
         Bias & MPL & 0.004 & -0.006 & -0.002 & -0.002 & 0.007 & -0.009 & 0.004 & -0.003 & -0.001 & 0.001 & -0.002 & 0.003 \\
         & JM & 0.009 & -0.011 & 0.001 & -0.001 & 0.034 & -0.054 & 0.006 & 0.010 & -0.005 & 0.004 & -0.001 & 0.004 \\
         SE & MPL & 0.103 & 0.076 & 0.016 & 0.031 & 0.060 & 0.045 & 0.103 & 0.076 & 0.016 & 0.027 & 0.027 & 0.018\\
         & & (0.105) & (0.077) & (0.017) & (0.034) & (0.056) & (0.041) & (0.094) & (0.075) & (0.016) & (0.025) & (0.028) & (0.019) \\
         & JM & 0.103 & 0.083 & 0.009 & 0.053 & 0.139 & 0.096 & 0.103 & 0.080 & 0.004 & 0.027 & 0.060 & 0.040 \\
         & & (0.108) & (0.083) & (0.066) & (0.080) & (0.279) & (0.228) & (0.095) & (0.077) & (0.085) & (0.150) & (0.185) & (0.150) \\
         CP & MPL & 0.94 & 0.95 & 0.89 & 0.94 & 0.98 & 0.99 & 0.96 & 0.95 & 0.95 & 0.97 & 0.93 & 0.94\\
         & & (0.95) & (0.96) & (0.91) & (0.94) & (0.96) & (0.96) & (0.93) & (0.94) & (0.95) & (0.96) & (0.95) & (0.94) \\
         & JM & 0.92 & 0.96 & 0.20 & 0.78 & 0.61 & 0.58 & 0.96 & 0.96 & 0.08 & 0.32 & 0.53 & 0.42 \\
         & & (0.93) & (0.95) & (0.98) & (0.95) & (0.96) & (0.96) & (0.93) & (0.96) & (0.95) & (0.93) & (0.96) & (0.96) \\
         \hline
    \end{tabular}
    \caption{Study 1 (right censoring) regression parameter simulation results; $n$ refers to the total sample size, $\bar{n}_i$ refers to the average number of longitudinal observations, $\pi^E$ refers to the proportion of non-censored observations. For all scenarios summarised in this table, the true value of $\sigma_{\varepsilon}^2 = 0.05^2$.}
    \label{right_cens1}
\end{table}

\begin{figure}
    \centering
    \includegraphics[width=0.40\textwidth]{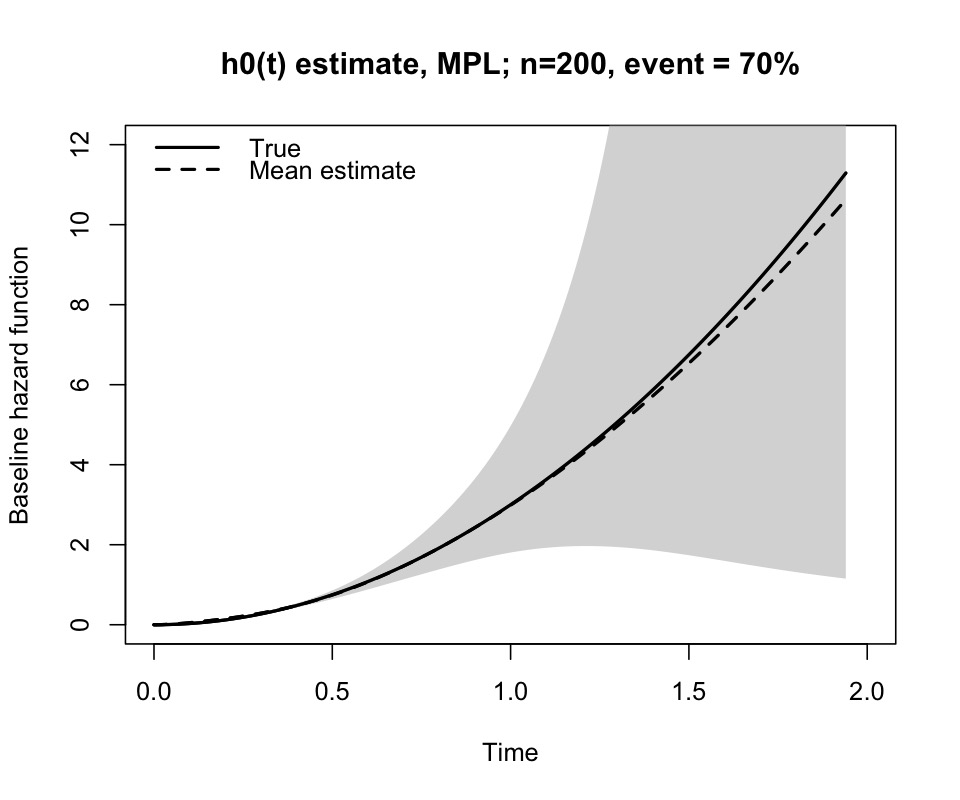}
    \includegraphics[width=0.40\textwidth]{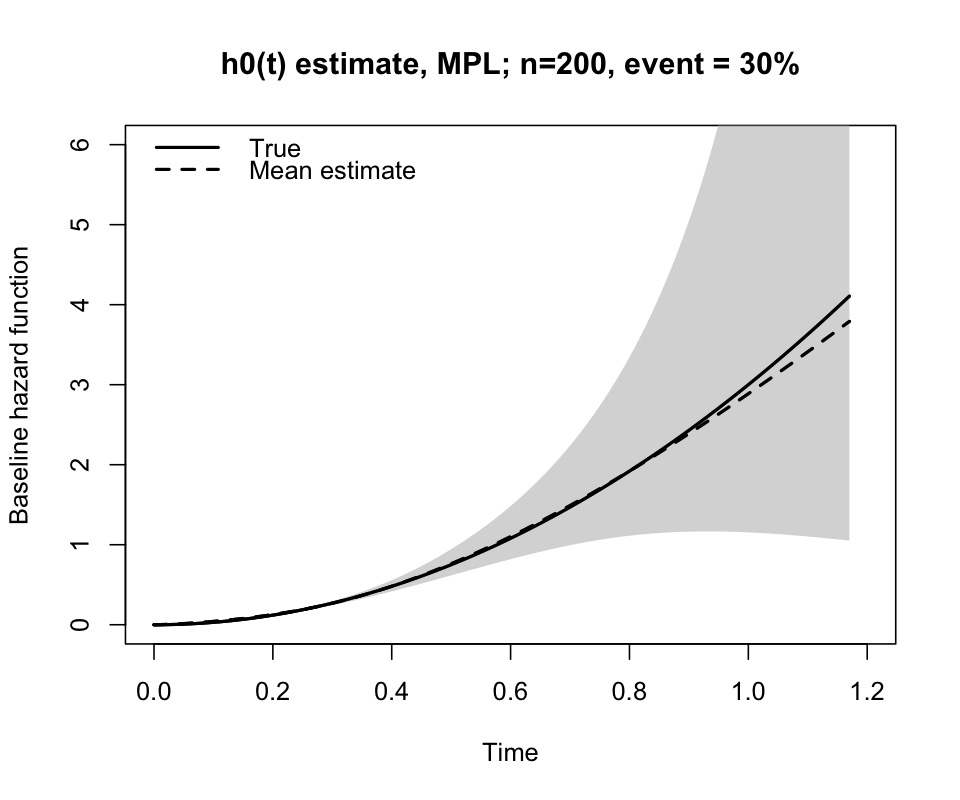}
    \includegraphics[width=0.40\textwidth]{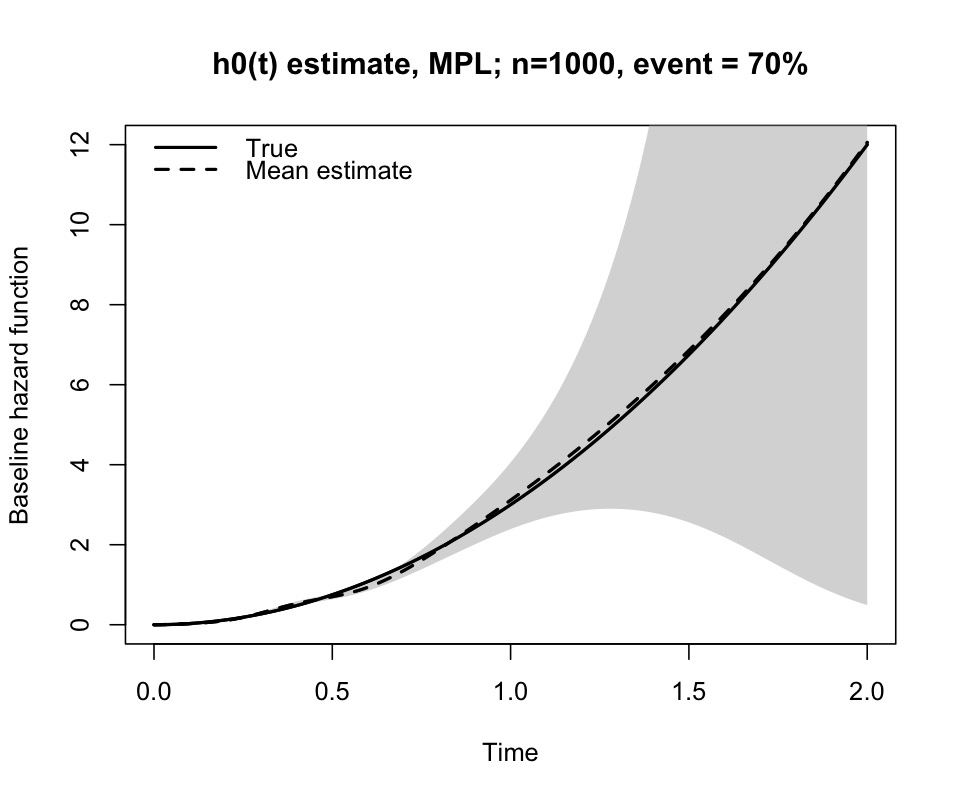}
    \includegraphics[width=0.40\textwidth]{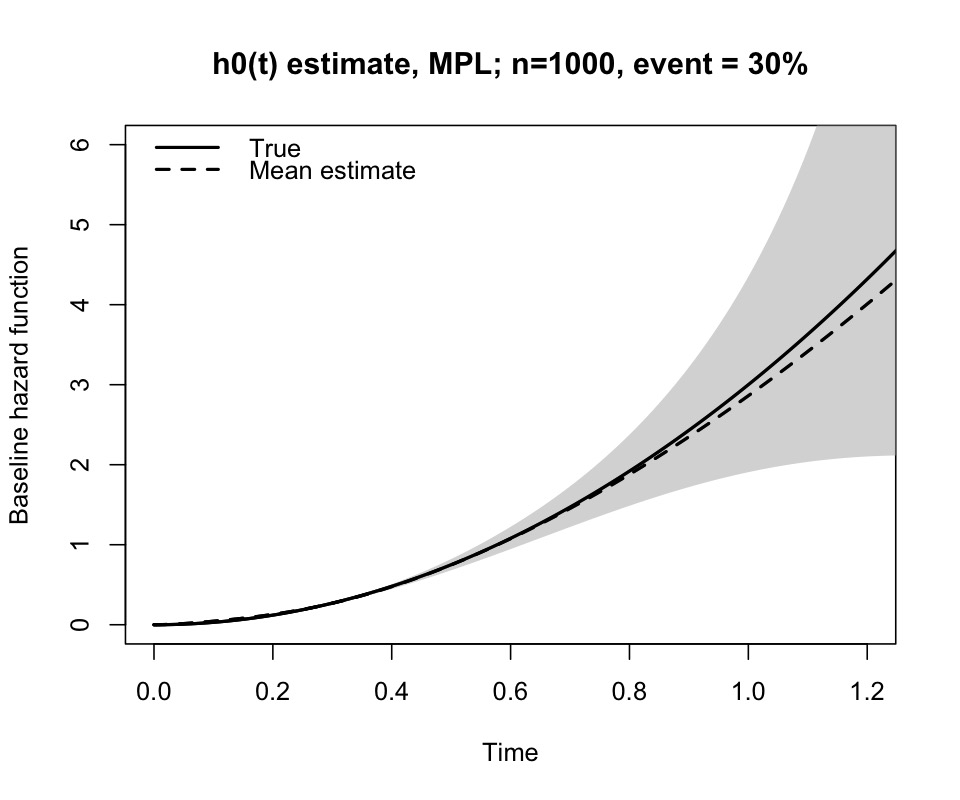}
    \caption{Estimates of the baseline hazard function for $\sigma_{\varepsilon} = 0.05$ and $\bar{n}_i = 20$. The solid line is the true baseline hazard function, the dashed line is the mean estimate and the grey area represents the asymptotic $95\%$ coverage probability.}
    \label{fig:right_cens_h0t}
\end{figure}

\subsection{Study 2a (interval censored) simulation results}

The simulation study results from Study 2a, including the bias, asymptotic and Monte Carlo standard error estimates and the asymptotic and Monte Carlo coverage probabilities for the regression coefficients in the survival and longitudinal models, are summarised in Table \ref{int_cens}. The biases of the MPL estimates of $\beta_1$ and $\beta_2$ are small across the sample sizes and the number of longitudinal observations considered in this study. The MPL asymptotic standard errors for the $\beta$ are reasonably close to the Monte Carlo estimates and the coverage probabilities are reasonable, especially as the sample size and/or number of longitudinal observations increases. In comparison, midpoint imputation produces estimates of $\beta_1$ and $\beta_2$ with large biases and therefore the coverage probabilities are low. The MPL and midpoint imputation estimates of $\gamma$ have comparable biases and coverage probabilities, though the MPL estimates of $\gamma$ have decreasing bias and more reasonable coverage probabilities with larger sample sizes. 

The MPL and midpoint imputation estimates of $\boldsymbol{\alpha}$ have comparable biases, suggesting the midpoint imputation has little effect on the longitudinal part of the joint model. In Study 2a, the longitudinal model included a baseline covariate and an interaction term between this covariate and a function of time, and there is no notably increased error in the MPL estimates of the longitudinal model parameters as a result of this more complex model. Like in Study 1, the MPL asymptotic standard error estimates for $\boldsymbol{\alpha}$ are very close to the Monte Carlo estimates across all scenarios and the coverage probabilities are good, while the midpoint imputation results (fitted using the JM package as in Study 1) give asymptotic standard error estimates for $\boldsymbol{\alpha}$ that are too small and result in coverage probabilities that are consistently below the nominal value of 0.95. 

Figure \ref{fig:int_cens_h0t} shows the spline approximation to the log-logistic baseline hazard function from the MPL method along with the point-wise 95\% confidence intervals 
area based on the asymptotic standard error estimates. Like in Study 1, the spline approximation can estimate the baseline hazard function reasonably well. Compared to the midpoint imputation method, the MPL estimate of the baseline hazard function consistently has lower MISE values, indicating the success of including a penalty term to smooth the baseline hazard function and reduce sensitivity to the choice of knots. The penalty function improves the baseline hazard function estimate (Supplementary Table 2).

\begin{table}
    \centering
    \scriptsize
    \begin{tabular}{ll | ccccccccc }
         \hline
         && \multicolumn{9}{c}{$n = 200$, $\Bar{n}_i = 5$, $\pi^R = 0.3$, $\pi^L = 0.1$, $\pi^I = 0.6$} \\
         \hline
         && $\beta_1$ & $\beta_2$ & $\gamma$ & $\alpha_0$ & $\alpha_1$ & $\alpha_2$ & $\alpha_3$ & $\alpha_4$ & $\alpha_5$ \\
         \hline
         Bias & MPL & 0.030 & -0.001 & -0.037 & -0.002 & -0.001 & 0.005 & -0.001 & 0.005 & -0.011 \\
         & JM-mid &  0.116 & -0.094 & -0.067 & -0.006 & -0.016 & -0.015 & 0.007 & 0.011 & -0.001 \\
         SE & MPL & 0.105 & 0.105 & 0.166 & 0.020 & 0.071 & 0.093 & 0.041 & 0.028 & 0.070 \\
         & & (0.108) & (0.104) & (0.206) & (0.023) & (0.073) & (0.101) & (0.044) & (0.035) & (0.072) \\
         & JM-mid & 0.090 & 0.090 & 0.175 & 0.018 & 0.072 & 0.127 & 0.059 & 0.025 & 0.038 \\
         & & (0.090) & (0.083) & (0.193) & (0.033) & (0.101) & (0.161) & (0.081) & (0.049) & (0.124) \\
         CP & MPL & 0.88 & 0.89 & 0.80 & 0.87 & 0.93 & 0.92 & 0.92 & 0.84 & 0.93 \\
         & & (0.95) & (0.91) & (0.93) & (0.97) & (0.96) & (0.97) & (0.96) & (0.97) & (0.96) \\
         & JM-mid & 0.76 & 0.87 & 0.92 & 0.95 & 0.93 & 0.95 & 0.97 & 0.96 & 0.95 \\
         & & (0.76) & (0.87) & (0.92) & (0.95) & (0.93) & (0.95) & (0.97) & (0.96) & (0.95) \\
         \hline
         && \multicolumn{9}{c}{$n = 200$, $\Bar{n}_i = 20$, $\pi^R = 0.3$, $\pi^L = 0.1$, $\pi^I = 0.6$}\\
         \hline
         && $\beta_1$ & $\beta_2$ & $\gamma$ & $\alpha_0$ & $\alpha_1$ & $\alpha_2$ & $\alpha_3$ & $\alpha_4$ & $\alpha_5$ \\
         \hline
         Bias & MPL & 0.006 & 0.005 & -0.015 & -0.001 & 0.001 & -0.001 & 0.001 & 0.002 & 0.003 \\
         & JM-mid & 0.089 & -0.088 & -0.021 & 0.004 & 0.016 & -0.015 & 0.006 & -0.009 & 0.014 \\
         SE & MPL &  0.107 & 0.108 & 0.171 & 0.020 & 0.056 & 0.041 & 0.017 & 0.028 & 0.070 \\
         && (0.087) & (0.110) & (0.168) & (0.019) & (0.052) & (0.044) & (0.018) & (0.030) & (0.073) \\
         & JM-mid & 0.090 & 0.091 & 0.165 & 0.011 & 0.042 & 0.063 & 0.027 & 0.013 & 0.017 \\
         && (0.078) & (0.099) & (0.168) & (0.061) & (0.152) & (0.106) & (0.054) & (0.083) & (0.146) \\
         CP & MPL & 0.96 & 0.92 & 0.74 & 0.96 & 0.96 & 0.93 & 0.94 & 0.92 & 0.91 \\
         && (0.95) & (0.95) & (0.94) & (0.95) & (0.95) & (0.95) & (0.96) & (0.95) & (0.93) \\
         & JM-mid & 0.85 & 0.77 & 0.92 & 0.31 & 0.60 & 0.70 & 0.64 & 0.26 & 0.22 \\
         && (0.81) & (0.83) & (0.94) & (0.97) & (0.99) & (0.97) & (0.94) & (0.95) & (0.94) \\
         \hline
         && \multicolumn{9}{c}{$n = 1000$, $\Bar{n}_i = 5$, $\pi^R = 0.3$, $\pi^L = 0.1$, $\pi^I = 0.6$}\\
         \hline
         && $\beta_1$ & $\beta_2$ & $\gamma$ & $\alpha_0$ & $\alpha_1$ & $\alpha_2$ & $\alpha_3$ & $\alpha_4$ & $\alpha_5$ \\
         \hline
         Bias & MPL & -0.009 & 0.006 & 0.009 & -0.001 & 0.003 & -0.007 & 0.004 & 0.001 & 0.004 \\
         & JM-mid & 0.108 & -0.115 & -0.061 & 0.001 & 0.006 & -0.005 & -0.001 & 0.001 & -0.008  \\
         SE & MPL & 0.048 & 0.043 & 0.078 & 0.009 & 0.031 & 0.037 & 0.015 & 0.013 & 0.031 \\
         & & (0.048) & (0.054) & (0.068) & (0.011) & (0.034) & (0.038) & (0.015) & (0.013) & (0.032) \\
         & JM-mid & 0.039 & 0.039 & 0.073 & 0.009 & 0.035 & 0.057 & 0.025 & 0.013 & 0.021  \\
         & & (0.042) & (0.038) & (0.071) & (0.018) & (0.067) & (0.079) & (0.040) & (0.026) & (0.088)  \\
         CP & MPL & 0.81 & 0.78 & 0.83 & 0.92 & 0.93 & 0.88 & 0.90 & 0.97 & 0.97  \\
         & & (0.95) & (0.93) & (0.93) & (0.95) & (0.97) & (0.92) & (0.90) & (0.98) & (0.97)  \\
         & JM-mid & 0.25 & 0.18 & 0.85 & 0.68 & 0.70 & 0.87 & 0.78 & 0.68 & 0.40 \\
         & & (0.25) & (0.17) & (0.78) & (0.98) & (0.95) & (0.97) & (0.97) & (0.95) & (0.95) \\
         \hline
         && \multicolumn{9}{c}{$n = 1000$, $\Bar{n}_i = 20$, $\pi^R = 0.3$, $\pi^L = 0.1$, $\pi^I = 0.6$} \\
         \hline
         && $\beta_1$ & $\beta_2$ & $\gamma$ & $\alpha_0$ & $\alpha_1$ & $\alpha_2$ & $\alpha_3$ & $\alpha_4$ & $\alpha_5$ \\
         \hline
         Bias & MPL & -0.022 & 0.009 & 0.004 & -0.001 & -0.001 & 0.001 & 0.001 & 0.002 & 0.004 \\
         & JM-mid & 0.097 & -0.110 & -0.058 & 0.007 & -0.014 & -0.008 & 0.004 & -0.011 & 0.016 \\
         SE & MPL & 0.044 & 0.042 & 0.178 & 0.009 & 0.025 & 0.017 & 0.006 & 0.013 & 0.032 \\
         & & (0.050) & (0.047) & (0.075) & (0.011) & (0.025) & (0.014) & (0.005) & (0.012) & (0.031) \\
         & JM-mid & 0.039 & 0.039 & 0.069 & 0.006 & 0.019 & 0.027 & 0.012 & 0.008 & 0.010 \\
         & & (0.042) & (0.039) & (0.071) & (0.041) & (0.073) & (0.048) & (0.024) & (0.053) & (0.111) \\
         CP & MPL & 0.77 & 0.71 & 1.00 & 0.90 & 0.97 & 0.98 & 1.00 & 0.94 & 0.95 \\
         & & (0.95) & (0.92) & (0.97) & (0.93) & (0.97) & (0.97) & (0.97) & (0.92) & (0.95) \\
         & JM-mid & 0.29 & 0.21 & 0.89 & 0.19 & 0.35 & 0.94 & 0.60 & 0.31 & 0.19 \\
         & & (0.35) & (0.21) & (0.89) & (0.95) & (0.95) & (0.95) & (0.98) & (0.92) & (0.94) \\
    \end{tabular}
    \caption{Study 2 (interval censoring) regression parameter simulation results; $n$ refers to the total sample size, $\bar{n}_i$ refers to the average number of longitudinal observations, $\pi^R$ refers to the proportion of right censored observations. For all scenarios summarised in this table, the true value of $\sigma_{\varepsilon}^2 = 0.1^2$. JM-mid: mid-point imputation (using the \texttt{JM} package).}
    \label{int_cens}
\end{table}


\begin{figure}
    \centering
    \includegraphics[width=0.40\textwidth]{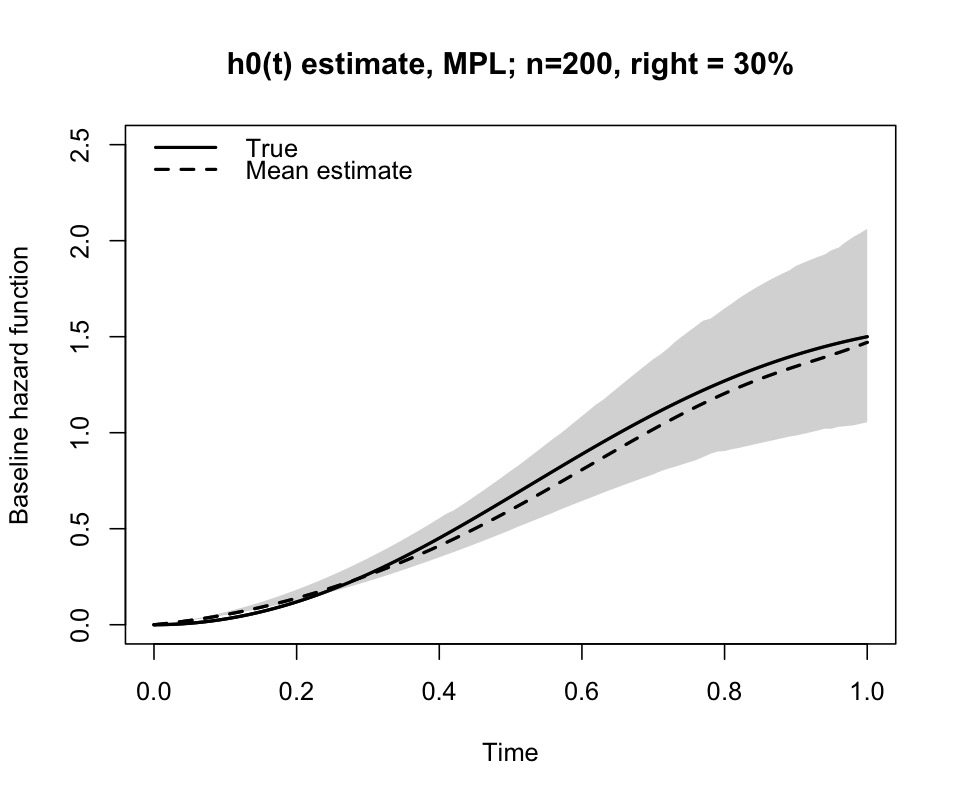}
    \includegraphics[width=0.40\textwidth]{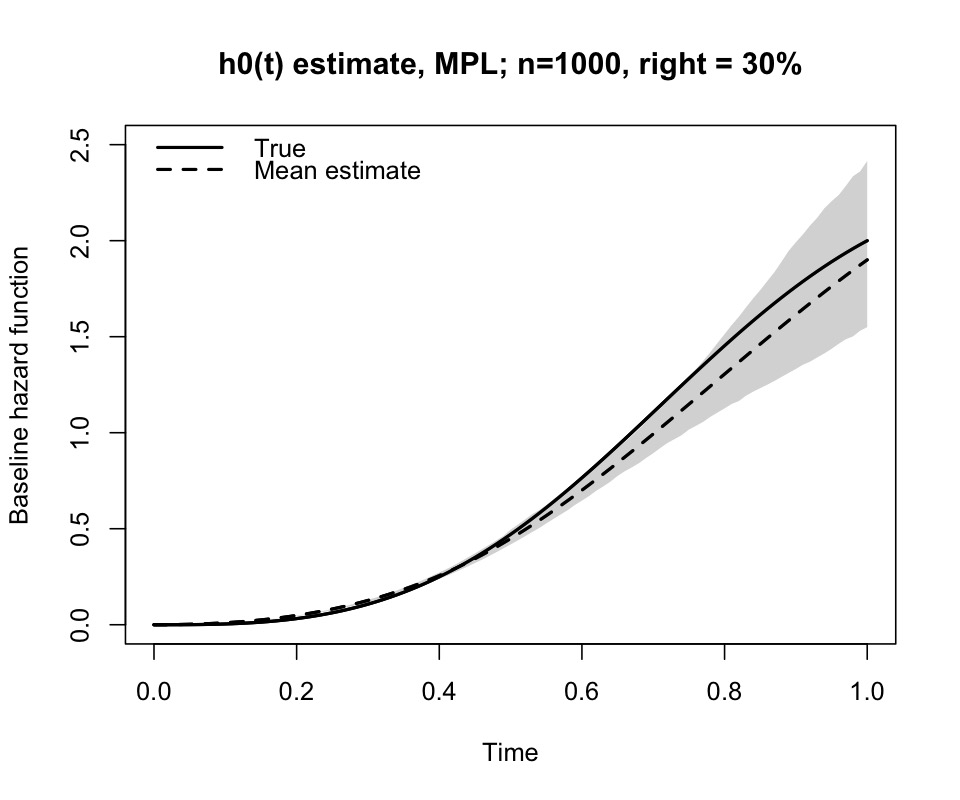}
    \caption{Estimates of the baseline hazard function for $\bar{n}_i = 20$. The solid line is the true baseline hazard function, the dashed line is the mean estimate and the grey area represents the asymptotic $95\%$ coverage probability. }
    \label{fig:int_cens_h0t}
\end{figure}

\subsection{Study 2b simulation results}

The results of Study 2b, comparing our proposed MPL method with the existing method of \cite{chen_2018_TJM_IC}, are presented in Table \ref{chen_mpl_comparison} and Figure \ref{fig:Chen_cumHaz}. Table \ref{chen_mpl_comparison} shows that the bias in the Cox regression model parameters $\beta$ and $\gamma$ is reduced in our MPL approach compared to the approach of \cite{chen_2018_TJM_IC} for both smaller and larger sample sizes. The bias in the longitudinal model parameters $\alpha_0$, $\alpha_1$ and $\alpha_2$ is comparable between the two methods, even with the relatively small number of average longitudinal observations used in this part of the simulation study. Our proposed MPL approach produces asymptotic standard error estimates for Cox and longitudinal regression parameters that are close to the empirical Monte Carlo results from both the MPL and \cite{chen_2018_TJM_IC} methods, without the need to implement a computationally intensive bootstrap as proposed by \cite{chen_2018_TJM_IC}. Figure \ref{fig:Chen_cumHaz} shows that our proposed spline approximation to the baseline (cumulative) hazard function works well in both large and small samples and produces a smooth estimate over the whole follow-up period, compared to the baseline cumulative hazard estimates produced by \cite{chen_2018_TJM_IC}'s method, where estimates are only available at observed censoring times, and which appears to be biased at later timepoints when the sample is smaller.

Results for Study 2c are available in the Supplementary Material (Table 5). These results show that the bias in the Cox regression parameter estimates is reasonable for interval censored data even when the longitudinal model is more complex e.g. using smoothing splines. Standard error estimates and coverage probabilities for the $\boldsymbol{\beta}$ parameters are reasonable, though the performance of the standard error estimate for $\gamma$, the parameter associating the longitudinal covariate with the survival times, is somewhat poorer with this more complex model than in other simulation scenarios we have considered. However, even with this more complex model, our MPL approach is able to obtain good estimates of the mean and individual longitudinal trajectories, indicated by the small mean integrated square error values in Supplementary Table 5.

\begin{table}[h!]
    \centering
    \begin{tabular}{ll | ccccc  }
         \hline
         && \multicolumn{5}{c}{$n = 200$, $\Bar{n}_i = 5$, $\pi^R = 0.25$} \\
         \hline
         && $\beta$ & $\gamma$ & $\alpha_0$ & $\alpha_1$ & $\alpha_2$ \\
         \hline
         Bias & MPL & 0.006 & -0.008 & 0.002 & -0.003 & 0.001  \\
         & Ch &  -0.012 & -0.066 & 0.002 & -0.004 & 0.003  \\
         SE & MPL & 0.108 & 0.272 & 0.029 & 0.028 & 0.072   \\
         & & (0.094) & (0.269) & (0.019) & (0.030) & (0.028)  \\
         & Ch & (0.099) & (0.310) & (0.019) & (0.030) & (0.030) \\
         CP & MPL & 0.97 & 0.95 & 0.99 & 0.95 & 0.99  \\
         & & (0.94) & (0.95) & (0.93) & (0.96) & (0.93) \\
         & Ch & (0.96) & (0.95) & (0.95) & (0.96) & (0.94) \\
         \hline
         && \multicolumn{5}{c}{$n = 1000$, $\Bar{n}_i = 5$, $\pi^R = 0.25$} \\
         \hline
        && $\beta$ & $\gamma$ & $\alpha_0$ & $\alpha_1$ & $\alpha_2$ \\
         \hline
         Bias & MPL & -0.007 & -0.016 & -0.001 & -0.003 & 0.004 \\
         & Ch & -0.027 & -0.064 & -0.004 & -0.002 & 0.010\\
         SE & MPL & 0.046 & 0.106 & 0.009 & 0.013 & 0.013 \\
         & & (0.050) & (0.117) & (0.010) & (0.015) & (0.015) \\
         & Ch & (0.049) & (0.125) & (0.010) & (0.014) & (0.015) \\
         CP & MPL & 0.97 & 0.93 & 0.93 & 0.85 & 0.88 \\
         & & (0.97) & (0.97) & (0.95) & (0.95) & (0.93) \\
         & Ch & (0.92) & (0.92) & (0.92) & (0.97) & (0.89) \\
         \hline
         \hline
    \end{tabular}
    \caption{Study 2b regression parameter simulation results; $n$ refers to the total sample size, $\bar{n}_i$ refers to the average number of longitudinal observations, $\pi^R$ refers to the proportion of right censored observations. For all scenarios summarised in this table, the true value of $\sigma_{\varepsilon}^2 = 0.1^2$.}
    \label{chen_mpl_comparison}
\end{table}

\begin{figure}
    \centering
    \includegraphics[width=0.40\textwidth]{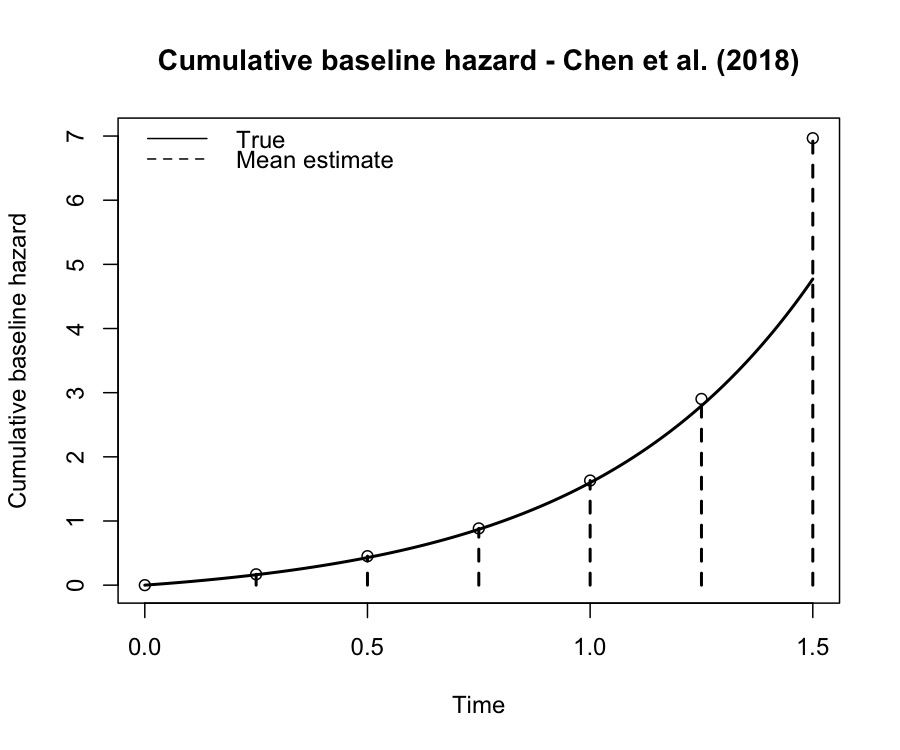}
    \includegraphics[width=0.40\textwidth]{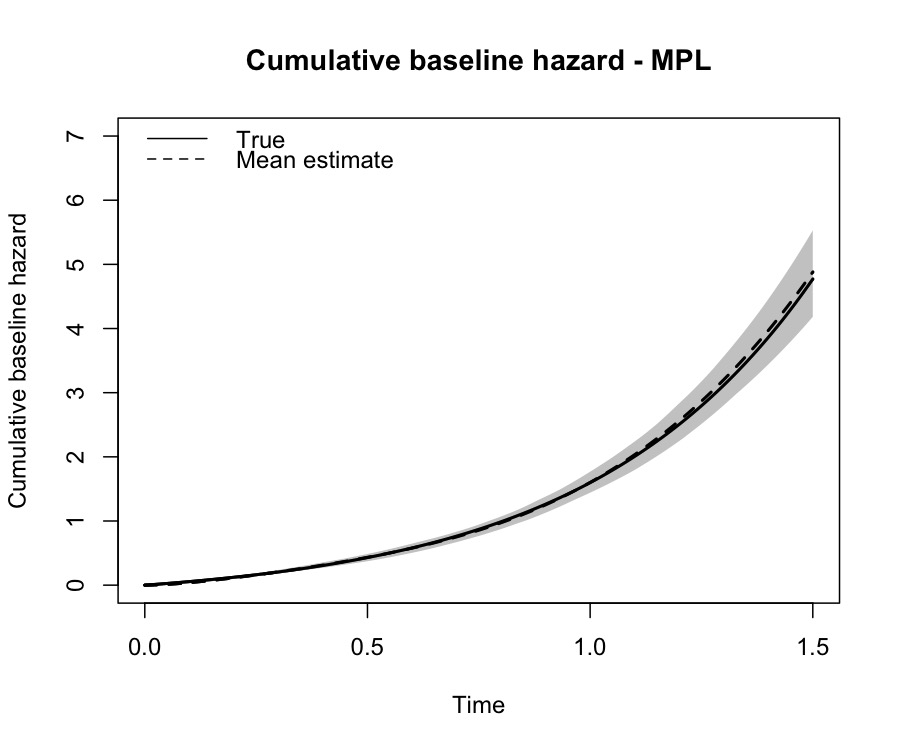}
    \includegraphics[width=0.40\textwidth]{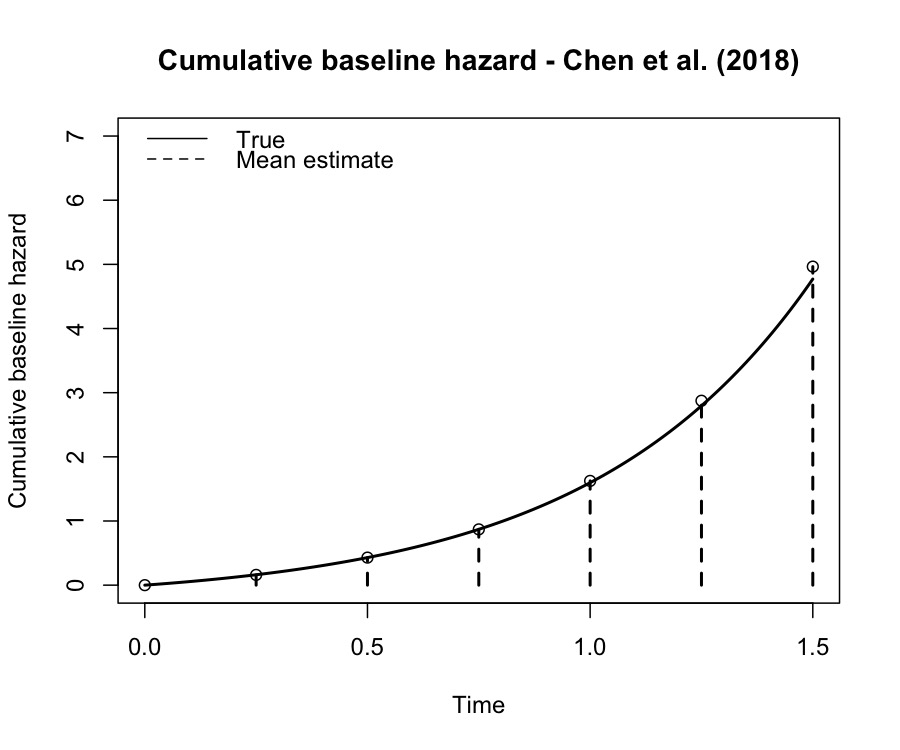}
    \includegraphics[width=0.40\textwidth]{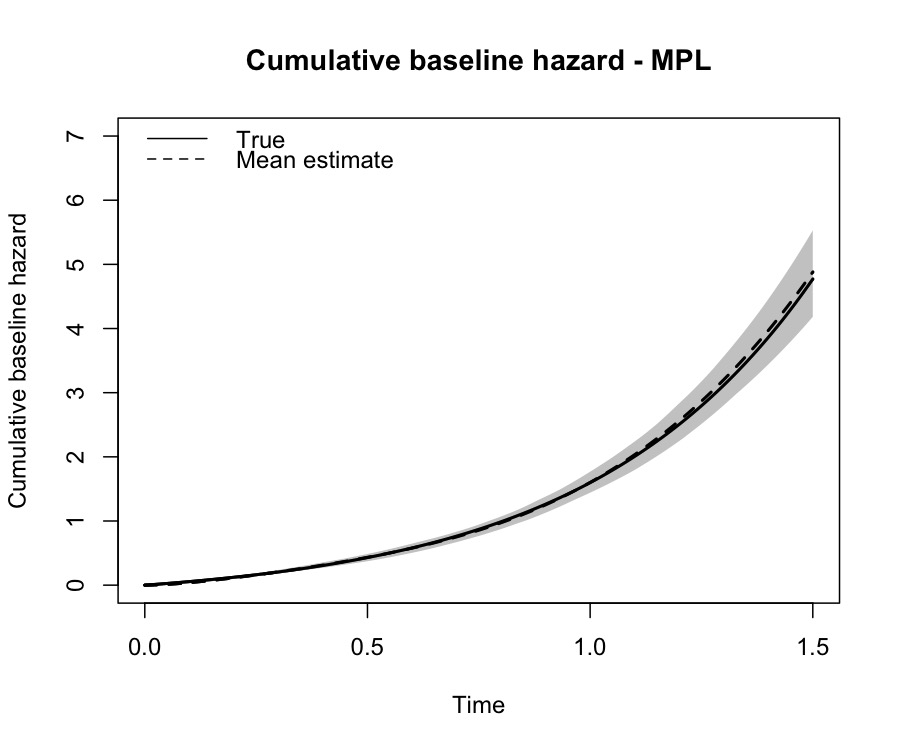}
    \caption{Estimates of the baseline cumulative hazard function for $n=200$ (top row) and $n=1000$ (bottom row) for \cite{chen_2018_TJM_IC} and our proposed MPL method.}
    \label{fig:Chen_cumHaz}
\end{figure}

\section{An application to melanoma clinical trial data}
\label{sec:app}
This section demonstrates the application of the MPL method for fitting a joint model with a partly-interval censored survival outcome, using real data from a melanoma clinical trial. The Anti-PD1 brain collaboration (ABC) trial was the first randomised study that demonstrated the efficacy of immunotherapy in melanoma brain metastases patients.  The study enrolled $n=79$ patients in Australia between 2014 and 2017, investigated treatment of melanoma brain metastases across three cohorts: (A) nivolumab plus ipilimumab on patients without previous local brain therapy ($n=36$), (B) nivolumab monotherapy on patients without previous local brain therapy ($n=27$), and (C) nivolumab monotherapy on patients with a history of local brain therapy failure ($n=16$) \citep{Long_ABC_2018}. 

We proposed to re-analyse the primary outcome of the trial, time to intracranial disease progression 
considering its partly-interval censored nature. All individuals in the trial had this event time either interval-, left-, or right-censored. The time interval in which disease progression occurred 
was denoted as $[t_f, t_p]$, where $t_f$ was the last known progression-free time point (which was 0 for left-censored individual and $>0$ for interval-censored individuals) and $t_p$ was the time point of the scan upon which a diagnosis of disease progression was made. Some individuals did not have a record of $t_p$ but did have a date of death due to intracranial disease progression, $t_d$, and these individuals where also left- or interval-censored in the interval $[t_f, t_d]$. Right-censored individuals had only a record of the last known progression-free time point $t_f$, and hence had an event interval $[t_f, \infty)$.

The longitudinal covariate of interest in this study was the repeated assessments via MRI of the sum of the diameters of intracranial lesions (called lesion size). The average number of longitudinal observations included in the analysis was $\bar{n}_i = 3.4$ (range $1-7$). The trajectories of lesion size of each patient during the study follow-up are displayed in Figure  \ref{fig:observed_lesions}. Baseline (time-fixed) covariates of interest in this analysis in addition to treatment group (A; B; C) included sex (male; female), age at diagnosis ($<50$ years; $\geq 50$ years), previous combined BRAF and MEK inhibitor therapy (yes; no), and lactate dehydrogenase (LDH) blood serum level (normal; high).

\begin{figure}
    \centering
    \includegraphics[width=0.65\textwidth]{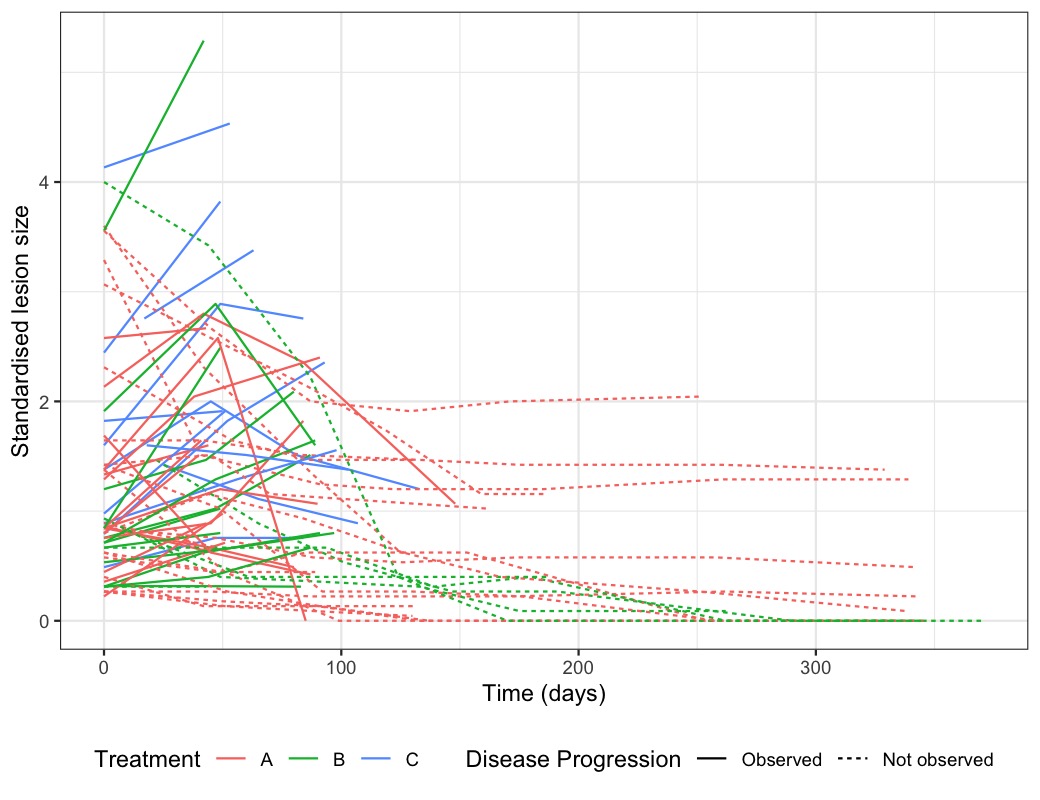}
    \caption{Individual trajectories of lesion sizes (standardised) during the trial follow-up. Line colour indicates treatment group (A, B or C). Line type indicates whether disease progression was observed during the follow-up period or not.}
    \label{fig:observed_lesions}
\end{figure}

To improve model performance, we first standardised the longitudinal variable lesion size by dividing the recorded values by the sample standard deviation $s = 22.5$; hence the value of $z_i(t)$ in the models below should be thought of as standardised 
lesion size, rather than lesion size directly. We considered the following joint model for 
the standardised lesion size measurement and 
time to disease progression:
\begin{align*}
    h_i(t) &= h_0(t) \exp\{\beta_1 x_{i1} + \beta_2 x_{i2} + \beta_3 x_{i3} + \beta_4 x_{i4} + \beta_5 x_{i5} + \beta_6 x_{i6} + \gamma z_i(t)\} \\
    z_i(t) &= \sum_{r = 1}^5 \big[\alpha_r \phi_r(t) \big] + 
    \sum_{r = 1}^3 \big[ \alpha_{r+5} x_{i1}\phi_r(t) \big] + 
    \sum_{r = 1}^4 \big[\kappa_{ir} \phi_r(t) \big]
\end{align*}
where
    $x_{i1} = 1$ if  $i$ as in treatment group A, $0$ otherwise;  
    $x_{i2} = 1$  if $i$ was in treatment group C, 0 otherwise; 
    $x_{i3} = 1$  if $i$ had previous BRAF treatment, 0 otherwise; 
    $x_{i4} = 1$ if $i$ was male, 0  otherwise; 
    $x_{i5} = 1$ if $i$ was diagnosed before the age of 50, 0 otherwise; 
    and $x_{i6} = 1$ if $i$ had a high LDH serum level, 0 otherwise.
%
Each $\phi_r(t)$ for $r = 1, \dots 5$ was a B-spline basis function, with external knots located at the minimum and maximum of the longitudinal follow-up times and one internal knot placed at the $50th$ percentile of the longitudinal follow-up times. An interaction term between $x_{i1}$ and $\phi_r(t)$ was included for the first three basis functions so that the mean trajectory of lesion size for participants who received nivolumab plus ipilimumab was allowed to vary from that of participants who received nivolumab alone (with or without a history of local brain therapy failure); and each $\kappa_{ir}$ was a normally distributed random effect. 

Results from the fitted joint model are shown in Table \ref{tab:mel_results}. According to the MPL model, individuals who received nivolumab plus ipilimumab (Group A) had significantly lower risk of disease progression compared to those who received nivolumab alone and were local brain therapy naive (Group B) (HR: $0.366$, $p = 0.003$). Individuals who had a history of local brain therapy failure and received nivolumab alone (Group C) had an increased, but not significantly different, risk of disease progression compared to those who received nivolumab alone and were local brain therapy naive (Group B) ($p = 0.255$). Larger lesion sizes were significantly associated with increased risk of disease progression (HR for a standard deviation unit increase: 1.895, $p < 0.001$). Previous BRAF treatment, having a high LDH serum level and being male were all non-significantly associated with increased risk of disease progression, while a younger age at diagnosis was non-significantly associated with a lower risk of disease progression. Results from the midpoint imputation model give comparable results for the time-fixed covariates, but substantially different results for the association between lesion size and risk of disease progression, with the midpoint imputation model results suggesting that there is no significant association between larger lesions and higher risk of disease progression. 

\begin{table}[]
    \scriptsize
    \centering
    \begin{tabular}{ll | ccc | ccc}
         \hline
         && \multicolumn{3}{c}{Penalised likelihood model} & \multicolumn{3}{c}{Midpoint imputation model} \\ 
         \hline
         Covariates & Level/Comparison & HR & 95\% CI & p-value & HR & 95\% CI & p-value \\
         \hline
          Trt Group & A vs B & 0.366 & 0.178, 0.753 & 0.006 & 0.366 & 0.180, 0.743 & 0.005\\
          & C vs B & 1.850 & 0.641, 5.336 & 0.255 & 1.890 & 0.739, 4.835 & 0.184\\
          Previous BRAF Trt & Yes vs No & 1.508 & 0.677, 3.361 & 0.315 & 1.426 & 0.690, 2.947 & 0.338\\
          Gender & Male vs Female & 1.236 & 0.642, 2.381 & 0.525 & 1.343 & 0.667, 2.708 & 0.409\\
          Age at Diagnosis & $<50$ vs $\geq 50$ & 0.743 & 0.324, 1.701 & 0.482 & 0.880 & 0.423, 1.829 & 0.731 \\
          LDH Level & High vs Normal & 1.672 & 0.848, 3.296 & 0.138 & 1.608 & 0.813, 3.177 & 0.172 \\
          Standardised Lesion Size & +1SD (+22.5cm) & 1.895 & 1.316, 2.728 & $<0.001$ & 1.035 & 0.752, 1.424 & 0.834 \\
          \hline
    \end{tabular}
    \caption{Results from the fitted joint model. Results from the penalised likelihood model fitted to the interval censored data are on the left, and results from a joint model fitted using midpoint imputation with the \texttt{JM} package are on the right.}
    \label{tab:mel_results}
\end{table}

\begin{figure}
    \centering
    \includegraphics[width=0.40\textwidth]{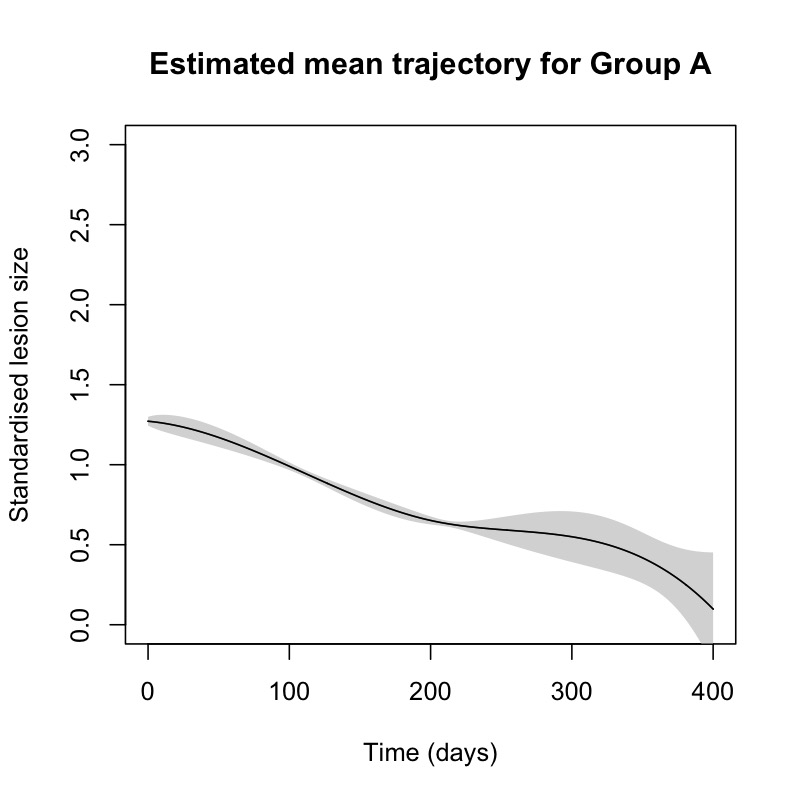}
    \includegraphics[width=0.40\textwidth]{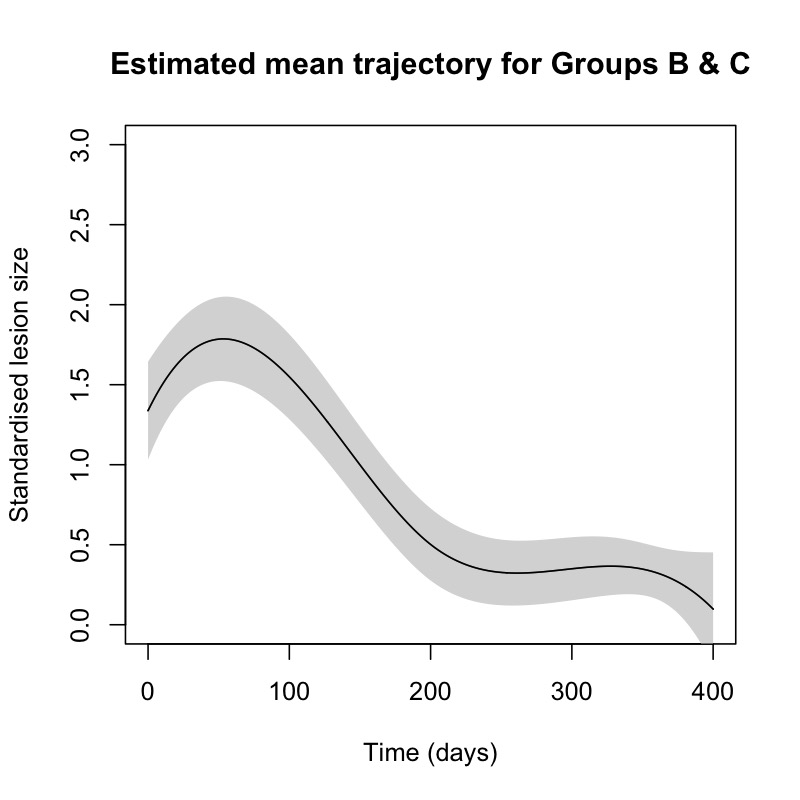}
    \caption{Estimated mean longitudinal trajectories of standardised lesion size for Treatment Group A (nivolumab plus ipilimumab) and Treatment Groups B and C (nivolumab alone) with point-wise 95\% confidence intervals.}
    \label{fig:mean_long_trajectory}
\end{figure}

Figure \ref{fig:mean_long_trajectory} shows the estimated mean longitudinal trajectory of standardised lesion size over the follow-up period for Group A (nivolumab plus ipilimumab) and Groups B and C (nivolumab alone) along with their point-wise 95\% confidence intervals. The estimated mean trajectories indicate that although these two groups started out with similar mean lesion sizes, those who received nivolumab plus ipilimumab experienced, on average, a steady decline in lesion size over the follow-up, while those who receieved nivolumab alone had, on average, an increase in their lesion size early in the trial period. Note that in the longitudinal model specified above, parameters $\alpha_6$, $\alpha_7$ and $\alpha_8$ represent the interaction between the B-spline term for time and the treatment group. A test of significance of these three parameters produced $p$-values of 0.7171, $<0.001$ and $<0.001$ respectively, indicating that the spline coefficients for the basis functions $\phi_2(t)$ and $\phi_3(t)$ significantly differ between the groups.

\begin{figure}
    \centering
    \includegraphics[width=0.40\textwidth]{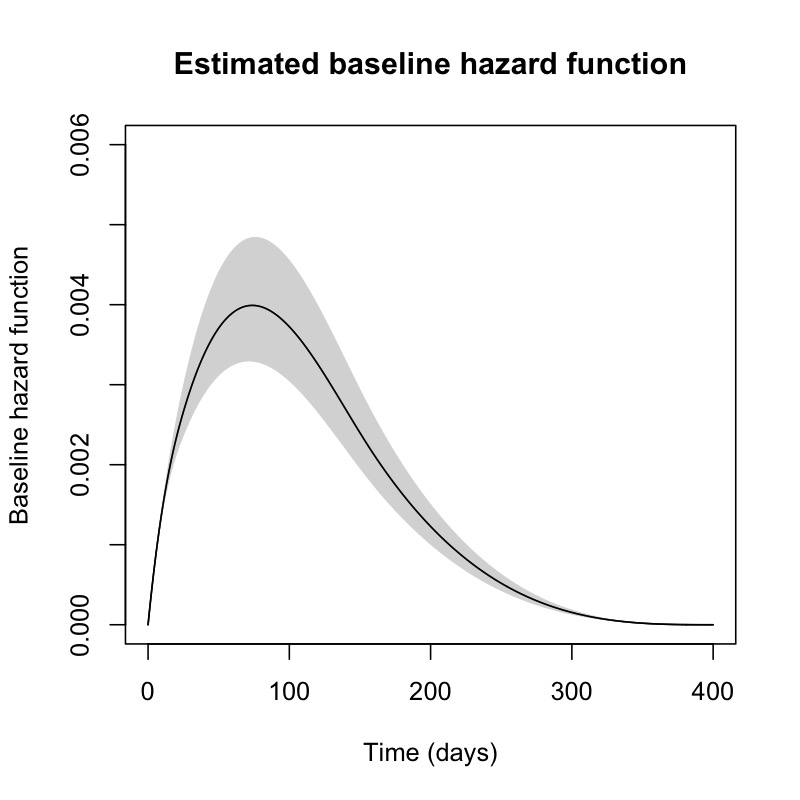}
    \includegraphics[width=0.40\textwidth]{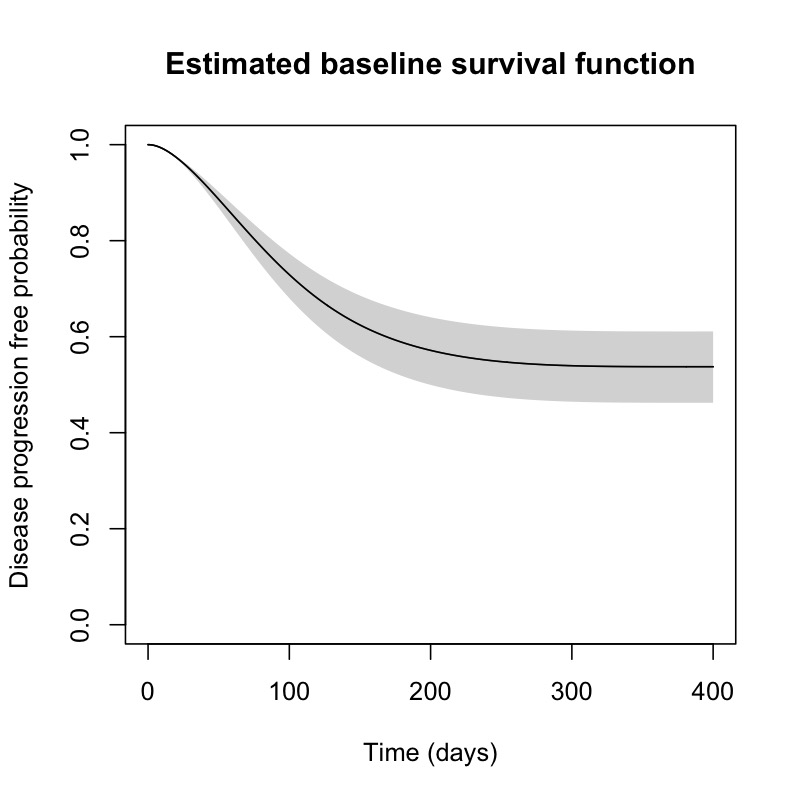}
    \caption{Estimated baseline hazard function (left) and baseline survival function (right) with 95\% point-wise confidence intervals in grey.}
    \label{fig:baseline_plots}
\end{figure}

Figure \ref{fig:baseline_plots} shows the estimated baseline hazard function and baseline survival function obtained from the fitted joint model using M-spline approximation. The grey area on both plots shows the point-wise 95\% confidence intervals. From the baseline hazard, it is clear that the risk of disease progression increases sharply over the first approximately 90 days, after which it decreases and after 300 days is very close to zero. Similarly, the baseline survival function shows that the sharpest decrease in probability of remaining free of disease progression is early in the follow-up period, with the risk of disease progression stabilising later in the follow-up. 

\begin{figure}
    \centering
    \includegraphics[width=0.65\textwidth]{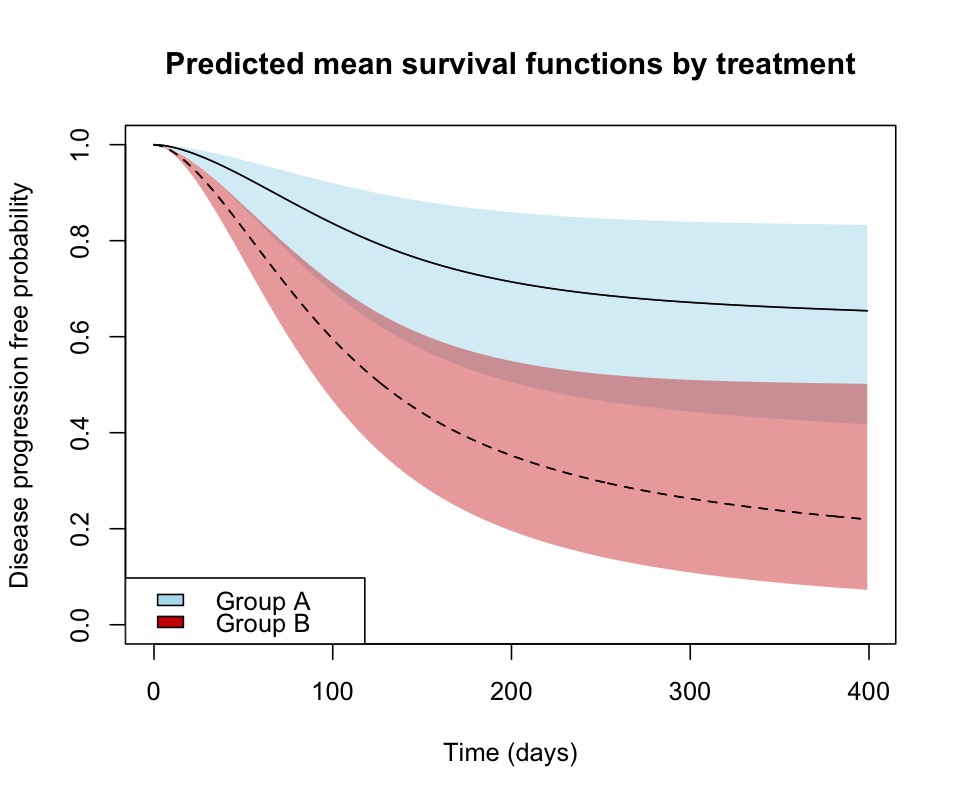}
    \caption{Predicted survival functions showing the probability of being free of disease progression over time for Group A (nivolumab plus ipilimumab, blue) and Group B (nivolumab alone, red), with their 95\% point-wise confidence intervals.}
    \label{fig:predicted_surv}
\end{figure}

The proposed MPL method for fitting the joint model allows for the prediction of survival probabilities for combinations of covariates, along with their 95\% point-wise confidence intervals. Figure \ref{fig:predicted_surv} compares the mean predicted survival functions for those who received nivolumab plus ipilimumab (Group A) and those who received nivolumab alone (Group B), with all other covariates set to zero (i.e. with no previous BRAF treatment, female gender, over the age of 50 at diagnosis and normal LDH levels). When computing these predicted survival functions, the lesion size trajectory for Group A was assumed to be the mean trajectory shown on the left hand plot in Figure \ref{fig:mean_long_trajectory}, and the lesion size trajectory for Group B was assumed to be the trajectory shown on the right hand plot in Figure \ref{fig:mean_long_trajectory}. 
The confidence intervals around these predicted survival functions account for uncertainty in the estimates of the fixed effects parameters in the model but do not account for the random effects, and hence should be interpreted as the confidence intervals for the mean or expected survival functions for these two populations. It is clear that individuals receiving nivolumab plus ipilimumab are expected to have a higher probability of being free of disease progression compared to those who received nivolumab alone. After one year of follow-up, the estimated probability of being free of disease progression for those who were treated with nivolumab plus ipilimumab was 0.66 (95\% CI: 0.43, 0.83), compared to 0.23 (95\% CI: 0.08, 0.50) for patients who received nivolumab alone. 

As our proposed MPL method includes direct estimation of the random effect values for individuals in the sample to which the joint model is fitted, it is also straightforward to make individual-level predictions about survival probabilities over time. To illustrate this, we consider two individuals from the ABC trial cohort. The first (refereed as individual A thereafter) was in Treatment Group C, had previous BRAF treatment, was female and over 50 years of age at the time of diagnosis, and had a normal LDH level. Individual A experienced a disease progression between $t = 65$ days and $t = 107$ days. The second (refereed as individual B) was in Treatment Group A, had not had previous BRAF treatment, was female, over 50 years old at the time of diagnosis, and had normal LDH levels. Individual B was right censored at $t = 127$ days. Figure \ref{fig:individual_predictions} shows the predicted longitudinal trajectories and progression-free survival probabilities for these two individuals across their respective follow-up times. 

\begin{figure}
    \centering
    \includegraphics[width=0.40\textwidth]{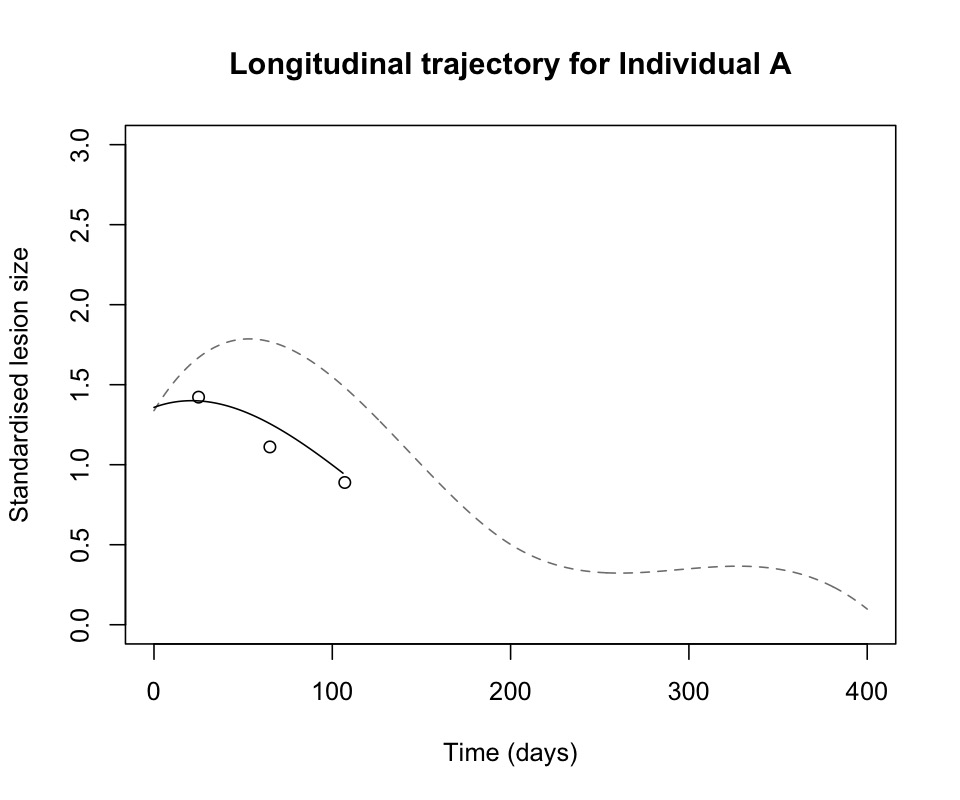}
    \includegraphics[width=0.40\textwidth]{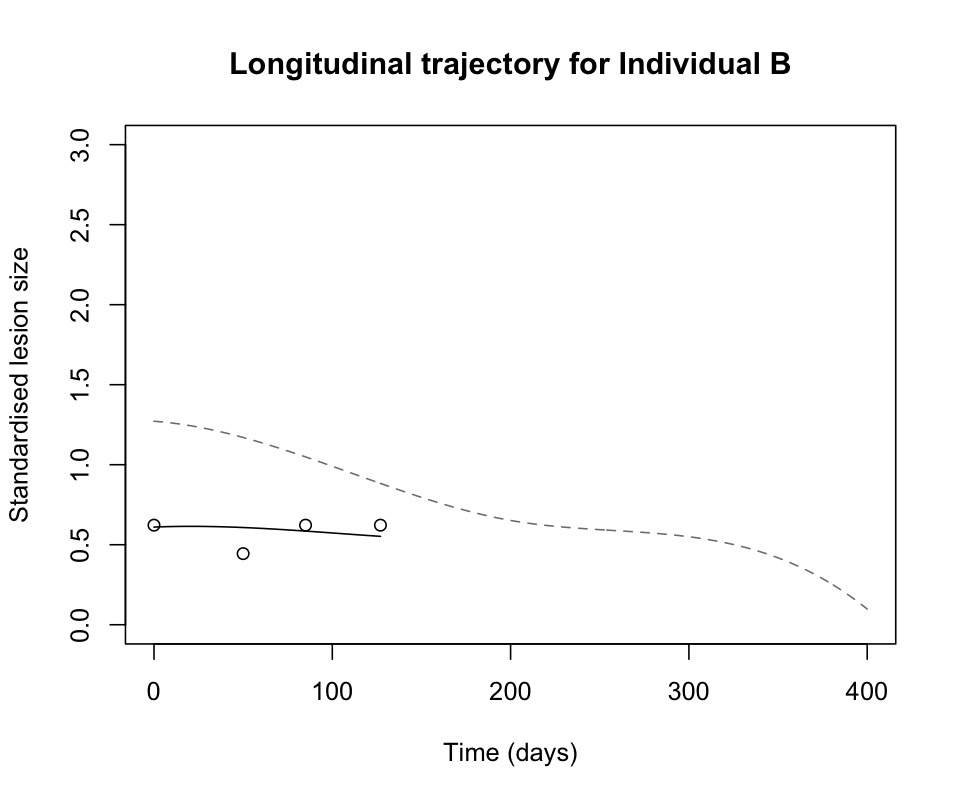}
    \includegraphics[width=0.40\textwidth]{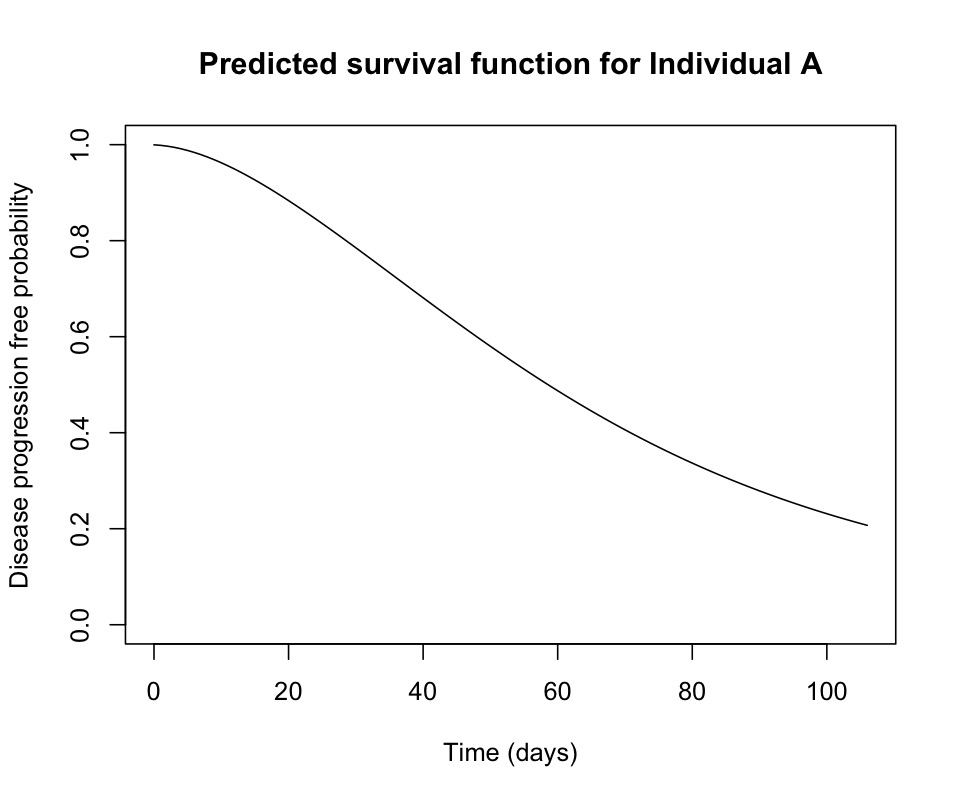}
    \includegraphics[width=0.40\textwidth]{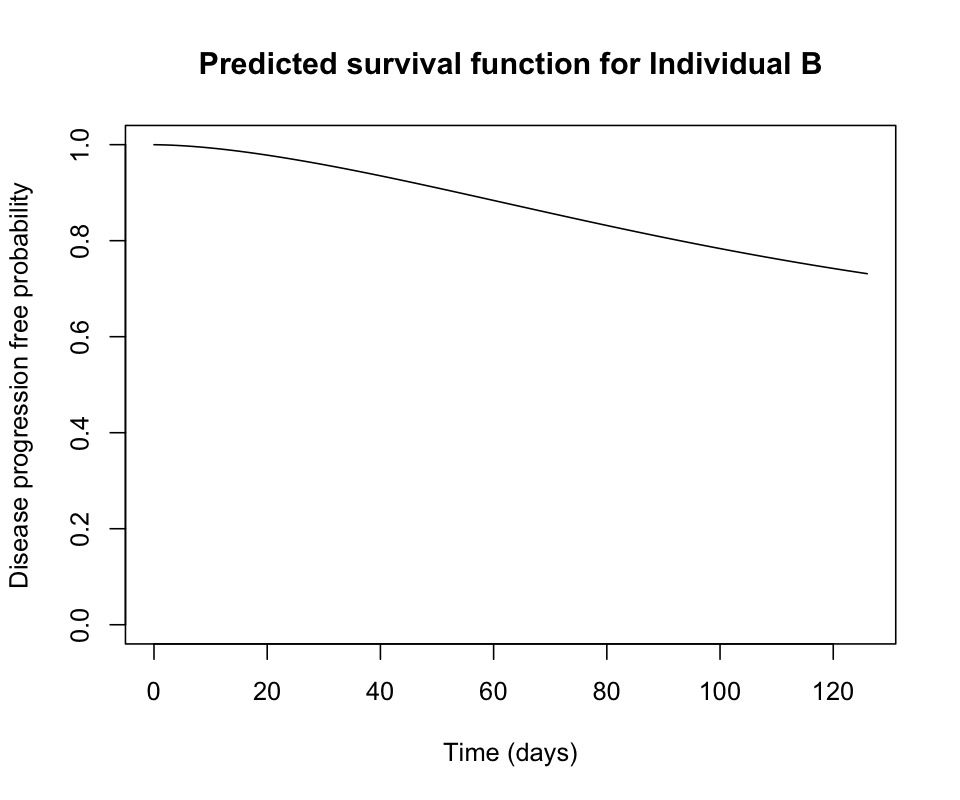}
    \caption{Predicted individual longitudinal trajectories and survival probabilities for Individual A (left) and Individual B (right). Plots of longitudinal trajectories show the predicted individual trajectories (solid line), the mean treatment group trajectories from which the individuals deviate (dashed line) and individual observations of standardised lesion sizes (open circles). }
    \label{fig:individual_predictions}
\end{figure}

For Individual B in particular, as they were right censored with no disease progression at $t = 127$ days of follow-up, a natural question might be what their probability of continuing to be free of disease progression is. For instance, it may be of interest to estimate the probability for individual B to remain disease-progression free after one year, given their observed covariates, predicted longitudinal trajectory, and their disease progression free survival to $t = 127$. This is equivalent to making a prediction of their conditional survival probability \citep{rizop_dynpred}, which can be expressed as $\pi_i(u | t) = S(u)/S(t)$, where $t$ is the latest known event-free time and $u$ is a future time point of interest. For Individual B, we might be interested in $\pi_B(365 | 127) = S(365)/S(127)$. From our fitted joint model, we predict that $\hat{\pi}_B(365 | 127) = 0.61/0.73 = 0.84$, i.e. the estimated probability of Individual B surviving free of disease progression to $t = 365$ days given their survival to $t = 127$ days is 84\%. \cite{rizop_dynpred} considers in detail the issues of making these types of dynamic predictions for individuals outside of the original joint model sample and obtaining standard error estimates and confidence intervals for these predictions in the case of joint models fitted to right censored data. In future research, we intend to extend these methods to joint models fitted to interval censored data using our proposed MPL method.

\section{Discussion and concluding remarks}
\label{sec:concl}

In this paper, we have developed a new penalized likelihood-based estimation method for joint modeling of partly interval-censored event times and error-contaminated longitudinal covariates. The latter model contains random effects. Our method differs from the typical approach of maximising a joint likelihood using an EM algorithm where random effects are treated as missing data (e.g. \cite{Rizo_2012}), in that we estimate all the parameters, including the random effects (treated as random parameters). We obtain a smooth approximation to the baseline hazard function through the inclusion of a penalty term in our joint likelihood function. Smoothing parameters and the variance components of the random effects and measurement error distributions are estimated using an REML approach, which differs from existing methods which obtain MLE estimates of the variance components.  

In our simulation studies, we compared our method with the R \texttt{JM} package, and found that for right-censored event times, our approach generally produced lower biases and better coverage probabilities. If the event times are partly interval-censored, our method generally performs better than using middle-point imputation, and also performs favourably compared to the method proposed by \cite{chen_2018_TJM_IC}. In particular, for study designs where there are a large number of unique event monitoring times (e.g. the ABC trial), our proposed spline approximation to the baseline hazard function using $m \ll n$ basis functions may be computationally advantageous compared to the step-wise estimate of the baseline cumulative hazard used by \cite{chen_2018_TJM_IC}, which requires an additional parameter to be estimated for each unique censoring time. Additionally, this paper proposes asymptotic variance estimators that do not rely on the bootstrap method and consequently significantly boost computational efficiency compared to existing methods for interval censored joint models \citep{chen_2018_TJM_IC, yi_2020_IC_TVC, wu_2020_multivariate_IC_TVC}.
An application of our proposed method to melanoma clinical trial data illustrates the utility of our method for making population and individual-level survival predictions using data sets with interval-censored event times. 

The method proposed in this paper builds on previous methods for including time-varying covariates in a Cox model fitted to partly-interval censored data \citep{WeMa23} by allowing for measurement error in the longitudinal observations and by incorporating individual trajectories via random effects. An important application of joint models is to make dynamic predictions, where individual-level survival predictions can be updated over time as observations of the longitudinal trajectory accrue \citep{proust_lima_2009, taylor_dynPred_2013, rizop_dynpred}. As such, a focus of our future research will be to develop dynamic prediction methods appropriate to partly-interval censored data using this proposed MPL method for fitting joint models, including the development of predictive discrimination and calibration methods. Predictive discrimination and calibration measures have been developed for interval-censored survival data (see \cite{Kim2022}), but to our knowledge there are no existing methods available which incorporate both time-varying covariates and interval-censored event times.

Our investigation of the proposed MPL approach in this paper has some limitations. Firstly, although extending our method to incorporate multiple longitudinal covariates simultaneously would be straightforward, we have limited our scope to a single longitudinal covariate and did not consider a scenario with multiple longitudinal covariates in our simulation studies. Several previous authors have considered extensions of a joint model to include multiple longitudinal covariates under right censoring, often adopting a Bayesian framework for estimation \citep{rizo_2011_BayesianSP_JM, andrin_dynpred_2014, tang_bayesian_JM_2014, baghfalaki_BayesianJM_2014}. In particular, with multiple longitudinal covariates there may be correlated random effects across the covariates, and in such a case our proposed marginal likelihood method for obtaining the variance components would need to be extended to account for this correlation. Existing methods for fitting a joint model to interval censored data \citep{chen_2018_TJM_IC, yi_2020_IC_TVC, wu_2020_multivariate_IC_TVC} have likewise not investigated model performance with more than one longitudinal covariate, although \cite{wu_2020_multivariate_IC_TVC} does consider multivariate event time outcomes. Furthermore, we have not investigated the robustness of our proposed MPL approach to the assumption of normality in the random effects or measurement error. Previous research has indicated that joint models with a single longitudinal covariate may be robust to mis-specification of the random effects distribution \citep{song2002semiparametric, rizo_RE_misspec_2008}. However it is unclear whether this holds in the case of partly-interval censoring, and furthermore this robustness may be reduced where multiple longitudinal covariates are included in the joint model \citep{hickey_JM_review_2016}. As such, future investigation of the performance of the proposed MPL approach under more complex scenarios, e.g. with multiple longitudinal covariates, correlated random effects and/or mis-specified distributions, may be warranted. 

Finally, we note that long-term follow-up data of participants in the ABC trial indicated that a proportion had not experienced disease progression after several years. In the context of this study, being ``cured" (i.e. being un-susceptible to disease progression) is biologically plausible, indicating that a mixture cure joint model may be appropriate for this data. Some authors have considered mixture cure joint models in the context of right censored data (e.g. \cite{Kim_Cure_JM_2013}, \cite{Barbieri_Cure_JM_2020}), but to our knowledge there are no methods available for fitting a joint model to interval-censored data with a cured fraction. However, the penalised likelihood approach used here has previously been applied to a mixture cure Cox model (with no time-varying covariates) \citep{WebbMa_Cure_2022}. In the future we intend to extend this joint modelling approach to accommodate partly-interval censored survival data with a cured fraction.

\newpage

\section*{Data availability statement}

Data analysed in this manuscript may be accessed upon reasonable request by contacting the corresponding author. 

\section*{Conflict of interest}

The authors have no conflict of interest to declare.

\newpage

\bibliographystyle{abbrvnat}

\bibliography{references,Otherref}

\newpage

\section{Supplementary Material}

\appendix

\section{Details of score vector and Hessian matrix elements for Newton-MI algorithm}

\subsection{Details for updating estimate of $\beta$}

The estimate of $\boldsymbol{\beta}$ is updated at each iteration using the formula given by \ref{betaupdat}. In this formula, both $\mathbf{C}$ and $\mathbf{A}$ are $n \times n$ diagonal matrices with diagonal vectors $\mathbf{c}$ and $\mathbf{a}$ respectively. For vector $\mathbf{a}$, the $i$-th element is given by
$$
a_i = (\delta_i + \delta_i^R) H(t_i) - \delta_i^L \frac{H(t_i) S(t_i)}{1 - S(t_i)} + \delta_i^I \frac{S(t_i^L)H(t_i^L) - S(t_i^R)H(t_i^R)}{S(t_i^L) - S(t_i^R)}
$$
and for vector $\mathbf{c}$, the $i$-th element is given by 
\begin{align*}
    c_i = & (\delta_i + \delta_i^R) H(t_i) - \delta_i^L  \frac{H(t_i) S(t_i) - H^2(t_i) S(t_i)}{1 - S(t_i)} +  \delta_i^L  \frac{(H(t_i) S(t_i) )^2}{(1 - S(t_i))^2} \\
    & + \delta_i^I \frac{S(t_i^L)H(t_i^L) - S(t_i^R)H(t_i^R)}{S(t_i^L) - S(t_i^R)} + \delta_i^I \frac{S(t_i^L)S(t_i^R)(H(t_i^L) - H(t_i^R)^2}{(S(t_i^L) - S(t_i^R))^2}
\end{align*}

\subsection{Details for updating estimate of $\gamma$}

The estimate of $\boldsymbol{\gamma}$ is updated at each iteration using the formula given in \ref{gam_update}. In this formula, the vector $\mathbf{g}_{\boldsymbol{\gamma}}$ is of length $q$, with $r$-th element
\begin{align*}
    \mathbf{g}_{\boldsymbol{\gamma}, r} = & \sum_{i = 1}^n \delta_i z_{ir}(t_i) - (\delta_i + \delta_i^R) e^{\mathbf{x}_i ^{T}\boldsymbol{\beta}
    } A_{ir}(t_i) + \delta_i^L e^{\mathbf{x}_i ^{T}\boldsymbol{\beta}} \frac{A_{ir}(t_i) S(t_i)}{1-S(t_i)} \\
    &- \delta_i^I e^{\mathbf{x}_i ^{T}\boldsymbol{\beta}} \frac{S(t_i^L) A_{ir}(t_i^L) - S(t_i^R)A_{ir}(t_i^R)}{S(t_i^L) - S(t_i^R)}
\end{align*}
where $A_{ir}(t) = \int_0^{t} h_0(s) z_{ir}(s) \exp(\mathbf{z}_i(s)^{\top}\boldsymbol{\gamma}) ds$. The matrix $\mathbf{H}_{\gamma}$ has dimensions $q \times q$ with the $(r,s)$-th element equal to
\begin{align*}
    \mathbf{H}_{\gamma, r, s} = & \sum_{i = 1}^n e^{\mathbf{x}_i^{\top} \boldsymbol{\beta}} \bigg\{ - (\delta_i + \delta_i^R)  A_{irs}(t) + \delta_i^L \frac{S(t)A_{irs}(t)}{1 - S(t)} - \delta_i^L e^{\mathbf{x}_i^{\top} \boldsymbol{\beta}} \frac{A_{ir}(t)A_{is}(t)  S(t)}{(1-S(t))^2} \\
    & - \delta_i^I \frac{S(t_i^L)A_{irs}(t_i^L) - S(t_i^R)A_{irs}(t_i^R)}{S(t_i^L) - S(t_i^R)} \\
    & - \delta_i^I e^{\mathbf{x}_i^{\top} \boldsymbol{\beta}} \frac{S(t_i^L)S(t_i^R)(A_{ir}(t_i^L) - A_{ir}(t_i^R)) (A_{is}(t_i^L) - A_{is}(t_i^R))}{(S(t_i^L) - S(t_i^R))^2} \bigg\}
\end{align*}
where $A_{irs}(t) = \int_0^{t} h_0(s) z_{ir}(s) z_{is}(s) \exp(\mathbf{z}_i(s)^{\top}\boldsymbol{\gamma}) ds$.

\subsection{Details for updating estimate of $\theta$}

The estimate of $\boldsymbol{\theta}$ is updated at each iteration using the formula given in \ref{theta_update}. In this formula, $\mathbf{g_{\boldsymbol{\theta}}}$ is a vector of length $m$, with the $u$-th element equal to
\begin{align*}
    \mathbf{g}_{\boldsymbol{\theta}, u} = & \sum_{i = 1}^n \bigg( \delta_i \psi_u(t_i) - (\delta_i + \delta_i^R) e^{\mathbf{x}_i ^{T}\boldsymbol{\beta}} \Psi_{ui}^*(t_i) + \delta_i^L e^{\mathbf{x}_i ^{T}\boldsymbol{\beta}} \frac{ \Psi_{ui}^*(t_i) S(t_i)}{1-S(t_i)} \\
    &  - \delta_i^I e^{\mathbf{x}_i ^{T}\boldsymbol{\beta}} \frac{\Psi_{ui}^*(t_i^L) S(t_i^L) - \Psi_{ui}^*(t_i^R) S(t_i^R)}{S(t_i^L) - S(t_i^R)}\bigg) - 2 \lambda_{\theta} [\mathbf{G}_{\theta} \boldsymbol{\theta}]_u.
\end{align*}
The diagonal elements of $\mb{S}_{\bsb\theta}$ in Equation \ref{theta_update} are given by $\theta_u^{(k)}/d_u^{(k)}$ for $u = 1, ..., m$, where 
$ 
    d_u^{(k)} = \Big[\mb{g}^{(k)}_{\bsb\theta, u}\Big]^{-} + \upsilon.
$ 
Here, $\mb{g}^{(k)}_{\bsb\theta, u}$ denotes the $u$-th element of $\mb{g}_{\bsb\theta}^{(k)}$ and the notations $[a(x)]^-$ and $[a(x)]^{+}$ represent the positive and negative part of function $a(x)$, such that $a(x) = [a(x)]^+ - [a(x)]^-$. Additionally, $\upsilon$ is a small positive constant, such as $10^{-3}$ used to avoid the possible numerical issue of a zero denominator in the calculation of $d_u^{(k)}$, and does not have any impact on the final solution for $\bsb\theta$. Each $d_u$ is equal to
\begin{align*}
    d_u = & \sum_{i = 1}^n \bigg( (\delta_i + \delta_i^R) e^{\mathbf{x}_i ^{T}\boldsymbol{\beta}} \Psi_{ui}^*(t_i) + \delta_i^I e^{\mathbf{x}_i ^{T}\boldsymbol{\beta}} \frac{\Psi_{ui}^*(t_i^L) S(t_i^L)}{S(t_i^L) - S(t_i^R)}\bigg) + \big[2 \lambda_{\theta} [\mathbf{G}_{\theta} \boldsymbol{\theta}]_u \big]^{+}
\end{align*}
where $\Psi_{ui}^*(t_i) = \int_0^{t_i} \psi_u(s) \exp(\mathbf{z}_i(s)^{\top}\boldsymbol{\gamma}) ds$.

\subsection{Details for updating estimate of $\alpha$}

The estimate of $\boldsymbol{\alpha}$ is updated at each iteration using the formula given in \ref{alp_update}. In this formula, the vector $\mathbf{g}_{\boldsymbol\alpha}$ is of length $l'$, with $b$-th element
\begin{align*}
    \mathbf{g}_{\boldsymbol{\alpha}, b} = &  \sum_{i = 1}^n \bigg( \delta_i \gamma \phi_{b} (t_i) - (\delta_i + \delta_i^R)  e^{\mathbf{x}_i ^{T}\boldsymbol{\beta}} \gamma D_{i b}(t_i) + \delta_i^L \gamma e^{\mathbf{x}_i ^{T}\boldsymbol{\beta}} \frac{S(t_i) D_{ib}(t_i)}{1 - S(t_i)}  \\
    & - \delta_i^I \gamma e^{\mathbf{x}_i ^{T}\boldsymbol{\beta}} \frac{S(t_i^L) D_{ib}(t_i^L) - S(t_i^R) D_{ib}(t_i^R)}{S(t_i^L) - S(t_i^R)} \\
    & + \frac{1}{\sigma_{\varepsilon}^2} \sum_{a = 1}^{n_i} \phi_{b}(t_{ia})(\tilde{z}_{i}(t_{ia}) - z_{i}(t_{ia}))  \bigg) - 2 \lambda_{\alpha} [ \mathbf{G}_{\alpha} \boldsymbol{\alpha} ]_b
\end{align*}
where $D_{ib}(t) = \int_0^{t_i} h_0(s) \phi_{b}(s) \exp(\mathbf{z}_i(s)^{\top}\boldsymbol{\gamma}) ds$. The matrix $\mathbf{H}_{\alpha}$ has dimensions $l' \times l'$, and the $(b, d)$-th element (for $b = 1, \dots, l'$ and $d = 1, \dots, l'$) is equal to
\begin{align*}
    \mathbf{H}_{\boldsymbol{\alpha}, b, d} = & \sum_{i=1}^n e^{\mathbf{x}_i^{\top} \boldsymbol{\beta}} \gamma^2 \bigg\{ - (\delta_i + \delta_i^R) D_{i bd}(t) + \delta_i^L \frac{S(t) D_{i bd}(t)}{1 - S(t)} - e^{\mathbf{x}_i^{\top} \boldsymbol{\beta}} \delta_i^L \frac{S(t) D_{i b}(t) D_{id}(t)}{(1 - S(t))^2}\\
    & - \delta_i^I \frac{S(t_i^L)D_{i bd}(t_i^L) - S(t_i^R)D_{i bd}(t_i^R)}{S(t_i^L) - S(t_i^R)} \\
    & - \delta_i^I e^{\mathbf{x}_i^{\top} \boldsymbol{\beta}} 
    \frac{S(t_i^L)S(t_i^R)(D_{i b}(t_i^L) - D_{i b}(t_i^R))(D_{id}(t_i^L) - D_{id}(t_i^R))}{(S(t_i^L) - S(t_i^R))^2} \\
    & - \frac{1}{\sigma^2_{\varepsilon}} \sum_{a = 1}^{n_i} \phi_{b}(t_{ia}) \phi_d(t_{ia}) \bigg\} - 2 \lambda_{\alpha} [\mathbf{G}_{\alpha}]_{b,d}
\end{align*}
where $D_{i bd}(t) = \int_0^t h_0(s) \phi_{b}(s) \phi_d(s) \exp(\mathbf{z}_i(s)^{\top}\boldsymbol{\gamma}) ds$.

\subsection{Details for updating estimate of $\kappa$}

The estimate of $\boldsymbol{\kappa}$ is updated at each iteration using the formula given in \ref{kap_update}. In this formula, the vector $\mathbf{g}_{\boldsymbol{\kappa}}$ is of length $l \times n$. The element corresponding to the $c$-th basis function and the $i$-th individual is
\begin{align*}
    \mathbf{g}_{\boldsymbol{\kappa}, ci} & = \delta_i \gamma \xi_{c} (t_i) - (\delta_i + \delta_i^R)  e^{\mathbf{x}_i ^{T}\boldsymbol{\beta}} \gamma D_{i c}(t_i) + \delta_i^L \gamma e^{\mathbf{x}_i ^{T}\boldsymbol{\beta}} \frac{S(t_i) D_{ic}(t_i)}{1 - S(t_i)}  \\
    & - \delta_i^I \gamma e^{\mathbf{x}_i ^{T}\boldsymbol{\beta}} \frac{S(t_i^L) D_{ic}(t_i^L) - S(t_i^R) D_{ic}(t_i^R)}{S(t_i^L) - S(t_i^R)}  + \frac{1}{\sigma_{\varepsilon}^2} \sum_{a = 1}^{n_i} \xi_{c}(t_{ia})(\tilde{z}_{i}(t_{ia}) - z_{i}(t_{ia})) - \frac{\kappa_{ci}}{\sigma^2_{\kappa_c}}
\end{align*}
where $D_{ic}(t) = \int_0^{t_i} h_0(s) \xi_{c}(s) \exp(\mathbf{z}_i(s)^{\top}\boldsymbol{\gamma}) ds$. The matrix $\mathbf{H}_{\kappa}$ has $(l \times n) \times (l \times n)$ dimensions, and the element corresponding to the $c$-th and $f$-th bases (both for individual $i$) is equal to
\begin{align*}
    \mathbf{H}_{\kappa, ci, fi} & = e^{\mathbf{x}_i^{\top} \boldsymbol{\beta}} \gamma^2 \bigg\{ - (\delta_i + \delta_i^R) D_{icf}(t) + \delta_i^L \frac{S(t) D_{icf}(t)}{1 - S(t)} - \delta_i^L e^{\mathbf{x}_i^{\top} \boldsymbol{\beta}} \frac{S(t) D_{ic}(t) D_{if}(t)}{(1 - S(t))^2} \\
    & - \delta_i^I \frac{S(t_i^L)D_{icf}(t_i^L) - S(t_i^R)D_{icf}(t_i^R)}{S(t_i^L) - S(t_i^R)} \\
    & - \delta_i^I e^{\mathbf{x}_i^{\top} \boldsymbol{\beta}} \frac{S(t_i^L)S(t_i^R)(D_{ic}(t_i^L) - D_{ic}(t_i^R))(D_{if}(t_i^L) - D_{if}(t_i^R))}{(S(t_i^L) - S(t_i^R))^2} \\
    & - \frac{1}{\sigma^2_{\varepsilon}} \sum_{a = 1}^{n_i} \xi_{f}(t_{ia}) \xi_c(t_{ia}) \bigg\}
\end{align*}
where $D_{icf}(t) = \int_0^{t_i} h_0(s) \xi_{c}(s)\xi_{f}(s) \exp(\mathbf{z}_i(s)^{\top}\boldsymbol{\gamma}) ds$. Note however in the case where $c = f$, we have the expression
\begin{align*}
    \mathbf{H}_{\kappa, ci, ci} & = e^{\mathbf{x}_i^{\top} \boldsymbol{\beta}} \gamma^2 \bigg\{ - (\delta_i + \delta_i^R) D_{icc}(t) + \delta_i^L \frac{S(t) D_{icc}(t)}{1 - S(t)} - \delta_i^L e^{\mathbf{x}_i^{\top} \boldsymbol{\beta}} \frac{S(t) D_{ic}^2(t)}{(1 - S(t))^2} \\
    & - \delta_i^I \frac{S(t_i^L)D_{icc}(t_i^L) - S(t_i^R)D_{icc}(t_i^R)}{S(t_i^L) - S(t_i^R)} \\
    & - \delta_i^I e^{\mathbf{x}_i^{\top} \boldsymbol{\beta}} \frac{S(t_i^L)S(t_i^R)(D_{ic}(t_i^L) - D_{ic}(t_i^R))^2}{(S(t_i^L) - S(t_i^R))^2} \\
    & - \frac{1}{\sigma^2_{\varepsilon}} \sum_{a = 1}^{n_i} \xi_{c}(t_{ia}) \xi(t_{ia}) - \frac{1}{\sigma_{\kappa_c}^2} \bigg\}
\end{align*}
where $D_{icc}(t) = \int_0^{t_i} h_0(s) \xi_{c}^2(s) \exp(\mathbf{z}_i(s)^{\top}\boldsymbol{\gamma}) ds$. Additionally, note that elements of this matrix corresponding to different individuals are equal to 0.

\subsection{Additional Hessian matrix elements}

The additional second derivatives required for the Hessian matrix, in addition to the diagonal elements which are detailed in the previous sections, are given below.

\begin{align*}
    \frac{\partial \Phi ( \boldsymbol{\eta} )}{\partial \beta_j \partial \gamma_r^{\top}} = &  \sum_{i = 1}^n x_{ij} e^{\mathbf{x}_i^{\top} \boldsymbol{\beta}}  \bigg\{ -(\delta_i + \delta_i^R)  A_{ir}(t)  + \delta_i^L \frac{S(t)A_{ir}(t)}{1 - S(t)}  - \delta_i^L \frac{H(t)S(t)A_{ir}(t)}{ (1 - S(t))^2 }\\
    & - \delta_i^I  \frac{S(t_i^L)A_{ir}(t_i^L) - S(t_i^R)A_{ir}(t_i^R)}{S(t_i^L) - S(t_i^R)} - \delta_i^I  \frac{S(t_i^L)S(t_i^R)(H(t_i^L) - H(t_i^R))(A_{ir}(t_i^L) - A_{ir}(t_i^R))}{(S(t_i^L) - S(t_i^R))^2}\bigg\} 
\end{align*}
where $A_{ir}(t)$ is defined as above.

\begin{align*}
    \frac{\partial \Phi ( \boldsymbol{\eta} )}{\partial \beta_j \partial \theta_u^{\top}} = & \sum_{i = 1}^n 
    x_{ij} e^{\mathbf{x}_i^{\top} \boldsymbol{\beta}} \bigg\{ -(\delta_i + \delta_i^R) \Psi_{ui}^*(t) + \delta_i^L \frac{S(t) \Psi_{ui}^*(t)}{1 - S(t)} - \delta_i^L \frac{S(t)H(t) \Psi_{ui}^*(t)}{(1-S(t))^2} \\
    & - \delta_i^I \frac{S(t_i^L) \Psi_{ui}^*(t_i^L) - S(t_i^R) \Psi_{ui}^*(t_i^R)}{S(t_i^L) - S(t_i^R) } - \delta_i^I \frac{S(t_i^L)S(t_i^R)(H(t_i^L) - H(t_i^R))(\Psi_{ui}^*(t_i^L) - \Psi_{ui}^*(t_i^R))}{(S(t_i^L) - S(t_i^R))^2} \bigg\}
\end{align*}
where $\Psi_{ui}^*(t)$ is defined as above.

\begin{align*}
    \frac{\partial \Phi ( \boldsymbol{\eta} )}{\partial \beta_j \partial \alpha_{b}^{\top}} = & \sum_{i=1}^n x_{ij} e^{\mathbf{x}_i^{\top} \boldsymbol{\beta}} \gamma \bigg\{ -(\delta_i + \delta_i^R)D_{i b}(t) + \delta_i^L \frac{S(t)D_{i b}(t)}{1 - S(t)} - \delta_i^L  \frac{S(t) H(t) D_{i b}(t) }{(1-S(t))^2}\\
    & - \delta_i^I \frac{S(t_i^L)D_{i b}(t_i^L) - S(t_i^R)D_{i b}(t_i^R)}{S(t_i^L) - S(t_i^R)} - \delta_i^I \frac{S(t_i^L)S(t_i^R)(H(t_i^L) - H(t_i^L))(D_{i b}(t_i^L) - D_{i b}(t_i^R)}{(S(t_i^L) - S(t_i^R))^2} \bigg\}
\end{align*}
where $D_{i b}(t)$ is defined as above.

\begin{align*}
    \frac{\partial \Phi ( \boldsymbol{\eta} )}{\partial \beta_j \partial \kappa_{ic}^{\top}} = & x_{ij} e^{\mathbf{x}_i^{\top} \boldsymbol{\beta}} \gamma \bigg\{ -(\delta_i + \delta_i^R)D_{ic}(t) + \delta_i^L \frac{S(t)D_{i c}(t)}{1 - S(t)} - \delta_i^L  \frac{S(t) H(t) D_{ic}(t) }{(1-S(t))^2}\\
    & - \delta_i^I \frac{S(t_i^L)D_{ic}(t_i^L) - S(t_i^R)D_{i c}(t_i^R)}{S(t_i^L) - S(t_i^R)} - \delta_i^I \frac{S(t_i^L)S(t_i^R)(H(t_i^L) - H(t_i^L))(D_{ic}(t_i^L) - D_{ic}(t_i^R)}{(S(t_i^L) - S(t_i^R))^2} \bigg\}
\end{align*}
where $D_{i c}(t)$ is defined as above.

\begin{align*}
    \frac{\partial \Phi ( \boldsymbol{\eta} )}{\partial \gamma_r \partial \theta_u^{\top}} = & \sum_{i=1}^n e^{\mathbf{x}_i^{\top} \boldsymbol{\beta}} \bigg\{ -(\delta_i + \delta_i^R)P_{iru}(t) + \delta_i^L \frac{S(t)P_{iru}(t)}{1 - S(t)} - \delta_i^L e^{\mathbf{x}_i^{\top} \boldsymbol{\beta}} \frac{A_{ir}(t) S(t) \Psi_u^*(t)}{(1-S(t))^2}\\
    & - \delta_i^I \frac{S(t_i^L)P_{iru}(t_i^L) - S(t_i^R) P_{iru}(t_i^R)}{S(t_i^L) - S(t_i^R)} - \delta_i^I e^{\mathbf{x}_i^{\top} \boldsymbol{\beta}} \frac{S(t_i^L)S(t_i^R)(A_{ir}(t_i^L) - A_{ir}(t_i^R))(\Psi_u^*(t_i^L) - \Psi_u^*(t_i^R))}{(S(t_i^L) - S(t_i^R))^2} \bigg\}
\end{align*}
where $P_{iru}(t) = \int_0^{t} \psi_u(s) z_{ir}(s) \exp(\mathbf{z}_i(s)^{\top}\boldsymbol{\gamma}) ds$ and both $A_{ir}(t)$ and $\Psi_u^*(t)$ are defined as above.

\begin{align*}
    \frac{\partial \Phi ( \boldsymbol{\eta} )}{\partial \gamma_r \partial \alpha_{b}^{\top}} = & \sum_{i = 1}^n \bigg\{ \delta_i \phi_{b}(t) - (\delta_i + \delta_i^R) e^{\mathbf{x}_i^{\top} \boldsymbol{\beta}} E_{i b r}(t) \\
    & + \delta_i^L e^{\mathbf{x}_i^{\top} \boldsymbol{\beta}} \frac{S(t) E_{i br}(t)}{1 - S(t)} - \delta_i^L e^{2 \mathbf{x}_i^{\top} \boldsymbol{\beta}} \gamma \frac{S(t)A_{ir}(t)D_{i b}(t)}{(1-S(t)^2}\\
    & - \delta_i^I e^{\mathbf{x}_i^{\top} \boldsymbol{\beta}} \frac{S(t_i^L) E_{i br}(t_i^L) - S(t_i^R) E_{i br}(t_i^R)}{S(t_i^L) - S(t_i^R)} \\
    & - \delta_i^I e^{2 \mathbf{x}_i^{\top} \boldsymbol{\beta}} \gamma \frac{S(t_i^L)S(t_i^R)(D_{i b}(t_i^L) - D_{i b}(t_i^R))(A_{ir}(t_i^L) - A_{ir}(t_i^R))}{(S(t_i^L) - S(t_i^R))^2} \bigg\}
\end{align*}
where $E_{ibr}(t) = \int_0^t h_0(s)\phi_{b}(s)\exp(\mathbf{z}_i(s)^{\top}\boldsymbol{\gamma}) ds + \gamma_r \int_0^t h_0(s) z_{ir}(s) \phi_{b}(s) \exp(\mathbf{z}_i(s)^{\top}\boldsymbol{\gamma}) ds$, and both $A_{ir}(t)$ and $D_{i b}(t)$ are defined as above.

\begin{align*}
    \frac{\partial \Phi ( \boldsymbol{\eta} )}{\partial \gamma_r \partial \kappa_{ci}^{\top}} = & \sum_{i = 1}^n \bigg\{ \delta_i \xi_{c}(t) - (\delta_i + \delta_i^R) e^{\mathbf{x}_i^{\top} \boldsymbol{\beta}} E_{i cr}(t) \\
    & + \delta_i^L e^{\mathbf{x}_i^{\top} \boldsymbol{\beta}} \frac{S(t) E_{i cr}(t)}{1 - S(t)} - \delta_i^L e^{2 \mathbf{x}_i^{\top} \boldsymbol{\beta}} \gamma \frac{S(t)A_{ir}(t)D_{i c}(t)}{(1-S(t)^2}\\
    & - \delta_i^I e^{\mathbf{x}_i^{\top} \boldsymbol{\beta}} \frac{S(t_i^L) E_{i cr}(t_i^L) - S(t_i^R) E_{i cr}(t_i^R)}{S(t_i^L) - S(t_i^R)} \\
    & - \delta_i^I e^{2 \mathbf{x}_i^{\top} \boldsymbol{\beta}} \gamma \frac{S(t_i^L)S(t_i^R)(D_{i c}(t_i^L) - D_{i c}(t_i^R))(A_{ir}(t_i^L) - A_{ir}(t_i^R))}{(S(t_i^L) - S(t_i^R))^2} \bigg\}
\end{align*}
where $E_{icr}(t) = \int_0^t h_0(s)\xi_{c}(s)\exp(\mathbf{z}_i(s)^{\top}\boldsymbol{\gamma}) ds + \gamma_r \int_0^t h_0(s) z_{ir}(s) \xi_{c}(s) \exp(\mathbf{z}_i(s)^{\top}\boldsymbol{\gamma}) ds$ and both $A_{ir}(t)$ and $D_{i c}(t)$ are defined as above.

\begin{align*}
    \frac{\partial \Phi ( \boldsymbol{\eta} )}{\partial \theta_u \partial \theta_z^{\top}} = & \sum_{i = 1}^n \bigg\{ -\delta_i \frac{\psi_u(t) \psi_u(t)}{h^2_0(t)} - \delta_i^L e^{\mathbf{x}_i^{\top} \boldsymbol{\beta}} \frac{\Psi_{ui}^*(t) \Psi_{zi}^*(t)S(t)}{(1 - S(t))^2} \\
    & - \delta_i^I e^{\mathbf{x}_i^{\top} \boldsymbol{\beta}} \frac{S(t_i^L)S(t_i^R)(\Psi_{ui}^*(t_i^L) - \Psi_{ui}^*(t_i^R))(\Psi_{zi}^*(t_i^L) - \Psi_{zi}^*(t_i^L))}{(S(t_i^L) - S(t_i^R))^2} \bigg\} - 2 \lambda \mathbf{G}_{\theta}
\end{align*}
where $\Psi_{ui}^*(t)$ is defined as above.

\begin{align*}
    \frac{\partial \Phi ( \boldsymbol{\eta} )}{\partial \theta_u \partial \alpha_{b}^{\top}} = & \sum_{i=1}^n e^{\mathbf{x}_i^{\top} \boldsymbol{\beta}} \gamma \bigg\{ -(\delta_i + \delta_i^R) F_{i bu}(t) + \delta_i^L \frac{S(t) F_{i b u}(t) }{1 - S(t)} - \delta_i^L e^{\mathbf{x}_i^{\top} \boldsymbol{\beta}} \frac{\Psi_{ui}^*(t)S(t)D_{i b}(t)}{(1-S(t))^2}\\
    & - \delta_i^I \frac{S(t_i^L)F_{i \b u}(t_i^L) - S(t_i^R)F_{i bu}(t_i^R)}{S(t_i^L) - S(t_i^R)} \\
    & - \delta_i^I e^{\mathbf{x}_i^{\top} \boldsymbol{\beta}} \frac{S(t_i^L)S(t_i^R)(D_{i b}(t_i^L) - D_{i b}(t_i^R)(\Psi_{ui}^*(t_i^L) - \Psi_{ui}^*(t_i^R))}{(S(t_i^L) - S(t_i^R))^2} \bigg\}
\end{align*}
where $F_{i b u}(t) = \int_0^t \psi_u(s) \phi_{b}(s) \exp(\mathbf{z}_i(s)^{\top}\boldsymbol{\gamma}) ds$ and both $\Psi_{ui}^*(t)$ and $D_{i b}(t)$ are defined as above.

\begin{align*}
    \frac{\partial \Phi ( \boldsymbol{\eta} )}{\partial \theta_u \partial \kappa_{ci}^{\top}} = & e^{\mathbf{x}_i^{\top} \boldsymbol{\beta}} \gamma \bigg\{ -(\delta_i + \delta_i^R) F_{icu}(t) + \delta_i^L \frac{S(t) F_{icu}(t)}{1 - S(t)} - \delta_i^L e^{\mathbf{x}_i^{\top} \boldsymbol{\beta}} \frac{\Psi_{ui}^*(t)S(t)D_{ic}(t)}{(1-S(t))^2}\\
    & - \delta_i^I \frac{S(t_i^L)F_{icu}(t_i^L) - S(t_i^R)F_{icu}(t_i^R)}{S(t_i^L) - S(t_i^R)} \\
    & - \delta_i^I e^{\mathbf{x}_i^{\top} \boldsymbol{\beta}} \frac{S(t_i^L)S(t_i^R)(D_{ic}(t_i^L) - D_{ic}(t_i^R)(\Psi_{ui}^*(t_i^L) - \Psi_{ui}^*(t_i^R))}{(S(t_i^L) - S(t_i^R))^2} \bigg\}
\end{align*}
where $F_{icu}(t) = \int_0^t \psi_u(s) \xi_c(s) \exp(\mathbf{z}_i(s)^{\top}\boldsymbol{\gamma}) ds $ and both $\Psi_{ui}^*(t)$ and $D_{i c}(t)$ are defined as above.

\begin{align*}
    \frac{\partial \Phi ( \boldsymbol{\eta} )}{\partial \alpha_{b} \partial \kappa_{ci}^{\top}} = & e^{\mathbf{x}_i^{\top} \boldsymbol{\beta}} \gamma^2 \bigg\{ - (\delta_i + \delta_i^R) D_{i bc}(t) + \delta_i^L \frac{S(t) D_{i bc}(t)}{1 - S(t)} - \delta_i^L e^{\mathbf{x}_i^{\top} \boldsymbol{\beta}}
    \frac{S(t) D_{i b}(t) D_{ic}(t)}{(1-S(t))^2} \\
    & - \delta_i^I \frac{S(t_i^L)D_{i bc}(t_i^L) - S(t_i^R)D_{i bc}(t_i^R)}{S(t_i^L) - S(t_i^R)} \\
    & - \delta_i^I e^{\mathbf{x}_i^{\top} \boldsymbol{\beta}} \frac{S(t_i^L)S(t_i^R)(D_{i b}(t_i^L) - D_{i b}(t_i^R))(D_{ic}(t_i^L) - D_{ic}(t_i^R))}{(S(t_i^L) - S(t_i^R))^2} \\
    & - \frac{1}{\sigma^2_{\varepsilon}} \sum_{a = 1}^{n_i} \phi_{b}(t_{ia}) \xi_c(t_{ia}) \bigg\}
\end{align*}
where $D_{i bc}(t) = \int_0^t h_0(s) \phi_{b}(s) \xi_c(s) \exp(\mathbf{z}_i(s)^{\top}\boldsymbol{\gamma}) ds$ and both $D_{i b}(t)$ and $D_{i c}(t)$ are defined as above.

\section{Long formatted data}
We adopt a long-format data frame for longitudinal 
covariates as exhibited in Table \ref{tab:tab1}. The entries in the ``Status" column indicate if the event of interest may have occurred within the corresponding time interval specified by $[\text{Start, End}]$. For example, in this table, individual 1 has 1 sampling point and is right-censored. For individual 2, it has 3 sampling points: $t_{21} < t_{22} < t_2^L$, where $t_2^L$ denotes the left point of the censoring interval, and this survival time is interval-censored in $[t_2^L, t_2^R]$. 
	
\begin{table}[h!]
		\centering
		\begin{tabular}{llllccc}
			\hline
			Individual & Start  & End & Status & $z_{i1}(t)$ & \dots & $z_{iq}(t)$\\
			\hline
            1 & $0$ & $t_{11}$ & 0 & $\tilde{z}_{11}(0)$ & \dots & $\tilde{z}_{1q}(0)$ \\
            \phantom{1} & $t_{11}$ & $t_1^L$ & 0 & $\tilde{z}_{11}(t_{11})$ & \dots & $\tilde{z}_{1q}(t_{11})$ \\\\
            2 & $0$ & $t_{21}$ & 0 & $\tilde{z}_{21}(0)$ & \dots & $\tilde{z}_{2q}(0)$ \\
            \phantom{2} & $t_{21}$ & $t_{22}$ & 0 & $\tilde{z}_{21}(t_{21})$ & \dots & $\tilde{z}_{2q}(t_{21})$ \\
            \phantom{2} & $t_{22}$ & $t_2^L$ & 0 & $\tilde{z}_{21}(t_{22})$ & \dots & $\tilde{z}_{2q}(t_{22})$ \\
            \phantom{2} & $t_{2}^L$ & $t_2^R$ & 1 & $\tilde{z}_{21}(t_2^L)$ & \dots & $\tilde{z}_{2q}(t_2^L)$ \\
			\hline
\end{tabular}
  \caption{An example of a long-formatted time-varying data with different sampling points.} \label{tab:tab1}
\end{table}

\section{Variance component estimation}

For the penalised log-likelihood $\Phi(\bsb\eta)$ given in (\ref{Equa:mplcrit}), due to the fact that penalties on $\bsb\theta$ and all $\bsb\alpha_r$ are quadratic, $\Phi(\bsb\eta)$ can be viewed as a log-posterior with normal priors: $\bsb\theta \sim N(\mb{0}, \sigma_{\bsb\theta}^2\mb{R}_{\bsb\theta}^{-1})$ and $\bsb\alpha_r \sim N(\mb{0}, \sigma_{\bsb\alpha, r}^2\mb{R}_{\bsb\alpha, r}^{-1})$ for $r=1, \ldots, q$, where $\lambda_{\bsb\theta} = 1/2 \sigma_{\bsb\theta}^2$ and $\lambda_{\bsb\alpha, r} = 1/2 \sigma_{\bsb\alpha, r}^2$. 

Let 
$\bsb\sigma_{\bsb\alpha}^2=(\sigma_{\bsb\alpha, 1}^2, \ldots, \sigma_{\bsb\alpha, q}^2)^\top$ be a vector for all the variance components of
$\bsb\alpha$. 
Using the prior distributions for $\bsb\theta$ and $\bsb\alpha_r$'s, function $\Phi(\bsb\eta)$, which will be denoted as $\Phi(\bsb\eta; \sigma_{\varepsilon}^2, \sigma_{\bsb\theta}^2, \bsb\sigma_{\bsb\alpha}^2, \bsb\sigma_{\bsb\kappa}^2)$ in the following discussions to reflect its dependence on all the variance components, can be re-expressed as: 
\begin{align}
 \Phi(\bsb\eta; \sigma_{\varepsilon}^2, \sigma_{\bsb\theta}^2, \bsb\sigma_{\bsb\alpha}^2, \bsb\sigma_{\bsb\kappa}^2) = & l(\bsb\eta; \sigma_{\varepsilon}^2, \sigma_{\bsb\theta}^2, \bsb\sigma_{\bsb\alpha}^2, \bsb\sigma_{\bsb\kappa}^2) \nonumber - \frac{m}{2}\ln \sigma_{\bsb\theta}^2 - \frac{1}{2\sigma_{\bsb\theta}^2}\bsb\theta^\top \mb{R}_{\bsb\theta} \bsb\theta \\
 & + \sum_{r=1}^q\Big(-\frac{b}{2}\ln \sigma_{\bsb\alpha, r}^2 - \frac{1}{2\sigma_{\bsb\alpha, r}^2}\bsb\alpha_r^\top \mb{R}_{\bsb\alpha, r} \bsb\alpha_r\Big),
\end{align}
where $l(\bsb\eta; \sigma_{\varepsilon}^2, \sigma_{\bsb\theta}^2, \bsb\sigma_{\bsb\alpha}^2, \bsb\sigma_{\bsb\kappa}^2)$ is the log-likelihood given by (\ref{Equa:full_loglik}). 
 
We estimate the variance components $(\sigma_{\varepsilon}^2, \sigma_{\bsb\theta}^2, \bsb\sigma_{\bsb\alpha}^2, \bsb\sigma_{\bsb\kappa}^2)$ by maximising their marginal log-likelihood, which is obtained by integrating out $\bsb\eta$ from the posterior density function $\exp\{\Phi(\bsb\eta)\}$. However, since this is a high-dimensional integration, its closed-form expression is not available. The Laplace approximation can be adopted to give the following approximate marginal log-likelihood:
%
\begin{align}
    l_m(\sigma_{\varepsilon}^2, \sigma_{\bsb\theta}^2, \bsb\sigma_{\bsb\alpha}^2, \bsb\sigma_{\bsb\kappa}^2) \approx  
    \Phi(\widehat{\bsb\eta}; \sigma_{\varepsilon}^2, \sigma_{\bsb\theta}^2, \bsb\sigma_{\bsb\alpha}^2, \bsb\sigma_{\bsb\kappa}^2) -\frac{1}{2} \ln \big|\widehat{\mb{F}}_{\bsb\eta}\big| \label{marlike}
\end{align}
where 
$\widehat{\mb{F}}_{\bsb\eta} = -\partial^2 \Phi(\widehat{\bsb\eta}; \sigma_{\varepsilon}^2, \sigma_{\bsb\theta}^2, \bsb\sigma_{\bsb\alpha}^2, \bsb\sigma_{\bsb\kappa}^2)/\partial \bsb\eta\bsb\eta^\top$ is the negative Hessian of $\Phi$ with respect to $\bsb\eta$ and evaluated at the MPL estimate $\widehat{\bsb\eta}$ (where variance components are fixed at their current values). Details of all components of the matrix ${\mb{F}}_{\bsb\eta}$ can be found in the Supplementary Material.
For maximisation of this marginal likelihood with respect to the variance components, it will be useful to separate the variance components from $\bsb\eta$ in $\mb{F}_{\bsb\eta}$. To achieve  this, we re-express $\mb{F}_{\bsb\eta}$ as:
\begin{equation}
 \mb{F}_{\bsb\eta} = \mb{H}_{\bsb\eta} +  \mb{Q}_{\bsb\theta}+\sum_{r=1}^q \mb{Q}_{\bsb\alpha, r}+\mb{Q}_{\varepsilon} + \sum_{l=1}^c\sum_{r=1}^q\mb{Q}_{\bsb\kappa, lr}
\end{equation}
where the matrices on the right-hand-side are given by
\begin{align*}
    &\mb{H}_{\bsb\eta}=-\sum_i\partial^2 l_i/\partial \bsb\eta\partial\bsb\eta^\top \\ 
    &\mb{Q}_{\bsb\theta} =\frac{1}{\sigma_{\bsb\theta}^2} \text{diag}\Big(\mb{0}_{p \times p} , \mb{0}_{q\times q},  \mb{R}_{\theta}, \mb{0}_{bq\times bq}, \mb{0}_{ncq\times ncq}\Big) \\  
    &\mb{Q}_{\bsb\alpha, r} =\frac{1}{\sigma_{\bsb\alpha, r}^2} \text{diag}\Big(\mb{0}_{p \times p} , \mb{0}_{q\times q},  \mb{0}_{m\times m}, \mb{0}_{b\times b}, \ldots, \mb{0}_{b\times b}, \mb{R}_{\alpha, r}, \mb{0}_{b\times b}, \ldots, \mb{0}_{b\times b},  \mb{0}_{ncq\times ncq}\Big) \\  
    &\mb{Q}_{\varepsilon} =\frac{1}{\sigma_{\varepsilon}^2} \text{diag}\Big(\mb{0}_{p \times p} , \mb{0}_{q\times q},  \mb{0}_{m\times m}, \sum_{i=1}^n\sum_{a=1}^{n_i} \bsb\phi_1(t_{ia})\bsb\phi_1(t_{ia})^\top, \ldots, \sum_{i=1}^n\sum_{a=1}^{n_i}\bsb\phi_q(t_{ia})\bsb\phi_q(t_{ia})^\top, \\
&\hspace{1.5cm}\sum_{a=1}^{n_1} \bsb\xi_1(t_{1a})\bsb\xi_1(t_{1a})^\top, \ldots, \sum_{a=1}^{n_n} \bsb\xi_1(t_{na})\bsb\xi_1(t_{na})^\top, \ldots, \sum_{a=1}^{n_1} \bsb\xi_q(t_{1a})\bsb\xi_q(t_{1a})^\top, \\
&\hspace{1.5cm} \ldots,  \sum_{a=1}^{n_n} \bsb\xi_q(t_{na})\bsb\xi_q(t_{na})^\top\Big), \\
&\mb{Q}_{\bsb\kappa, lr} = \frac{1}{\sigma_{lr}^2} \text{diag}\Big(\mb{0}_{p \times p} , \mb{0}_{q\times q},  \mb{0}_{m\times m}, \mb{0}_{bq\times bq}, \mb{0}_{nc \times nc}, \ldots, \mb{0}_{nc \times nc}, \mb{D}_{lr}, \mb{0}_{nc \times nc}, \ldots, \mb{0}_{nc \times nc}\Big), 
\end{align*}
where $\mb{D}_{lr}$ is a block diagonal matrix consisting of $n$ identical $\mc{D}_l$ matrices: $\mb{D}_{lr} = \text{diag}(\mc{D}_l, \ldots, \mc{D}_l)$, where each $\mc{D}_l$ is a diagonal matrix with its $l$-th element being 1 and all others being 0. These matrices are used in the updates of $\sigma_{\varepsilon}^2$, $\sigma_{\theta}^2$, $\sigma_{\alpha, r}^2$ and $\sigma_{lr}^2$ as given in the main paper.

\section{Asymptotic results}

In this section, we provide some asymptotic results for the maximum penalised likelihood estimates of the parameters of the joint model, including a large sample distribution result for the estimated parameters.
The results are presented in 
Propositions \ref{prop:consistRandomEffect} and \ref{Prop:normalFixEffect}. Some assumptions and all the proofs for these results are available in subsequent sections of this Supplementary Material. 

Let 
$\bsb\zeta = (\bsb\beta^\top, \bsb\gamma^\top, \bsb\theta^\top, \bsb\alpha^\top)^\top$ be 
the collection of all the 
parameters except $\bsb\kappa$. Throughout this section, we denote by $\bsb\zeta_0$ 
the true value of $\bsb\zeta$. To emphasise its randomness, in this section we use 
$\cK_{ilr}$ to denote the random effect in \eqref{eq:longit} and use 
$\kappa_{ilr}$ a realization of $\cK_{ilr}$. 
Recall $\bkappa_{ir} = (\kappa_{i1r},\dots, \kappa_{icr})^{\top}$,  $\bkappa_{i} = (\bkappa_{i1}^{\top},\dots, \bkappa_{iq}^{\top})^{\top}$ and $\bkappa = (\bkappa_{1}^{\top},\bkappa_{2}^{\top},\dots,\bkappa_n)^{\top}$. Similarly, we can define their corresponding random vectors  $\bcK_{ir}$, $\bcK_{i}$ and $\bcK$. 
Finally, we let all the convergence below be with respect to $n\to \infty$ unless specified otherwise. 

\begin{Assu}\label{Assu:condLi}
Suppose the joint log-likelihood function $l (\bsb\eta, \bsb\kappa) $ in \eqref{Equa:full_loglik}, when the underlying probability is conditional on $\bcK = \bkappa^{*}$,  for almost surely all $\bkappa^{*}$, satisfies \cite[(a) and (c), pp. 197-198]{li2003efficiency} and satisfies \cite[(b), p. 198]{li2003efficiency}  with some $\bbb = \bbb(\bkappa^{*})$.
\end{Assu}
\begin{Rema}
It can be strenuous to check Assumption \ref{Assu:condLi}. By \cite[(i), p. 202]{li2003efficiency}, one middle step is the in-probability convergence of $\hbkappa$ to $\bkappa^{*}$ conditional on $\bcK = \bkappa^{*}$, which intuitively means there is sufficient information about each individual trajectory; see also \cite{WeMa23}. Such in-probability convergence is established in Proposition \ref{prop:consistRandomEffect} below. 
\end{Rema}
\begin{Prop}\label{prop:consistRandomEffect}
Recall $\widehat{\balpha}, \widehat{\bkappa}$ defined in \eqref{Equa:mplcrit}. Under 
\Cref{Assu:simple,Assu:phi,Assu:unifNegligible,Assu:uniqueMax} 
in the Supplementary Materials, for almost surely all $\bkappa^{*}$ and all $\xi>0$ 
\begin{align*}
    &\Prob\left(\|\hat{\balpha} - \balpha^{*}\|_{2}^{2}>\xi | \bcK = \bkappa^{*}\right) \to 0;
 & &\Prob\left(\frac{1}{n}\sum_{i=1}^{n}\|\hat{\bkappa}_{i} - \bkappa^{*}_{i}\|_{2}^{2}>\xi | \bcK = \bkappa^{*}\right) \to 0.
\end{align*}
\end{Prop}
\noindent
Now we state the result on the fixed parameters. 
\begin{Prop}\label{Prop:normalFixEffect}
Recall that $N=\sum_{i=1}^{n} n_i$. Suppose that \Cref{Assu:penaltyBias,Assu:simple,Assu:condLi} 
are in the Supplementary Materials of this paper) hold. Then for the estimator $\hbzeta$ defined in \eqref{Equa:mplcrit}, we have that for almost surely all $\bkappa^{*}$ and all $\bc$, 
\begin{equation}\label{Equa:asyNormal}
\Prob\left(\sqrt{N}(\hbzeta- \bzeta^{*})\leq \bc | \bcK = \bkappa^{*} \right) - \Prob\Big(\mathcal{N}(\mathdutchcal{\bbb}(\bkappa^{*}), \mathdutchcal{\bV}(\bkappa^{*}))\leq \bc\Big) \to 0,
\end{equation}
where $\mathdutchcal{\bbb}(\bkappa^{*})=\bbb(\bkappa^{*}) \sqrt{n^{2}/N}$, 
where $\bbb$ is the deterministic function given in \Cref{Assu:condLi} and
$ 
\mathdutchcal{\bV}(\bkappa^{*})= \lim_{n\to\infty}\frac{1}{n}\sum_{i=1}^{n}
I_{i}(\bzeta^{*},\bkappa^{*})^{-1}, 
$ 
where $I_{i}(\bzeta,\bkappa)$ 
is the expected negative Hessian matrix of individual $i$'s contribution to the log-likelihood function in \eqref{Equa:full_loglik} 
\end{Prop}

\begin{Rema}
Intuitively, \eqref{Equa:asyNormal} means
$ 
\sqrt{N}(\hbzeta- \bzeta^{*})  \ \stackrel{d}{\to} \ \mathcal{N}\left(\mathdutchcal{\bbb}(\bkappa^{*}), \mathdutchcal{\bV}(\bkappa^{*})\right).
$ 
Or, equivalently, \eqref{Equa:asyNormal} indicates, intuitively, for some standard normal random vector $\bZ$,
\[
\hbzeta- \bzeta^{*} \approx \frac{1}{\sqrt{N}}\mathdutchcal{\bbb}(\bkappa^{*}) + \frac{1}{\sqrt{N}}(\mathdutchcal{\bV}(\bkappa^{*}))^{1/2}\bZ = \frac{n}{N}\bbb(\bkappa^{*}) + \frac{1}{\sqrt{N}}(\mathdutchcal{\bV}(\bkappa^{*}))^{1/2}\bZ
\]
Particularly, as long as $n\to\infty$ and $n/N\to 0$, we get $\hbzeta$ converges in probability to $\bzeta^{*}$.
\end{Rema}
\begin{Rema}\label{Rema:b}
The expression of $\bbb$ seems similar to that in \cite[Theorem 1]{li2003efficiency}. In general, $\bbb$ is nonzero. On the other hand, in our setting, $\bkappa_{i}^{*}, i=1,2,\dots$ are realizations of a sequence of iid mean zero random variables. Hence, for almost surely all $\bkappa^{*}$, $\frac{1}{n}\sum_{i=1}^{n}\bkappa_{i}^{*} \to 0$. Hence, in our setting, $\bbb$ may eventually turn out to be zero. 
\end{Rema}
\begin{Rema}
While large sample normality for $\widehat{\bsb\eta}$ can be established, care must be taken due to the possibility of active $\bsb\theta\geq 0$ constraints. (Particularly, these active constraints are likely to occur when there are more knots selected than is necessary i.e. when the value of $m$ in \eqref{eq:approxBase} is larger than needed.) Currently, we only allow $\bsb\theta$ to stay within the boundary of the parameter space to ensure \Cref{Assu:condLi} and hence \Cref{Prop:normalFixEffect}; see also \cite[(a), p. 201]{li2003efficiency}. On the other hand, when $\theta_{u}=0$, it will lie on the boundary of the parameter space $\bsb\theta \geq 0$, and as a result, we need to remove these active constraints when considering asymptotic properties; see the large sample variance given below. 
\end{Rema}

\begin{Rema}\label{Rema:bootstrap}
Computationally intensive procedures, such as bootstrapping, can be applied when estimating the theoretical variance matrix $\bV(\bkappa^{*})$ becomes difficult. However, bootstrapping may not be computationally practical, especially when the algorithm is computational intensive. For this reason, we will next present a large sample covariance matrix result for $\widehat{\bsb\zeta}$. 
\end{Rema}

\section{Assumptions and Lemmas}

\begin{Lemm}\label{Lemm:parent}[Consistency of M-estimator with infinite dimension] 
Let $(\cD_{n}, d_{n}(\cdot,\cdot))$, $n=1,2,\dots,$ be a sequence of metric spaces. Suppose that 
$M_{n}:\bbR^{\cd}\times\cD_{n}\to \bbR$ and $\hM_{n}:\bbR^{\cd}\times\cD_{n}\to \bbR$ are sequences of deterministic and random functions, respectively. 
Suppose that $\cS\subset \bbR^{\cd}$ and $\cS_{n}\subset \cD_{n}, n=1,2,\dots$. Define $(\hbzeta,\hbkappa_{n}) = \max_{\bzeta \in \cS,\bkappa_{n}\in \cS_{n}}^{-1} \hM_{n}(\bzeta,\bkappa_{n})$. Suppose that for some $\bkappa_{n}^{*}\in \cS_{n}$, we have that for all $\ep>0$, there exist $\cn$ and $\delta>0$ that do not depend on $n$, such that for all $n>\cn$ and all deterministic $\bkappa_{n}\in \cS_{n}$ with 
$d_{n}(\bkappa_{n},\bkappa_{n}^{*})>\ep,$
we have 
\begin{equation}\label{eq:uniqueParent}
      M_{n}(\hbzeta,\bkappa_{n}^{*}) - M_{n}(\hbzeta,\bkappa_{n}) > \delta.
\end{equation}
Further, suppose that 
\begin{equation}\label{eq:unifParent}
\sup_{\bkappa_{n}\in \cS_{n}} \left|\hM_{n}(\hbzeta,\bkappa_{n})-M_{n}(\hbzeta,\bkappa_{n})\right| \pto 0.   
\end{equation}
Then,
\[
d_{n}(\hbkappa_{n} , \bkappa_{n}^{*})\pto 0. 
\]
\end{Lemm}

\begin{Assu}\label{Assu:simple}
Recall    
$\bkappa_{ir}^{*} = (\kappa_{i1r},\dots, \kappa_{icr})^{\top}$,  $\bkappa_{i}^{*} = (\bkappa_{i1}^{\top},\dots, \bkappa_{iq}^{\top})^{\top}$. Assume for simplicity
\begin{enumerate}[label=(\roman*)]
    \item In the generation of $z_{ir}(t)$ in \eqref{eq:longit}, assume $b=c$ and $\phi_{1r}=\xi_{1r},\dots,\phi_{br}=\xi_{br}$. Also assume $b$ and $\phi_{1r},\dots,\phi_{br}$ are known.  
    \item In the approximation of $h_{0}$, assume that $m$ and $\psi_{1},\dots, \psi_{m}$ are known. 
    \item The variances $\sigma_{\ep}^{2}$ and $\sigma_{lr}^{2}$ are known.
    \item It is known that the smooth parameters $\lambda_{\theta} = \lambda_{\balpha,r} = 0$.
      \item $n_{i} \equiv \cN $ for all $i$, $\cN = \cN_{n}$ is a function of $n$ and $\cN\to \infty$ as $n\to \infty$. 
        \item \label{item:interval} There exist priori known closed sets $\cS_{\balpha^{*}},  \cS_{\bkappa_{i}^{*}}, i=1,2,\dots$, 
    such that $\balpha^{*}\in \cS_{\balpha^{*}}$, $\bkappa_{i}^{*} \in \cS_{\bkappa_{i}^{*}}$, and the radii of $\cS_{\bkappa_{i}^{*}}$ is uniformly upper bounded with respect to all $i=1,2,\dots$ and all $n$. Further, $\hat{\balpha}, \hat{\bkappa}$ defined in \eqref{Equa:mplcrit} satisfy $\hat{\balpha}\in \cS_{\balpha^{*}}$ and $\hat{\bkappa}_{i} \in \cS_{\bkappa_{i}^{*}}$, $i=1,2,\dots$.
\end{enumerate}
\end{Assu}

\begin{Assu}\label{Assu:phi}
Assume for each $l=1\dots,b$ and $r=1,\dots, q$
\[
\frac{1}{\cN n}\sum_{i=1}^{n}\sum_{j=1}^{\cN }(\phi_{lr}(t_{ia}))^{2} = O(1)
\]    
\end{Assu}

\begin{Defi}\label{Defi:targetFunc}
Recall
\begin{align*}
 l_i(\bbbeta, \bgamma,  \btheta, \balpha \,|\, \bkappa)=  &\delta_i \big(\ln h_0(t_i) + \bx_i^{\top} \bbbeta + \bz_i(t_i)^{\top} \bgamma - H(t_i)\big) - \delta_i^R H(t_i^{L})  \\
 &+\delta_i^L \ln(1 - S(t_i^{R}))
     + \delta_i^I \ln (S(t_i^L) - S(t_i^R)). 
\end{align*}
Define 
\begin{align*}
    l_{\bkappa_{i}} &= - \frac{1}{2}\sum_{r=1}^q \sum_{l=1}^c \ln \sigma_{lr}^2 - \frac{1}{2}\sum_{r=1}^q \bkappa_{ir}^\top\bSigma_r^{-1}\bkappa_{ir},\\
    l_{\tilde{{\bz}}_{i}}(\balpha \,|\, \bkappa)&= -\frac{1}{2\sigma_{\varepsilon}^2} \sum_{a=1}^{n_i}\|\tilde{\bz}_{i}(t_{ia}) - \bz_i(t_{ia})\|^2- \frac{n_i}{2} \ln \sigma_{\varepsilon}^2.  
\end{align*}
Also, denote $
z_{ir}(t_{ia}) = \bphi_{r}(t_{ia})^\top\,(\balpha_r+\bkappa_{ir}).$
By \eqref{Equa:full_loglik},
\[
\frac{1}{\cN n}l (\bzeta, \bkappa) 
       = \hM_{n}^{(1)}(\balpha,  \bkappa) + \hM_{n}^{(2)}(\bzeta, \bkappa),
\]
where
\begin{align*}
\hM_{n}^{(1)}(\balpha,  \bkappa) &= 
\frac{1}{\cN n}\sum_{i = 1}^n  l_{\tbz_i}(\balpha \,|\, \bkappa)  \\
\hM_{n}^{(2)}(\bzeta, \bkappa)&= \frac{1}{\cN n}\sum_{i = 1}^n  \left[l_{\bkappa_{i}} + l_{i}(\bzeta \,|\, \bkappa)\right]
\end{align*}
Notice that 
\[
\hM_{n}^{(1)}(\balpha,  \bkappa) = - \frac{1}{2} \ln \sigma_{\varepsilon}^2 - \frac{1}{2\sigma_{\varepsilon}^2 \cN n} \sum_{i=1}^{n} 
 \sum_{a=1}^{\cN }\sum_{r=1}^{q}(\tilde{z}_{ir}(t_{ia}) - z_{ir}(t_{ia}))^2
\]
We define 
\[
M_{n} (\balpha, \bkappa, \bkappa^{*}) = - \frac{1}{2} \ln \sigma_{\varepsilon}^2 -  \Exp\left[\frac{1}{2\sigma_{\varepsilon}^2 \cN n} \sum_{i=1}^{n} \sum_{a=1}^{\cN }\sum_{r=1}^{q}(\tilde{z}_{ir}(t_{ia}) - z_{ir}(t_{ia}))^2 \vert \bcK = \bkappa^{*}\right].
\]
\end{Defi}

\begin{Lemm}\label{Lemm:kappaNegligible} Under \Cref{Assu:simple},
\[
\sup_{\bkappa_{i}\in \cS_{\bkappa_{i}^{*}}}
\left|\frac{1}{\cN n }\sum_{i=1}^{n}
l_{\bkappa_{i}} \right|  \to 0.
\]
\end{Lemm}

\begin{Assu}\label{Assu:unifNegligible}
Assume $\cN\to \infty$ fast enough such that for almost surely all $\bkappa^{*}$ and all $\xi>0$, 
\[
\Prob\left(\sup_{\bzeta, \bkappa_{i}\in \cS_{\bkappa_{i}^{*}}} \left|\frac{1}{\cN n}\sum_{i = 1}^n l_{i}(\bzeta|\bkappa)\right| > \xi | \bcK = \bkappa^{*} \right) \to 0.  
\]
\end{Assu}
\begin{Rema}
Intuitively, \Cref{Assu:unifNegligible} says that, when the longitudinal sample size $\cN$ is large, the contribution of the likelihood from the event time, compared with the contribution of the likelihood from the longitudinal data, is $\op{1}$ conditional on observed random effect $\bcK= \bkappa^{*}$. Intuitively, this may make sense, because for each $\bzeta, \bkappa, \bkappa^{*}$, under some conditions 
\begin{align*}
\Exp \left[\left( \frac{1}{\cN n}\sum_{i = 1}^n l_{i}(\bzeta, \bkappa) \right)^{2}\Big| \bcK = \bkappa^{*}  \right] &= \frac{1}{\cN^{2} n} \Var \left( l_{1}(\bzeta, \bkappa)| \bcK = \bkappa^{*}  \right) + \frac{1}{\cN^{2}} \left( \Exp \left[l_{1}(\bzeta, \bkappa)| \bcK = \bkappa^{*}  \right]\right)^{2} \\
&= O\left(\frac{1}{\cN^{2} }\right).
\end{align*}
\end{Rema}

\begin{Assu}\label{Assu:uniqueMax}
Recall $M_{n}$ defined in \Cref{Defi:targetFunc} and $\cS_{\balpha^{*}}, \cS_{\bkappa^{*}}$ defined in \Cref{item:interval} in \Cref{Assu:simple}. For almost surely all $\bkappa^{*}$ and all $\xi>0$, there exist $\cn$ and $\upsilon>0$ which do not depend on $n$ such that for all $n>\cn$ and all $\balpha\in \cS_{\balpha^{*}}, \bkappa_{i}\in \cS_{\bkappa_{i}^{*}}$, $i=1,2,\dots$, satisfying 
\begin{equation}\label{Equa:uniqueMax}
    \|\balpha- \balpha^{*}\|_{2}^{2}+\frac{1}{n}\sum_{i=1}^{n}\|\bkappa_{i}-\bkappa_{i}^{*}\|_{2}^{2} > \xi,
\end{equation}
we have that $M_{n}(\balpha^{*}, \bkappa^{*}, \bkappa^{*}) - M_{n}(\balpha, \bkappa, \bkappa^{*}) > \upsilon$.
\end{Assu}
\begin{Rema}
Intuitively, Assumption \ref{Assu:uniqueMax} says that, when the sample size $n$ is large, if we permute the longitudinal parameter by $\xi>0$ from the true parameter $\bkappa^{*}$ and $\bkappa^{*}$, then, no matter what dimension $\bkappa^{*}$ has, in the population level our likelihood will always decrease by some $\upsilon>0$ from the maximum likelihood. Given specific $\phi_{lr}$, we may be able to verify Assumption \ref{Assu:uniqueMax} by first replacing the `` $>$'' in \eqref{Equa:uniqueMax} by `` $=$'' and then applying the Lagrange Multiplier method.

\end{Rema}

\begin{Lemm}\label{Lemm:convLongitude} Recall $\hM_{n}^{(1)}$ and $M_{n}$ defined in \Cref{Defi:targetFunc}. Then under \Cref{Assu:simple,Assu:phi}, for almost surely all $\bkappa^{*}$ and all $\xi>0$, 
\[
\Prob\left(\sup_{\balpha\in \cS_{\balpha^{*}}, \bkappa_{i}\in \cS_{\bkappa_{i}^{*}}}|\hM_{n}^{(1)}(\balpha,  \bkappa) - M_{n}(\balpha, \bkappa, \bkappa^{*})|>\xi | \bcK= \bkappa^{*}\right) \to 0. 
\]
\end{Lemm}

\begin{Assu}\label{Assu:penaltyBias}
Recall the joint log-likelihood function $l (\bzeta, \bkappa) $ in \eqref{Equa:full_loglik}. Consider the setting in \Cref{Assu:simple}. Let
\[
(\bzeta_{n}^{\star}(\bkappa^{*}) , \bkappa_{n}^{\star}(\bkappa^{*}))= \argmax_{\bzeta, \bkappa}\Exp\left[l (\bzeta, \bkappa) | \bcK = \bkappa^{*} \right] 
\]
Assume that for almost surely all $\bkappa^{*}$,
\[
\bzeta_{n}^{\star}(\bkappa^{*}) - \bzeta^{*} = o\left(\frac{1}{\sqrt{N}}\right) 
\]
\end{Assu}
\begin{Rema}
\Cref{Assu:penaltyBias} intuitively assumes that the $l_{\bkappa_{i}}$ in the joint likelihood will cause a negligible effect at the population level conditional on $\bcK$.  
\end{Rema}

\section{Proofs}

\begin{proof}[{Proof of Proposition \ref{prop:consistRandomEffect}}]
By \Cref{Assu:simple,Assu:phi,Assu:unifNegligible,Assu:uniqueMax} and \Cref{Lemm:kappaNegligible,Lemm:convLongitude}, for almost surely all $\bkappa^{*}$, the proposition
follows after an application of Lemma \ref{Lemm:parent} to $\hM_{n}^{(1)}+\hM_{n}^{(2)}$ and $M_{n}$ in \Cref{Defi:targetFunc} with $\cd=m+p+q$, $\cS=\bbR^{\cd}$, $\cD_{n}=\bbR^{qb+qbn}$, $\cS_{n} = \cS_{\balpha^{*}} \times \cS_{\bkappa_{1}^{*}} \times \cdots \cS_{\bkappa_{n}^{*}}$, and $d_{n}$ below, which can be shown to be a sequence of metrics:
\[
d_{n}\left((\balpha,\bkappa),(\balpha^{*}, \bkappa^{*})\right) =\sqrt{ \|\balpha- \balpha^{*}\|_{2}^{2} + \frac{1}{n}\sum_{i=1}^{n}\|\bkappa_{i}-\bkappa_{i}^{*}\|_{2}^{2}}.
\]
\end{proof}

\begin{proof}[Proof of Lemma \ref{Lemm:parent}] By \eqref{eq:unifParent} and the definition of $\hbkappa_{n}$,
\[
M_{n}(\hbzeta,\hbkappa_{n}) -M_{n}(\hbzeta,\bkappa_{n}^{*}) = \hM_{n}(\hbzeta,\hbkappa_{n}) -\hM_{n}(\hbzeta,\bkappa_{n}^{*}) +\op{1} \geq \op{1}.
\]
Hence, by \eqref{eq:uniqueParent}, for all $\ep>0$, 
\[
\Prob\left(d_{n}(\hbkappa_{n}, \bkappa_{n}^{*})>\ep\right) \leq  \Prob(M_{n}(\hbzeta,\hbkappa_{n}) -M_{n}(\hbzeta,\bkappa_{n}^{*}) < -\delta) \leq \Prob(\op{1}<-\delta) \to 0. 
\]
\end{proof}

\begin{proof}[Proof of \Cref{Lemm:kappaNegligible}]
By \Cref{Assu:simple}, there exists constant $C>0$ such that for all $i=1,2,\dots,$ and all $\bkappa_{i}\in \cS_{\bkappa_{i}^{*}}$,
\[
\sum_{r=1}^q \bkappa_{ir}^\top\bSigma_r^{-1}\bkappa_{ir} \leq C.
\]
Hence, by \Cref{Defi:targetFunc},
\begin{align*} 
    \sup_{\bkappa_{i}\in \cS_{\bkappa_{i}^{*}}}
\left|\frac{1}{\cN n }\sum_{i=1}^{n}
l_{\bkappa_{i}} \right|  &\leq  \left|\frac{1}{2\cN}\sum_{r=1}^q \sum_{l=1}^c \ln \sigma_{lr}^2 \right| +     \sup_{\bkappa_{i}\in \cS_{\bkappa_{i}^{*}}}\frac{1}{2\cN n}\sum_{i=1}^n\sum_{r=1}^q \bkappa_{ir}^\top\bSigma_r^{-1}\bkappa_{ir}\\
&\leq \left|\frac{1}{2\cN}\sum_{r=1}^q \sum_{l=1}^c \ln \sigma_{lr}^2 \right|  + \frac{C}{2\cN} \to 0.
\end{align*}    
\end{proof}

\begin{proof}[{Proof of Lemma \ref{Lemm:convLongitude}}]
Recall, when $\bcK= \bkappa^{*}$,
\begin{align*}
    \tilde{z}_{ir}(t_{ia}) &= z_{ir}^{*}(t_{ia}) + \varepsilon_{ir}(t_{ia}), &
    z_{ir}^{*}(t_{ia}) &= \bphi_{r}(t_{ia})^\top\,(\balpha_r^{*}+\bkappa_{ir}^{*}).
\end{align*}
Also, recall in \Cref{Defi:targetFunc}, we define
\[
z_{ir}(t_{ia}) = \bphi_{r}(t_{ia})^\top\,(\balpha_r+\bkappa_{ir}).
\]
Then for almost surely all $\bkappa^{*}$
\begin{equation}\label{Equa:convLongitude}    
\begin{aligned}
&\mathph{=} M_{n}(\balpha, \bkappa, \bkappa^{*}) - \hM_{n}^{(1)}(\balpha,  \bkappa) \\
    &= \frac{1}{2\sigma_{\varepsilon}^2 \cN n} \sum_{i=1}^{n} 
 \sum_{a=1}^{\cN }\sum_{r=1}^{q}\left\{(\tilde{z}_{ir}(t_{ia}) - z_{ir}(t_{ia}))^2 - \Exp\left[(\tilde{z}_{ir}(t_{ia}) -z_{ir}(t_{ia}))^2 | \bcK= \bkappa^{*}\right]\right\}
  \\
    &= \frac{1}{2\sigma_{\varepsilon}^2 \cN n}  \sum_{i=1}^{n} 
 \sum_{a=1}^{\cN }\sum_{r=1}^{q}\left\{\left( z_{ir}^{*}(t_{ia})-z_{ir}(t_{ia}) + \ep_{ir}(t_{ia}) \right)^2- \Exp \left[\left(z_{ir}^{*}(t_{ia})-z_{ir}(t_{ia}) + \ep_{ir}(t_{ia})\right)^2 | \bcK= \bkappa^{*}\right]\right\}\\
&= \frac{1}{\sigma_{\varepsilon}^2 \cN n} \sum_{i=1}^{n} 
 \sum_{a=1}^{\cN }\sum_{r=1}^{q}\left\{z_{ir}^{*}(t_{ia})-z_{ir}(t_{ia})\right\}\ep_{ir}(t_{ia}) +A_{n}^{(2)} \\
&= A_{n}^{(1)} + A_{n}^{(2)},  
\end{aligned}
\end{equation}
where
\begin{align*}
    A_{n}^{(1)} &= \frac{1}{\sigma_{\varepsilon}^2 \cN n}\sum_{i=1}^{n} 
 \sum_{a=1}^{\cN }\sum_{r=1}^{q}\bphi_{r}(t_{ia})^{\top}\left[(\balpha_{r}^{*}-\balpha_{r})+(\bkappa_{ir}^{*}-\bkappa_{ir})\right]\ep_{ir}(t_{ia})\\
    A_{n}^{(2)} &= \frac{1}{2\sigma_{\varepsilon}^2 \cN n}\sum_{i=1}^{n} 
 \sum_{a=1}^{\cN }\sum_{r=1}^{q}\left((\ep_{ir}(t_{ia}))^{2}-\Exp\left[(\ep_{ir}(t_{ia}))^{2}\right]\right) 
\end{align*}
Since $\ep_{ir}(t_{ia})$, $i=1,2,\dots$, $a=1,2,\dots$, $r=1,\dots,q$, are all iid, by the Law of Large Numbers,
\begin{equation}\label{Equa:convLongitude1} 
\sup_{\balpha\in \cS_{\balpha^{*}}, \bkappa_{i}\in \cS_{\bkappa_{i}^{*}}}|A_{n}^{(2)}| = |A_{n}^{(2)}| = \op{1}.
\end{equation}
For $A_{n}^{(1)}$, notice that by \cref{item:interval} in \Cref{Assu:simple}, there exists constant $C>0$, such that 
\[
\sup_{\balpha\in \cS_{\balpha^{*}}, \bkappa_{i}\in \cS_{\bkappa_{i}^{*}}} \sqrt{\sum_{i=1}^{n}\|(\balpha_{r}^{*}-\balpha_{r}) + (\bkappa_{ir}^{*}-\bkappa_{ir})\|_{2}^{2}} \leq \sqrt{n}C
\]
Hence, by multiple applications of the Cauchy-Schwartz Inequality,
\begin{equation}\label{Equa:convLongitude2} 
\begin{aligned}
 &\mathph{=}  \sup_{\balpha\in \cS_{\balpha^{*}}, \bkappa_{i}\in \cS_{\bkappa_{i}^{*}}} |A_{n}^{(1)}| \\
 &\leq 
\sup_{\balpha\in \cS_{\balpha^{*}}, \bkappa_{i}\in \cS_{\bkappa_{i}^{*}}} \left|\frac{1}{\sigma_{\varepsilon}^2 \cN n}\sum_{r=1}^{q}\sum_{i=1}^{n} 
 \sum_{a=1}^{\cN }\|\bphi_{r}(t_{ia})\|_{2}\left\|(\balpha_{r}^{*}-\balpha_{r})+(\bkappa_{ir}^{*}-\bkappa_{ir})\right\|_{2}\ep_{ir}(t_{ia}) \right|\\
 &=  
\sup_{\balpha\in \cS_{\balpha^{*}}, \bkappa_{i}\in \cS_{\bkappa_{i}^{*}}} \left|\frac{1}{\sigma_{\varepsilon}^2 \cN n}\sum_{r=1}^{q}\sum_{i=1}^{n} \left\|(\balpha_{r}^{*}-\balpha_{r})+(\bkappa_{ir}^{*}-\bkappa_{ir})\right\|_{2}
 \left(\sum_{a=1}^{\cN }\|\bphi_{r}(t_{ia})\|_{2}\ep_{ir}(t_{ia})\right)\right|\\
&\leq \sup_{\balpha\in \cS_{\balpha^{*}}, \bkappa_{i}\in \cS_{\bkappa_{i}^{*}}} \frac{1}{\sigma_{\varepsilon}^2 \cN n}\sum_{r=1}^{q}\sqrt{\sum_{i=1}^{n}\left\|(\balpha_{r}^{*}-\balpha_{r})+(\bkappa_{ir}^{*}-\bkappa_{ir})\right\|_{2}^{2}}\sqrt{\sum_{i=1}^{n}\left(\sum_{a=1}^{\cN }\|\bphi_{r}(t_{ia})\|_{2}\ep_{ir}(t_{ia})\right)^{2}}\\
&\leq \frac{C}{\sigma_{\varepsilon}^2 \cN \sqrt{n}}\sum_{r=1}^{q}\sqrt{\sum_{i=1}^{n}\left(\sum_{a=1}^{\cN }\|\bphi_{r}(t_{ia})\|_{2}\ep_{ir}(t_{ia})\right)^{2}}.
\end{aligned}
\end{equation}
Notice that for each $r=1,\dots,  q$, by Assumption \ref{Assu:phi} and Jensen's Inequality,
\begin{equation}\label{Equa:convLongitude3} 
\begin{aligned}
\Exp\left[ \frac{1}{\cN \sqrt{n}}\sqrt{\sum_{i=1}^{n}\left(\sum_{a=1}^{\cN }\|\bphi_{r}(t_{ia})\|_{2}\ep_{ir}(t_{ia})\right)^{2}}\right]
&\leq \frac{1}{\cN \sqrt{n}}\sqrt{\sum_{i=1}^{n}\Var\left(\sum_{a=1}^{\cN }\|\bphi_{r}(t_{ia})\|_{2}\ep_{ir}(t_{ia})\right)}\\
&\leq \frac{\sigma_{\varepsilon}}{\sqrt{\cN}}\sqrt{\frac{1}{\cN n}\sum_{i=1}^{n}\sum_{a=1}^{\cN }\|\bphi_{r}(t_{ia})\|_{2}^{2}}\\
&\to 0.
\end{aligned}
\end{equation}
The lemma follows from \eqref{Equa:convLongitude}, \eqref{Equa:convLongitude1}, \eqref{Equa:convLongitude2}, and \eqref{Equa:convLongitude3} and Markov's Inequality.
\end{proof}

\begin{proof}[Proof of \Cref{Prop:normalFixEffect}]
Recall $\bzeta_{n}^{\star}$ defined in Assumption \ref{Assu:penaltyBias}. Recall \Cref{Defi:targetFunc}. Under \Cref{Assu:condLi} and \Cref{Assu:simple}, similar derivation as in \cite[p. 197]{li2003efficiency} with
\begin{align*}
    U_{0i1}(\bzeta,\bkappa_{i})(\mathbf{t}_i,\tilde{\bz}_{i}) &= \frac{\partial}{\partial \bzeta} \left[l_i(\bzeta \,|\, \bkappa) +  l_{\tilde{{\bz}}_{i}}(\balpha \,|\, \bkappa) +  l_{\bkappa_{i}}\right],\\
    V_{1i1}(\bzeta,\bkappa_{i})(\mathbf{t}_i,\tilde{\bz}_{i}) &= \frac{\partial}{\partial \bkappa_{i}} \left[l_i(\bzeta \,|\, \bkappa) +   l_{\tilde{{\bz}}_{i}}(\balpha \,|\, \bkappa) + l_{\bkappa_{i}}\right]\\
    V_{2i1}(\bzeta,\bkappa_{i})(\mathbf{t}_i,\tilde{\bz}_{i}) &= \left(V_{1i1}(\bzeta,\bkappa_{i})(\mathbf{t}_i,\tilde{\bz}_{i})\right)^{2}+ \frac{\partial^{2}}{\partial \bkappa_{i}^{2}} \left[l_i(\bzeta \,|\, \bkappa) +   l_{\tilde{{\bz}}_{i}}(\balpha \,|\, \bkappa) + l_{\bkappa_{i}}\right]
\end{align*}
gives that for almost surely all $\bkappa^{*}$ and all $\bc$,
\[
\Prob\left(\sqrt{N}(\hbzeta- \bzeta_{n}^{\star}(\bkappa^{*}))\leq \bc | \bcK = \bkappa^{*} \right) -\Prob\left(\pazocal{N}(\mathcal{\bb}(\bkappa^{*}), \mathcal{\bV}(\bkappa^{*}))\leq \bc\right) \to 0,
\]
The proposition follows from Assumption \ref{Assu:penaltyBias} and an application of the Slutsky Theorem conditional on $\bcK=\bkappa^{*}$.
\end{proof}

\section{Details of data generation for the additional simulation study}

In Study 2c, we again investigated the performance of our proposed MPL method but now generated a longitudinal covariate with an exponential decay-like shape to emulate the those observed in our real data example (see Section \ref{sec:app}), so that we had
\begin{equation}
    \tilde{z}_i(t) = 1 - 0.75e^{4t}/(1 + e^{4t}) + \kappa_{0i} + \kappa_{1i}t + \varepsilon_i(t)
\end{equation}
where the random slope was $\kappa_{0i} \sim \mathcal{N}(0, 0.1^2)$, the random intercept was $\kappa_{1i} \sim \mathcal{N}(0, 0.05^2)$, and the measurement error was $\varepsilon_i(t) \sim \mathcal{N}(0, 0.05^2)$. To investigate spline estimation of the longitudinal trajectory in our MPL method, we used a combination of a cubic smoothing spline model for the the mean trajectory and a linear model for the random effects, with the form
\begin{equation}
    \widehat{z}_i(t) = \sum_{r = 0}^4 (\phi_r(t) \widehat{\alpha}_r ) + \widehat{\kappa}_{i0} + \widehat{\kappa}_{i1}t
\end{equation}
where the cubic smoothing splines $\boldsymbol{\phi}(t)$ were defined with an interior knot at $t = 0.5$. For Study 2c we generated true event times from a log-normal baseline hazard with mean $0.3$ and variance $0.5^2$, and the true Cox regression parameters were $\beta_1 = 0.2$, $\beta_2 = -0.5$ and $\gamma = 1$, with time-fixed covariates in the Cox model drawn from $x_{i1} \sim \mathcal{N}(0, 1)$ and $x_{i2} \sim \text{Bern}(0.5)$. The observed censoring times $t_i^L$ and $t_i^R$ and the observation times for the longitudinal covariate (the $t_{ia}$) were simulated as per the description for Study 2a.

\section{Additional simulation results}

\subsection{Study 1 (right censoring)}

Table \ref{right_cens2} shows the bias, standard error estimations and coverage probabilities for Study 1 when the measurement error variance is increased to $\sigma^2_{\varepsilon} = 0.2^2$. Compared to the JM estimates, the MPL estimates appear to be more sensitive to larger measurement error values, with the bias in the MPL estimates in Table \ref{right_cens1} being smaller than those in Table \ref{right_cens2}. However, this difference is minimised in scenarios with larger sample sizes, larger numbers of longitudinal observations, and less censoring. Table \ref{tab:var_comp_stdy1} shows that in general the bias in the estimates of the variance components (the variance of the measurement error and random effects distributions) is in general smaller for the MPL method. However, the bias in the MPL variance component estimates increases with increasing $\sigma_{\varepsilon}^2$, which may explain the somewhat increased bias in the MPL regression coefficients.

\begin{table}[h!]
    \centering
    \scriptsize
    \begin{tabular}{ll | cccccc | cccccc}
         \hline
         && \multicolumn{6}{c}{$n = 200$, $\Bar{n}_i = 5$, $\pi^E = 0.7$} & \multicolumn{6}{c}{$n = 200$, $\Bar{n}_i = 20$, $\pi^E = 0.7$}\\
         \hline
         && $\beta$ & $\gamma$ & $\alpha_0$ & $\alpha_1$ & $\alpha_2$ & $\alpha_3$ & $\beta$ & $\gamma$ & $\alpha_0$ & $\alpha_1$ & $\alpha_2$ & $\alpha_3$ \\
         \hline
         Bias & MPL & -0.032 & 0.008 & -0.001 & -0.012 & 0.039 & -0.006 & -0.005 & 0.008 & -0.002 & -0.001 & -0.004 & 0.003 \\
         & JM & 0.006 & 0.005 & 0.001 & 0.007 & 0.001 & 0.008 & 0.012 & 0.001 & -0.015 & 0.015 & 0.003 & -0.002 \\
         SE & MPL & 0.143 & 0.093 & 0.036 & 0.123 & 0.242 & 0.125 & 0.144 & 0.086 & 0.036 & 0.080 & 0.107 & 0.050 \\
         & & (0.149) & (0.108) & (0.037) & (0.131) & (0.257) & (0.139)& (0.139) & (0.099) & (0.036) & (0.075) & (0.092) & (0.045) \\
         & JM & 0.148 & 0.106 & 0.028 & 0.126 & 0.236 & 0.114 & 0.148 & 0.099 & 0.016 & 0.074 & 0.115 & 0.050 \\
         & & (0.164) & (0.114) & (0.070) & (0.152) & (0.306) & (0.171) & (0.148) & (0.111) & (0.113) & (0.162) & (0.227) & (0.120) \\
         CP & MPL & 0.92 & 0.91 & 0.94 & 0.95 & 0.94 & 0.94 & 0.96 & 0.91 & 0.95 & 0.95 & 0.98 & 0.97 \\
         & & (0.95) & (0.94) & (0.95) & (0.97) & (0.94) & (0.94)& (0.95) & (0.96) & (0.94) & (0.93) & (0.96) & (0.95) \\
         & JM & 0.93 & 0.94 & 0.58 & 0.90 & 0.90 & 0.86 & 0.95 & 0.93 & 0.25 & 0.63 & 0.66 & 0.60 \\
         & & (0.94) & (0.94) & (0.94) & (0.94) & (0.93) & (0.93) & (0.95) & (0.95) & (0.97) & (0.95) & (0.96) & (0.96) \\
         \hline
         && \multicolumn{6}{c}{$n = 1000$, $\Bar{n}_i = 5$, $\pi^E = 0.7$} & \multicolumn{6}{c}{$n = 1000$, $\Bar{n}_i = 20$, $\pi^E = 0.7$}\\
         \hline
         && $\beta$ & $\gamma$ & $\alpha_0$ & $\alpha_1$ & $\alpha_2$ & $\alpha_3$ & $\beta$ & $\gamma$ & $\alpha_0$ & $\alpha_1$ & $\alpha_2$ & $\alpha_3$ \\
         \hline
         Bias & MPL & 0.013 & -0.033 & 0.003 & -0.009 & 0.020 & 0.006 & 0.012 & -0.001 & -0.005 & -0.001 & 0.001 & 0.001 \\
         & JM & 0.002 & 0.011 & 0.007 & 0.009 & 0.004 & 0.009 & 0.003 & 0.017 & 0.006 & 0.008 & 0.014 & -0.004 \\
         SE & MPL & 0.065 & 0.040 & 0.016 & 0.053 & 0.100 & 0.050 & 0.064 & 0.038 & 0.016 & 0.035 & 0.046 & 0.021 \\
         & & (0.065) & (0.045) & (0.017) & (0.059) & (0.103) & (0.053) & (0.064) & (0.045) & (0.017) & (0.036) & (0.047) & (0.022) \\
         & JM & 0.065 & 0.045 & 0.015 & 0.058 & 0.105 & 0.048 & 0.065 & 0.043 & 0.009 & 0.033 & 0.050 & 0.021 \\
         & & (0.064) & (0.043) & (0.053) & (0.081) & (0.163) & (0.089) & (0.064) & (0.046) & (0.067) & (0.112) & (0.121) & (0.063) \\
         CP & MPL & 0.95 & 0.81 & 0.93 & 0.93 & 0.95 & 0.94 & 0.94 & 0.89 & 0.94 & 0.97 & 0.98 & 0.95 \\
         & & (0.94) & (0.90) & (0.94) & (0.94) & (0.97) & (0.95) & (0.94) & (0.96) & (0.94) & (0.98) & (0.97) & (0.94) \\
         & JM & 0.96 & 0.95 & 0.45 & 0.82 & 0.78 & 0.73 & 0.96 & 0.90 & 0.25 & 0.40 & 0.56 & 0.44 \\
         & & (0.95) & (0.95) & (0.94) & (0.95) & (0.95) & (0.94) & (0.95) & (0.93) & (0.94) & (0.96) & (0.95) & (0.96) \\
         \hline
         && \multicolumn{6}{c}{$n = 200$, $\Bar{n}_i = 5$, $\pi^E = 0.3$} & \multicolumn{6}{c}{$n = 200$, $\Bar{n}_i = 20$, $\pi^E = 0.3$}\\
         \hline
         && $\beta$ & $\gamma$ & $\alpha_0$ & $\alpha_1$ & $\alpha_2$ & $\alpha_3$ & $\beta$ & $\gamma$ & $\alpha_0$ & $\alpha_1$ & $\alpha_2$ & $\alpha_3$ \\
         \hline
         Bias & MPL & 0.018 & -0.044 & 0.001 & -0.011 & 0.010 & 0.034 & 0.047 & 0.002 & -0.003 & -0.003 & 0.025 & -0.021 \\
         & JM & 0.027 & -0.042 & 0.007 & -0.005 & 0.025 & -0.023 & 0.057 & -0.010 & -0.008 & -0.012 & 0.033 & -0.026 \\
         SE & MPL & 0.238 & 0.203 & 0.037 & 0.174 & 0.519 & 0.396 & 0.239 & 0.180 & 0.036 & 0.102 & 0.237 & 0.167 \\
         & & (0.246) & (0.224) & (0.033) & (0.182) & (0.548) & (0.239) & (0.233) & (0.210) & (0.035) & (0.101) & (0.227) & (0.158) \\
         & JM & 0.238 & 0.200 & 0.029 & 0.177 & 0.511 & 0.376 & 0.238 & 0.188 & 0.015 & 0.104 & 0.246 & 0.164 \\
         & & (0.243) & (0.197) & (0.069) & (0.185) & (0.601) & (0.489) & (0.245) & (0.211) & (0.093) & (0.215) & (0.371) & (0.299) \\
         CP & MPL & 0.96 & 0.95 & 0.97 & 0.93 & 0.92 & 0.93 & 0.96 & 0.93 & 0.95 & 0.96 & 0.96 & 0.97 \\
         & & (0.96) & (0.94) & (0.98) & (0.96) & (0.94) & (0.96) & (0.94) & (0.96) & (0.95) & (0.95) & (0.94) & (0.95) \\
         & JM & 0.95 & 0.96 & 0.54 & 0.94 & 0.89 & 0.88 & 0.95 & 0.91 & 0.24 & 0.70 & 0.78 & 0.71 \\
         & & (0.94) & (0.95) & (0.94) & (0.95) & (0.94) & (0.95) & (0.95) & (0.96) & (0.94) & (0.97) & (0.96) & (0.96) \\
         \hline
         && \multicolumn{6}{c}{$n = 1000$, $\Bar{n}_i = 5$, $\pi^E = 0.3$} & \multicolumn{6}{c}{$n = 1000$, $\Bar{n}_i = 20$, $\pi^E = 0.3$}\\
         \hline
         && $\beta$ & $\gamma$ & $\alpha_0$ & $\alpha_1$ & $\alpha_2$ & $\alpha_3$ & $\beta$ & $\gamma$ & $\alpha_0$ & $\alpha_1$ & $\alpha_2$ & $\alpha_3$ \\
         \hline
         Bias & MPL & -0.019 & -0.034 & -0.001 & 0.001 & -0.028 & 0.059 & -0.002 & 0.001 & 0.002 & 0.001 & -0.010 & 0.005 \\
         & JM & -0.018 & 0.001 & -0.004 & 0.008 & -0.019 & 0.022 & 0.009 & 0.015 & 0.013 & -0.007 & -0.023 & 0.015 \\
         SE & MPL & 0.102 & 0.083 & 0.016 & 0.076 & 0.220 & 0.168 & 0.103 & 0.076 & 0.016 & 0.045 & 0.104 & 0.073\\
         & & (0.095) & (0.086) & (0.016) & (0.074) & (0.207) & (0.154) & (0.103) & (0.079) & (0.016) & (0.046) & (0.112) & (0.080) \\
         & JM & 0.103 & 0.084 & 0.015 & 0.082 & 0.230 & 0.166 & 0.103 & 0.076 & 0.016 & 0.045 & 0.104 & 0.073 \\
         & & (0.100) & (0.083) & (0.044) & (0.095) & (0.335) & (0.271) & (0.105) & (0.078) & (0.075) & (0.130) & (0.234) & (0.190) \\
         CP & MPL & 0.96 & 0.93 & 0.95 & 0.95 & 0.99 & 0.95 & 0.98 & 0.97 & 0.94 & 0.95 & 0.94 & 0.94\\
         & & (0.95) & (0.93) & (0.95) & (0.94) & (0.95) & (0.92) & (0.97) & (0.97) & (0.94) & (0.97) & (0.99) & (0.95) \\
         & JM & 0.95 & 0.94 & 0.55 & 0.91 & 0.81 & 0.73 & 0.97 & 0.95 & 0.17 & 0.55 & 0.69 & 0.63 \\
         & & (0.93) & (0.95) & (0.95) & (0.97) & (0.96) & (0.95) & (0.99) & (0.94) & (0.95) & (0.94) & (0.94) & (0.96) \\
         \hline
    \end{tabular}
    \caption{Study 1 (right censoring) regression parameter simulation results; $n$ refers to the total sample size, $\bar{n}_i$ refers to the average number of longitudinal observations, $\pi^E$ refers to the proportion of non-censored observations. For all scenarios summarised in this table, the true value of $\sigma_{\varepsilon}^2 = 0.2^2$.}
    \label{right_cens2}
\end{table}

Plots showing the estimates of the baseline hazard function from the MPL method in Study 1 are shown in Figures \ref{fig:right_cens_h0t}, \ref{fig:right_cens_h0t_2}, \ref{fig:right_cens_h0t_3}, and \ref{fig:right_cens_h0t_4}. Additional results detailing the bias in the estimates of the variance for the measurement error distribution and the random effects distributions for both the MPL and \texttt{JM} methods are shown in Table \ref{tab:var_comp_stdy1}. This table also gives the mean integrated square error (MISE) of the baseline hazard function estimation from both the MPL and \texttt{JM} methods.

\begin{figure}
    \centering
    \includegraphics[width=0.40\textwidth]{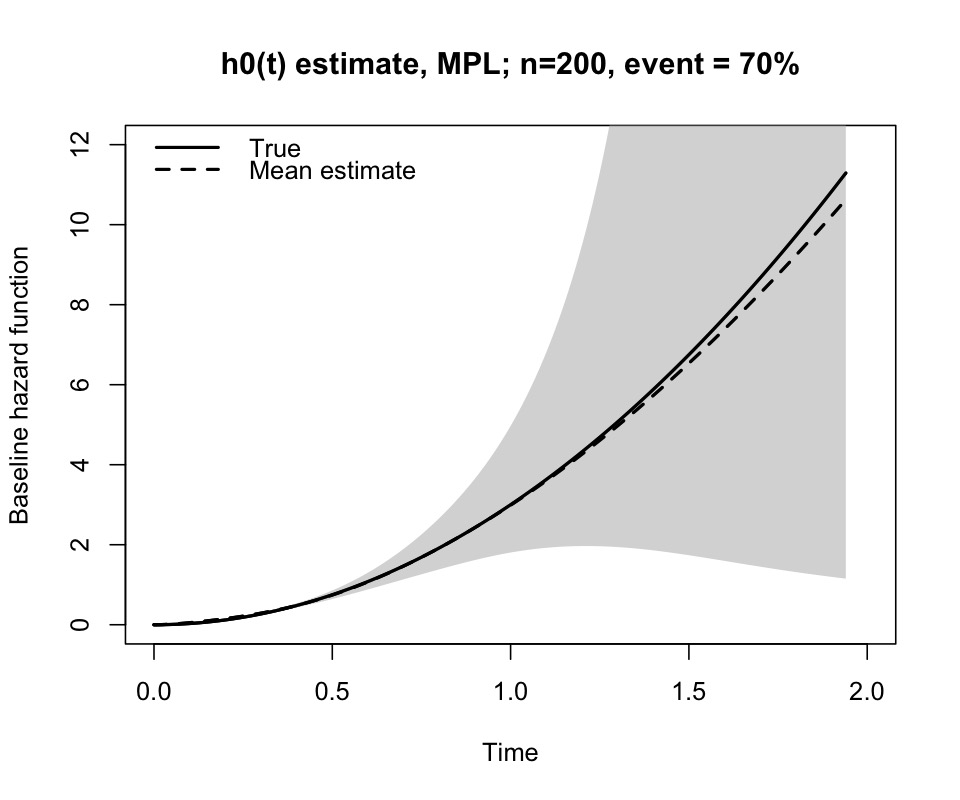}
    \includegraphics[width=0.40\textwidth]{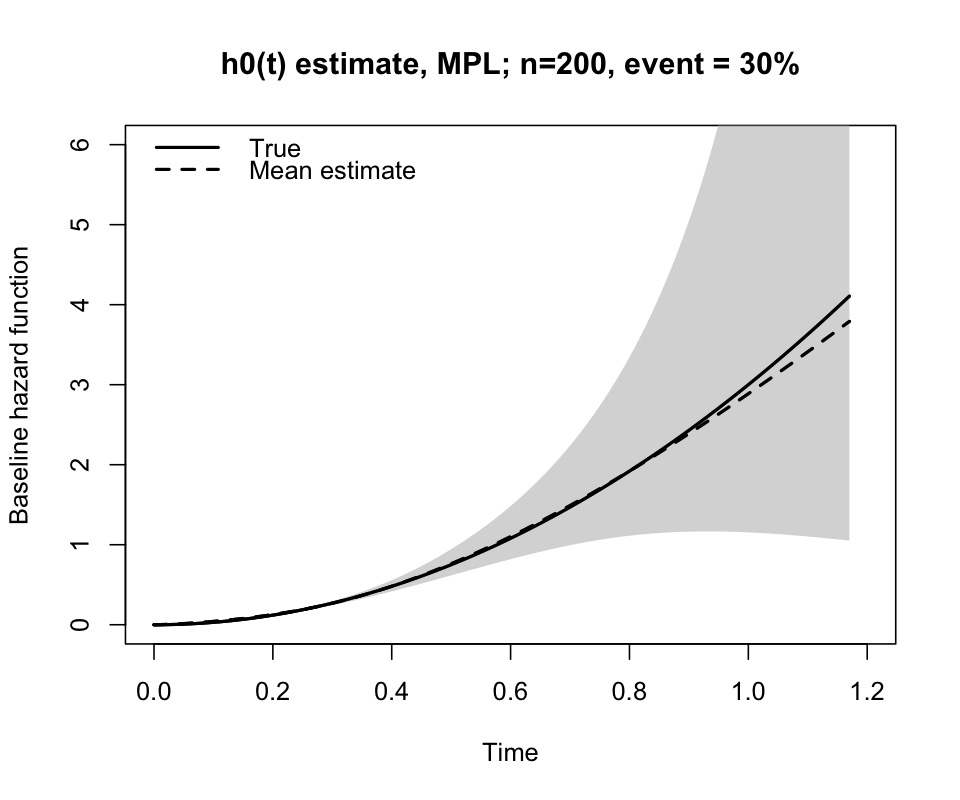}
    \includegraphics[width=0.40\textwidth]{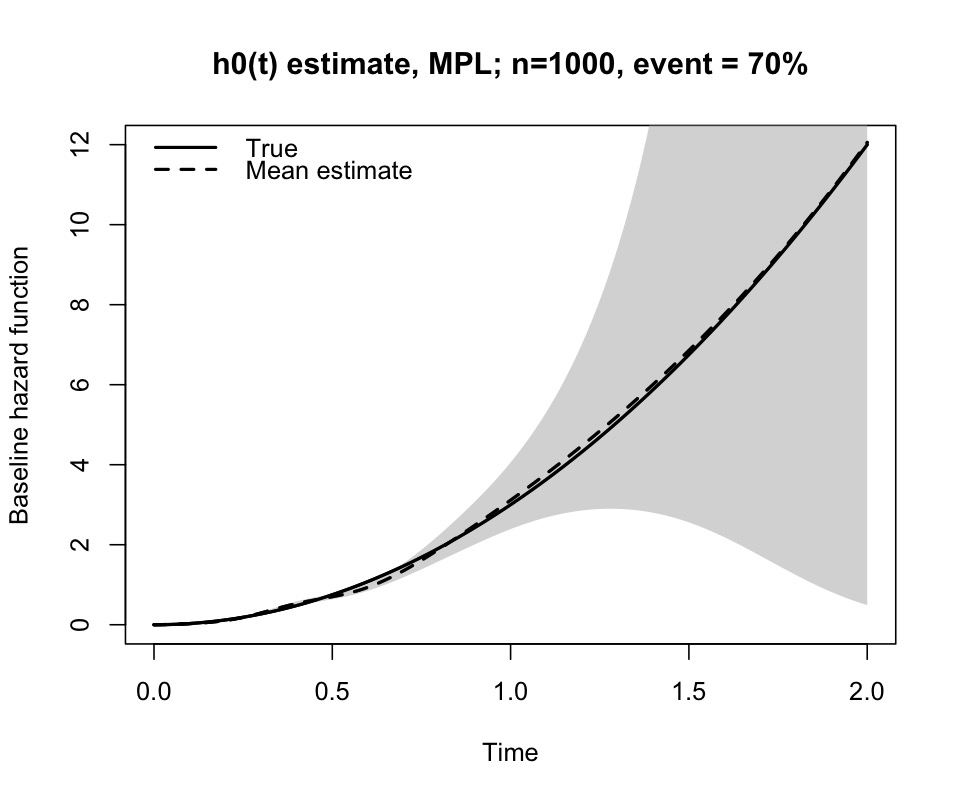}
    \includegraphics[width=0.40\textwidth]{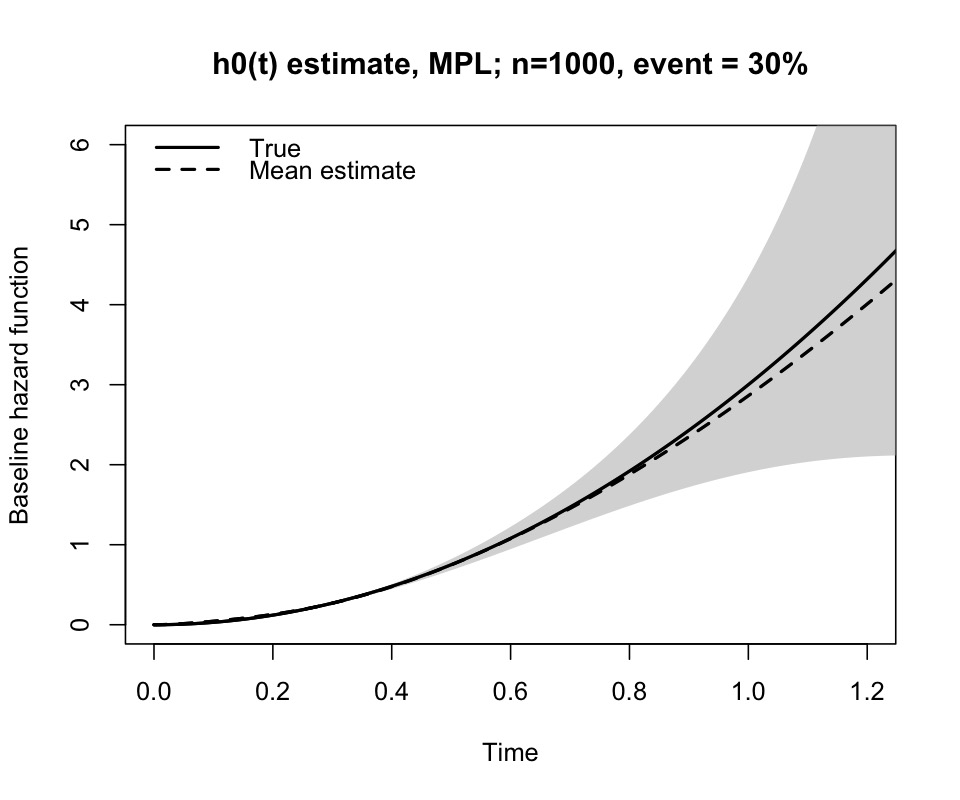}
    \caption{Estimates of the baseline hazard function for $\sigma_{\varepsilon} = 0.05$ and $\bar{n}_i = 20$. The solid line is the true baseline hazard function, the dashed line is the mean estimate and the grey area represents the asymptotic $95\%$ coverage probability.}
    \label{fig:right_cens_h0t}
\end{figure}

\begin{figure}
    \centering
    \includegraphics[width=0.49\textwidth]{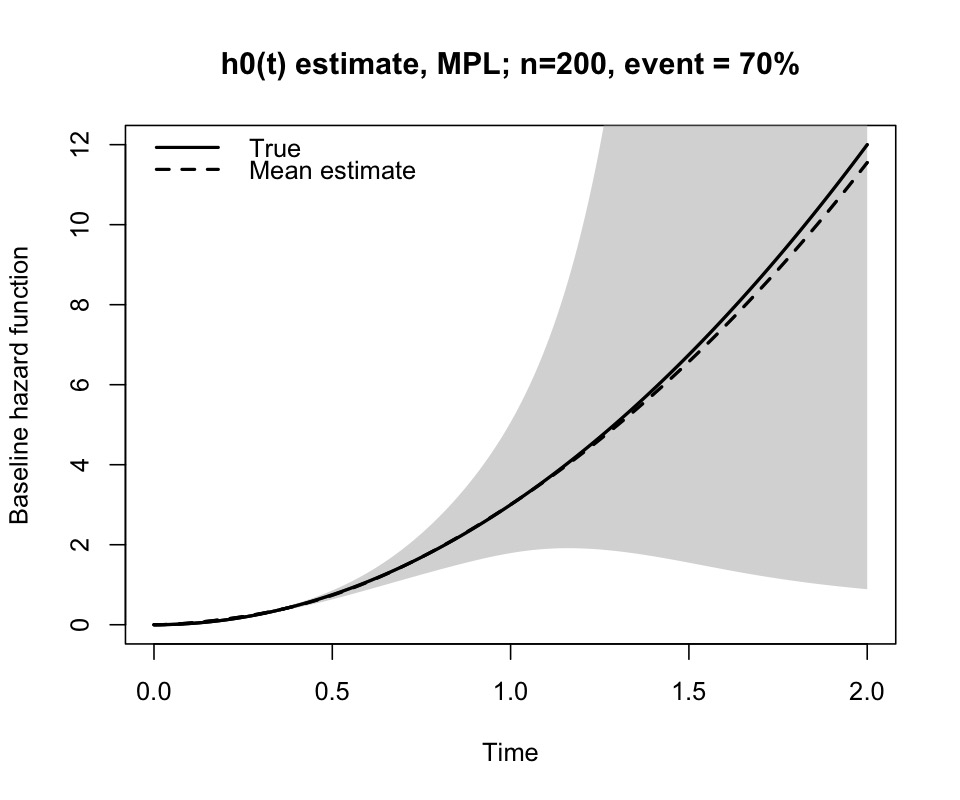}
    \includegraphics[width=0.49\textwidth]{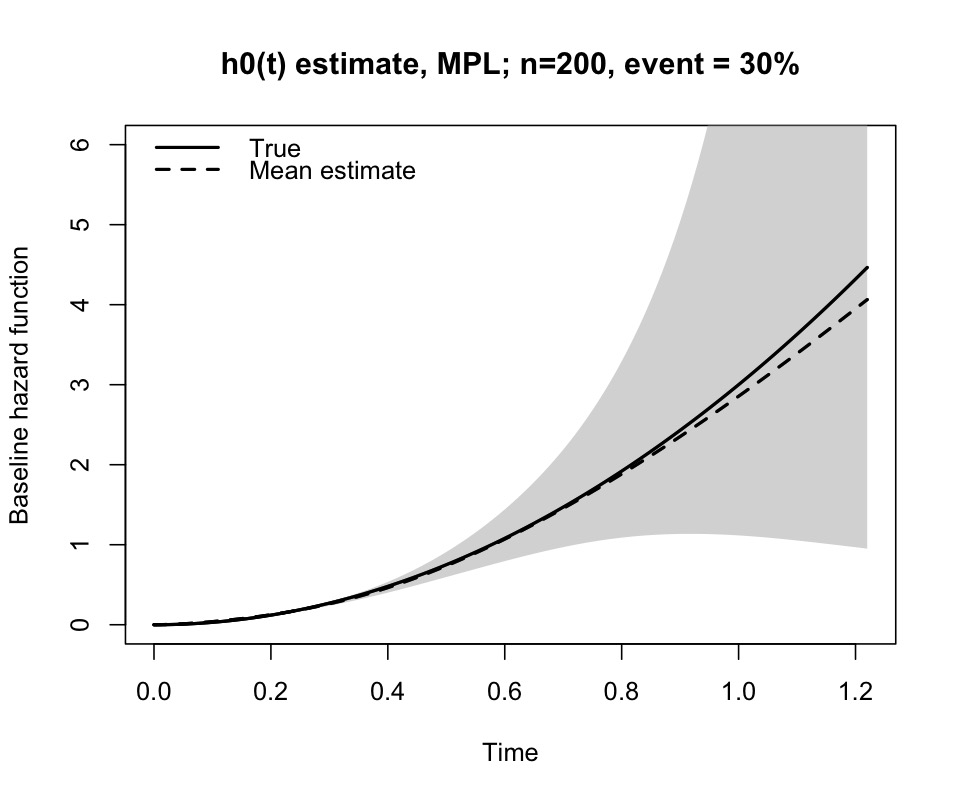}
    \includegraphics[width=0.49\textwidth]{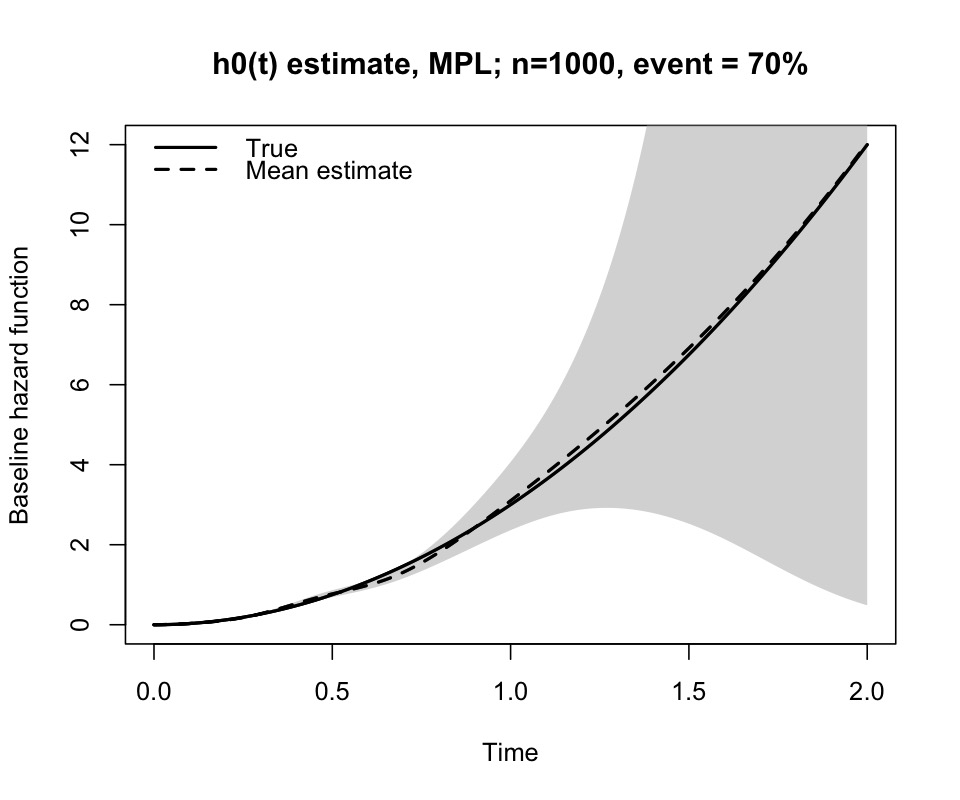}
    \includegraphics[width=0.49\textwidth]{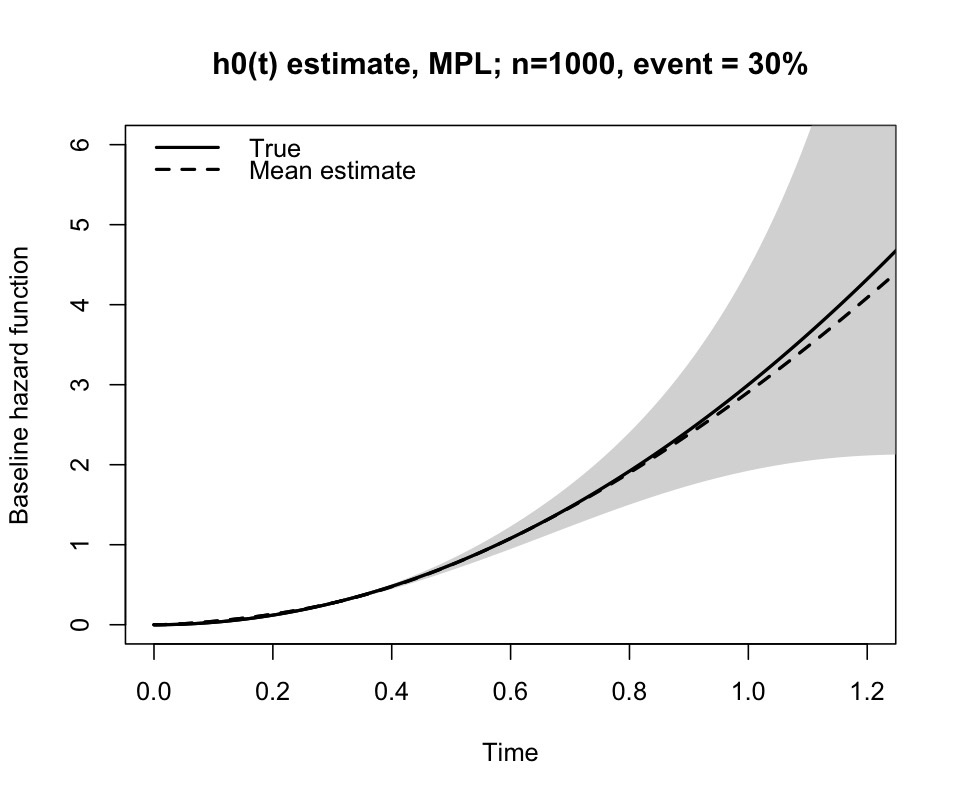}
    \caption{Estimates of the baseline hazard function for $\sigma_{\varepsilon} = 0.05$ and $\bar{n}_i = 5$. The solid line is the true baseline hazard function, the dashed line is the mean estimate and the grey area represents the asymptotic $95\%$ coverage probability.}
    \label{fig:right_cens_h0t_2}
\end{figure}

\begin{figure}
    \centering
    \includegraphics[width=0.49\textwidth]{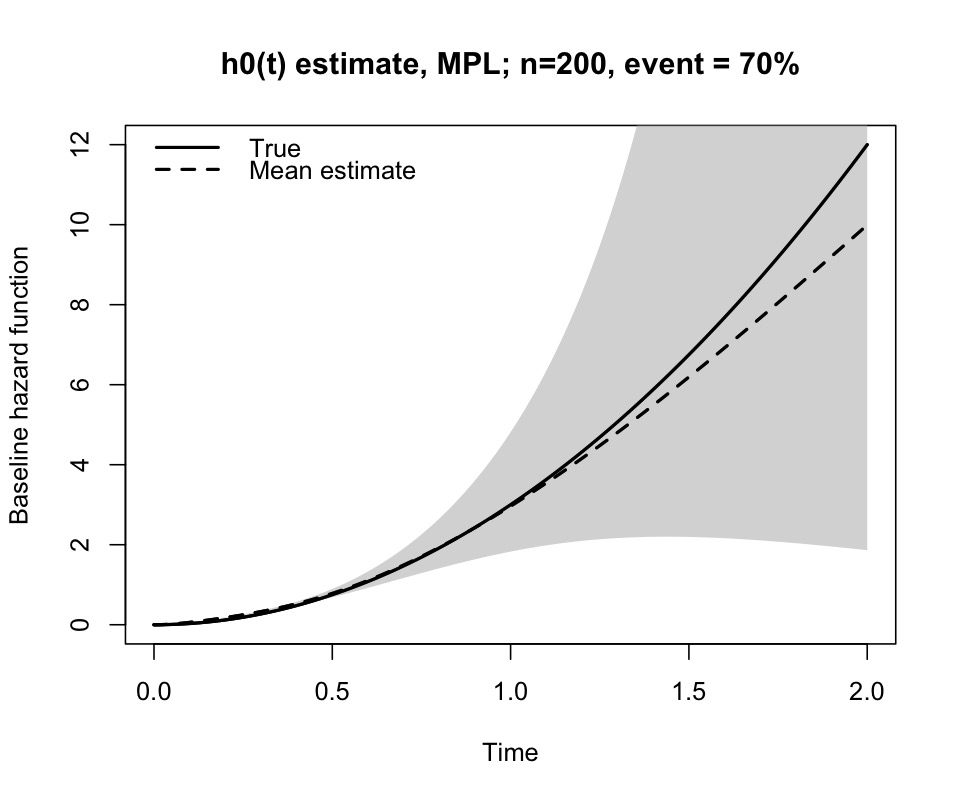}
    \includegraphics[width=0.49\textwidth]{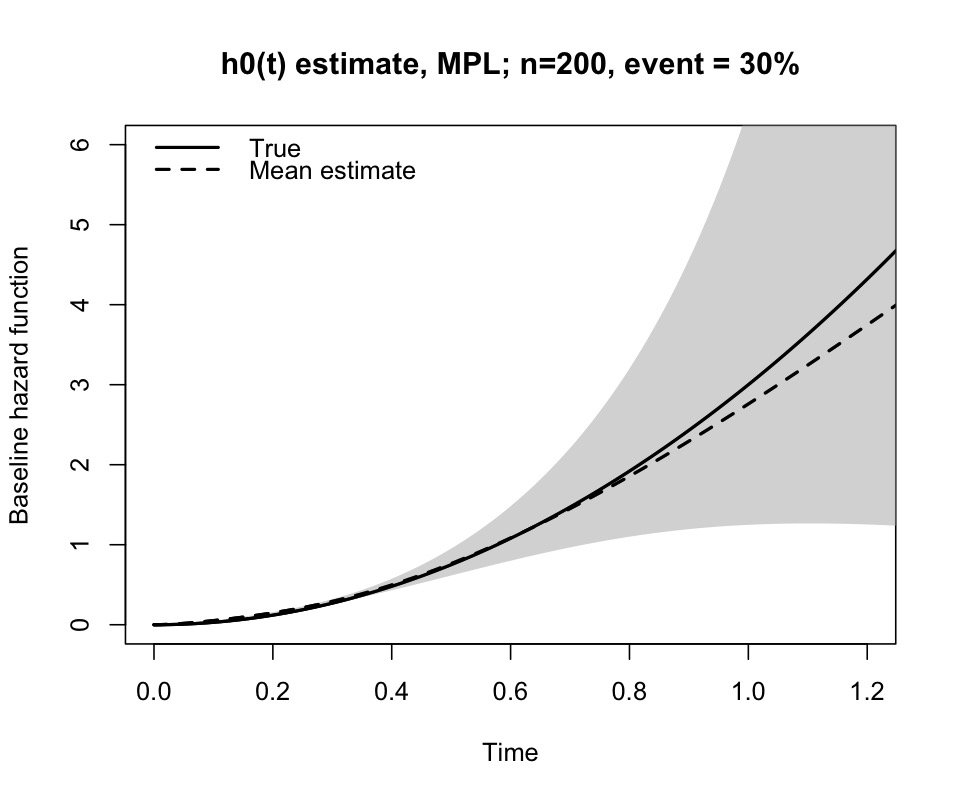}
    \includegraphics[width=0.49\textwidth]{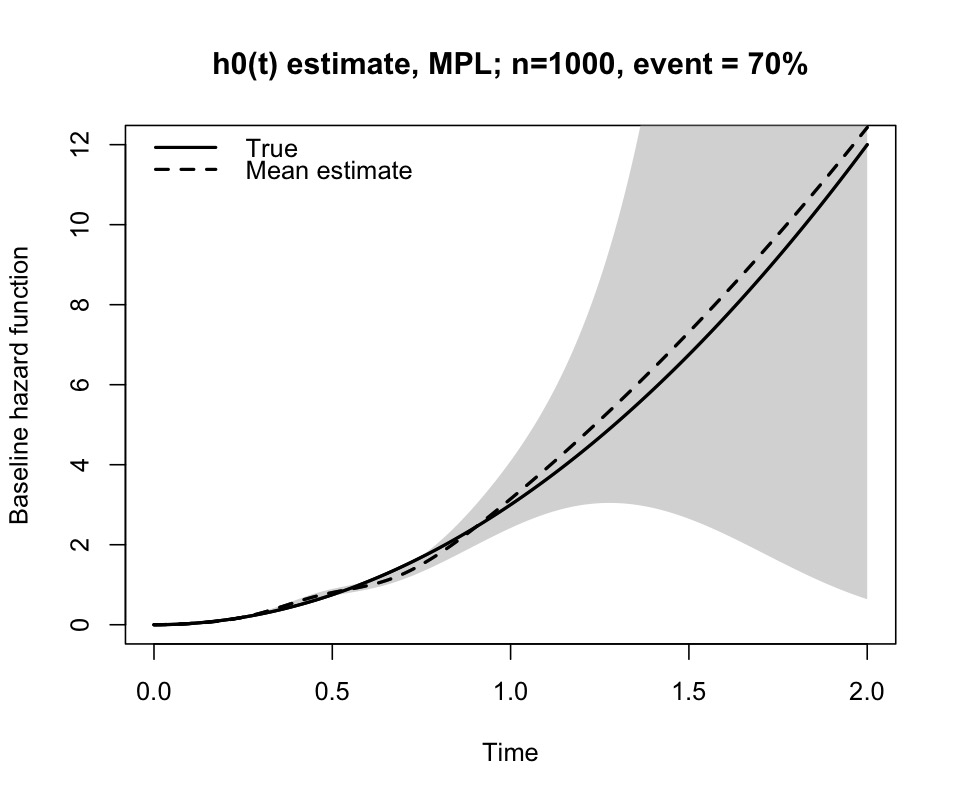}
    \includegraphics[width=0.49\textwidth]{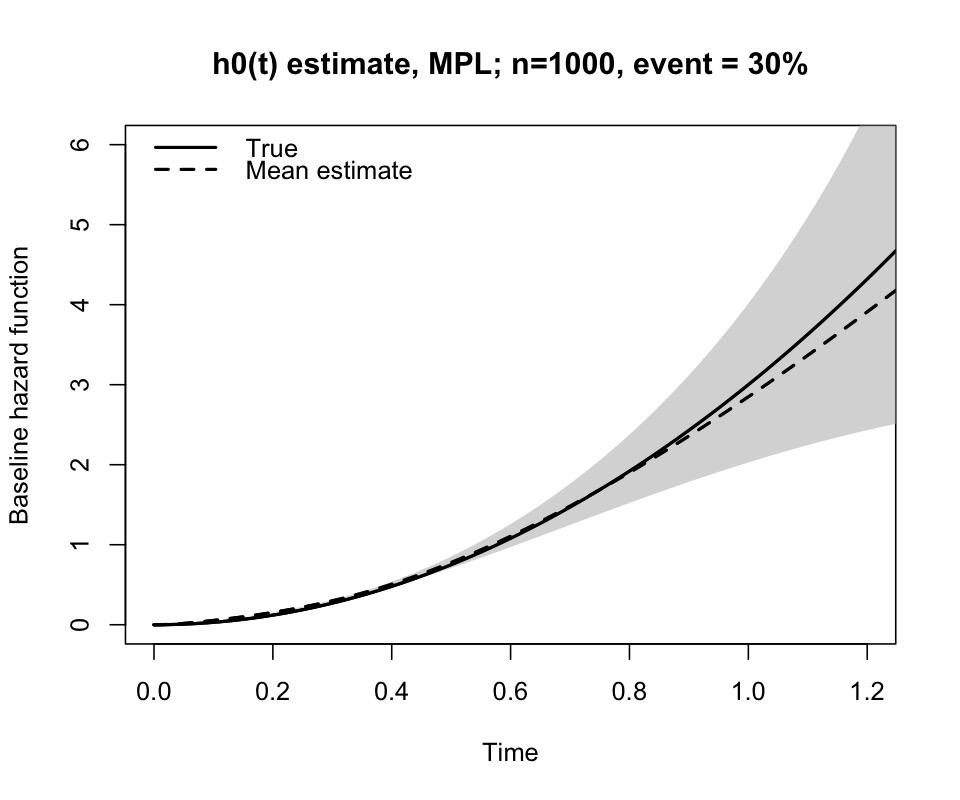}
    \caption{Estimates of the baseline hazard function for $\sigma_{\varepsilon} = 0.2$ and $\bar{n}_i = 5$. The solid line is the true baseline hazard function, the dashed line is the mean estimate and the grey area represents the asymptotic $95\%$ coverage probability.}
    \label{fig:right_cens_h0t_3}
\end{figure}

\begin{figure}
    \centering
    \includegraphics[width=0.49\textwidth]{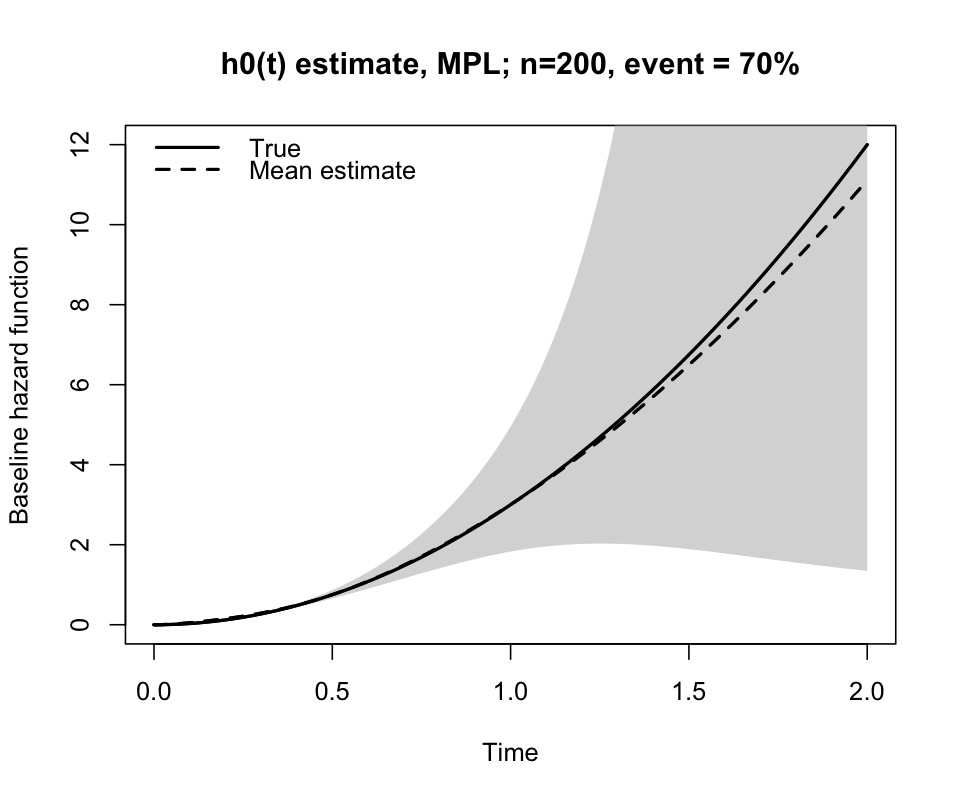}
    \includegraphics[width=0.49\textwidth]{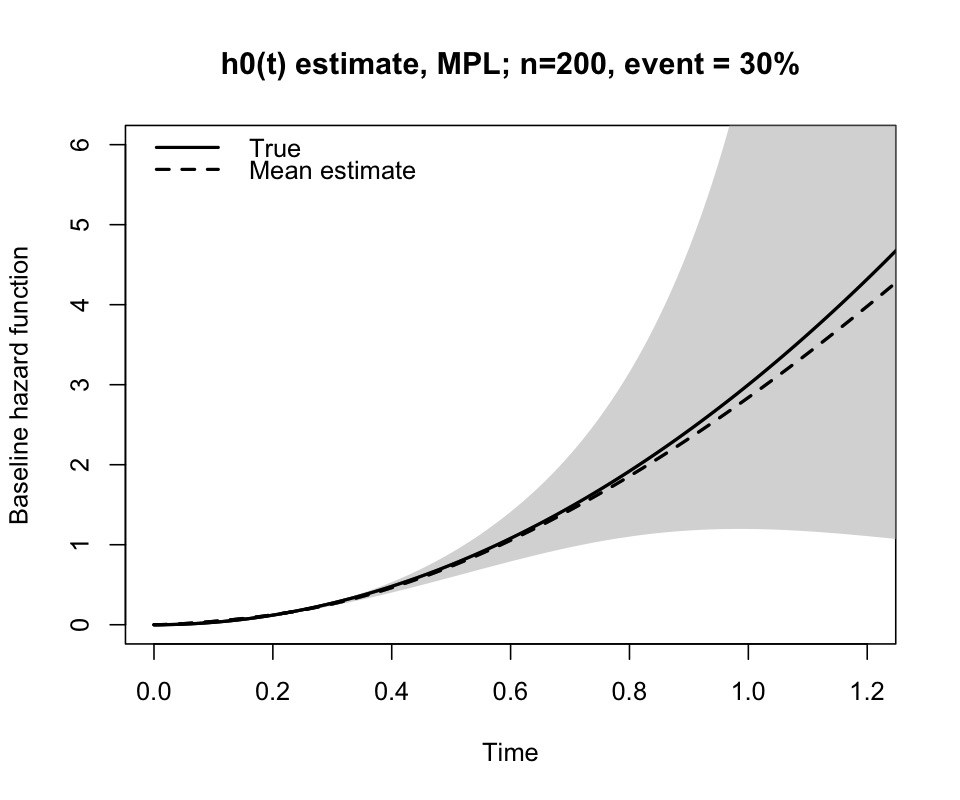}
    \includegraphics[width=0.49\textwidth]{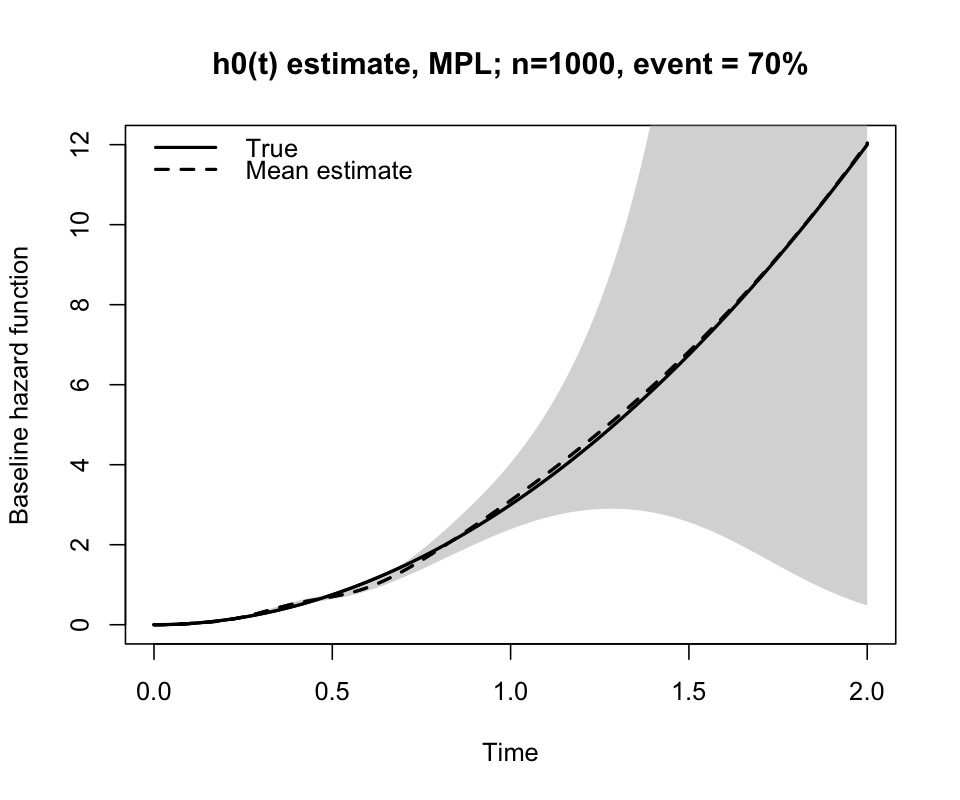}
    \includegraphics[width=0.49\textwidth]{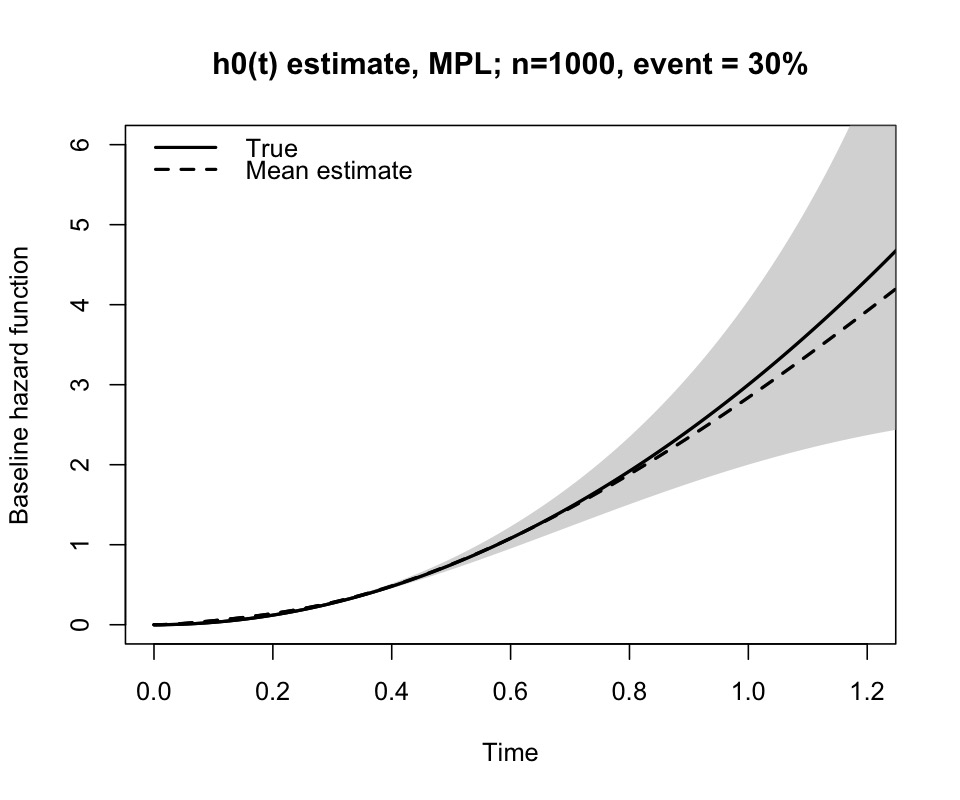}
    \caption{Estimates of the baseline hazard function for $\sigma_{\varepsilon} = 0.2$ and $\bar{n}_i = 20$. The solid line is the true baseline hazard function, the dashed line is the mean estimate and the grey area represents the asymptotic $95\%$ coverage probability.}
    \label{fig:right_cens_h0t_4}
\end{figure}

\begin{table}[]
    \scriptsize
    \centering
    \begin{tabular}{cc|cc|cc|cc|cc}
         \hline
         & & \multicolumn{4}{c}{True $\sigma_{\varepsilon} = 0.05$; event proportion $70\%$} & \multicolumn{4}{c}{True $\sigma_{\varepsilon} = 0.05$; event proportion $30\%$}\\
         \hline
         & $(n, \bar{n}_i)$ & $(200, 5)$ & $(200, 20)$ & $(1000, 5)$ & $(1000, 20)$ & $(200, 5)$ & $(200, 20)$ & $(1000, 5)$ & $(1000, 20)$\\
         \hline
         Bias($\sigma_{\varepsilon}$) & MPL & 0.001 & 0.001 & 0.001 & 0.001 & 0.001 & 0.001 & 0.001 & 0.001\\
         & JM & 0.067 & 0.097 & 0.102 & 0.097 & 0.055 & 0.076 & 0.083 & 0.076\\
         Bias($\sigma_{\kappa_1}$) & MPL & -0.005 & -0.003 & -0.002 & 0.001 & -0.005 & -0.003 & -0.003 & -0.002\\
         & JM & 0.006 & 0.027 & 0.012 & 0.029 & 0.008 & 0.023 & 0.012 & 0.024\\
         Bias($\sigma_{\kappa_2}$) & MPL & -0.018 & -0.005 & -0.020 & -0.006 & -0.030 & -0.006 & -0.037 & -0.007\\
         & JM & 0.011 & 0.079 & 0.023 & 0.089 & 0.025 & 0.134 & 0.062 & 0.158\\
         MISE($h_0(t)$) & MPL & 0.280 & 0.275 & 0.266 & 0.408 & 0.353 & 1.213 & 0.400 & 0.406 \\
         & JM & 0.335 & 0.333 & 0.252 & 0.396 & 1.721 & 6.829 & 2.722 & 3.039 \\
         \hline
         & & \multicolumn{4}{c}{True $\sigma_{\varepsilon} = 0.2$; event proportion $70\%$} & \multicolumn{4}{c}{True $\sigma_{\varepsilon} = 0.2$; event proportion $30\%$}\\
         \hline
         & $(n, \bar{n}_i)$ & $(200, 5)$ & $(200, 20)$ & $(1000, 5)$ & $(1000, 20)$ & $(200, 5)$ & $(200, 20)$ & $(1000, 5)$ & $(1000, 20)$\\
         \hline
         Bias($\sigma_{\varepsilon}$) & MPL & 0.004 & 0.001 & 0.003 & 0.001 & 0.013 & 0.001 & 0.011 & 0.001 \\
         & JM & 0.009 & 0.035 & 0.029 & 0.037 & 0.006 & 0.024 & 0.019 & 0.025\\
         Bias($\sigma_{\kappa_1}$) & MPL & -0.021 & -0.011 & -0.021 & -0.009 & -0.014 & -0.009 & -0.012 & -0.007\\
         & JM & 0.004 & 0.026 & 0.007 & 0.021 & 0.005 & 0.022 & 0.011 & 0.022 \\
         Bias($\sigma_{\kappa_2}$) & MPL & -0.193 & -0.044 & -0.181 & -0.040 & -0.391 & -0.081 & -0.368 & -0.073\\
         & JM & -0.001 & 0.067 & 0.015 & 0.072 & 0.009 & 0.109 & 0.051 & 0.125\\
         MISE($h_0(t)$) & MPL & 0.189 & 0.224 & 0.181 & 0.089 & 0.180 & 0.281 & 0.296 & 0.310 \\
         & JM & 0.885 & 0.769 & 0.121 & 0.123 & 1.378 & 1.208 & 2.780 & 2.468 \\
         \hline
    \end{tabular}
    \caption{Additional simulation results, for variance components (standard error of random effects and measurement error distributions) and baseline hazard function from MPL and JM methods for right censored data (Study 1).}
    \label{tab:var_comp_stdy1}
\end{table}

\subsection{Study 2a (interval censoring)}

Like the results from Study 1, the MPL estimates of the standard deviation 
of the measurement error distribution, $\sigma_{\varepsilon}$, consistently have very small bias regardless of sample size, number of longitudinal observations, and censoring proportions, and this bias is consistently smaller for the MPL method than for the midpoint imputation method (Table \ref{tab:var_comp_stdy2}). The bias in the estimates of the variance of the random effects distributions for the MPL method is also small and decreases with larger numbers of observations of the longitudinal trajectory.

\begin{table}[]
    \scriptsize
    \centering
    \begin{tabular}{cc|cc|cc}
         \hline
         & & \multicolumn{4}{c}{True $\sigma_{\varepsilon} = 0.1$; right proportion $30\%$} \\
         \hline
         & $(n, \bar{n}_i)$ & $(200, 5)$ & $(200, 20)$ & $(1000, 5)$ & $(1000, 20)$ \\
         \hline
         Bias($\sigma_{\varepsilon}$) & MPL & 0.001 & 0.001 & 0.001 & -0.001\\
         & JM-mid & 0.004 & 0.018 & 0.016 & 0.020  \\
         Bias($\sigma_{\kappa_1}$) & MPL & -0.020 & -0.006 & -0.017 & -0.005  \\
         & JM-mid & 0.001 & 0.021 & 0.004 & 0.018 \\
         Bias($\sigma_{\kappa_2}$) & MPL & -0.035 & -0.011 & -0.031 & -0.007  \\
         & JM-mid & -0.001 & 0.025 & 0.002 & 0.020 \\
         MISE($h_0(t)$) & MPL & 0.019 & 0.065 & 0.030 & 0.002\\
         & JM-mid & 0.152 & 0.101 & 0.124 & 0.129  \\
         \hline
         \hline
    \end{tabular}
    \caption{Additional simulation results, for variance components (standard error of random effects and measurement error distributions) and baseline hazard function from MPL and JM-midpoint imputation methods for interval censored data (Study 2a).}
    \label{tab:var_comp_stdy2}
\end{table}

\subsection{Study 2c (interval censoring)}

\begin{table}[h!]
    \centering
    \begin{tabular}{l | ccc | l | cc  }
         \hline
         & \multicolumn{6}{c}{$n = 1000$, $\Bar{n}_i = 5$, $\pi^I = 0.40$, $\pi^R = 0.25$} \\
         \hline
         & \multicolumn{2}{c}{Cox Regression parameters} & & & \multicolumn{2}{c}{Longitudinal trajectories} \\
         \hline
         & $\beta_1$ & $\beta_2$ & $\gamma$ &   & $\mu(t)$ & $z_i(t)$ \\
         \hline
         Bias & -0.004 & -0.033 & 0.088 & MISE & 0.00002 & 0.02070\\
         SE &  0.043 & 0.085 & 0.390 & & \\
         & (0.045) & (0.092) & (0.563)  & & \\
         CP & 0.92 &  0.89 & 0.75  & & \\
         & (0.92) & (0.93) & (0.88)  & & \\
         \hline
         & \multicolumn{6}{c}{$n = 1000$, $\Bar{n}_i = 20$, $\pi^I = 0.40$, $\pi^R = 0.25$} \\
         \hline
         & \multicolumn{2}{c}{Cox Regression parameters} & & & \multicolumn{2}{c}{Longitudinal trajectories} \\
         \hline
         & $\beta_1$ & $\beta_2$ & $\gamma$ &   & $\mu(t)$ & $z_i(t)$ \\
         \hline
         Bias & -0.002 & -0.038 & 0.029 & MISE & 0.00002 & 0.00073  \\
         SE & 0.043 & 0.085 & 0.386  &  & \\
         & (0.047) & (0.089) & (0.555) &  & \\
         CP & 0.95 & 0.92 & 0.76 &  &  \\
         & (0.96) & (0.93) & (0.90) &  & \\
         \hline
    \end{tabular}
    \caption{Regression parameter simulation results for Study 2c; $n$ refers to the total sample size, $\bar{n}_i$ refers to the average number of longitudinal observations, $\pi^I$ refers to the proportion of interval-censored observations and $\pi^R$ refers to the proportion of right censored observations.}
    \label{study2b_results}
\end{table}

\end{document}